\PassOptionsToPackage{unicode}{hyperref}
\PassOptionsToPackage{hyphens}{url}
\PassOptionsToPackage{dvipsnames,svgnames,x11names}{xcolor}
\documentclass[
  12pt]{article}

\usepackage{amsmath,amssymb}
\usepackage{iftex}
\ifPDFTeX
  \usepackage[T1]{fontenc}
  \usepackage[utf8]{inputenc}
  \usepackage{textcomp} 
\else 
  \usepackage{unicode-math}
  \defaultfontfeatures{Scale=MatchLowercase}
  \defaultfontfeatures[\rmfamily]{Ligatures=TeX,Scale=1}
\fi
\usepackage{lmodern}
\ifPDFTeX\else  
\fi
\IfFileExists{upquote.sty}{\usepackage{upquote}}{}
\IfFileExists{microtype.sty}{
  \usepackage[]{microtype}
  \UseMicrotypeSet[protrusion]{basicmath} 
}{}
\makeatletter
\@ifundefined{KOMAClassName}{
  \IfFileExists{parskip.sty}{%
    \usepackage{parskip}
  }{
    \setlength{\parindent}{0pt}
    \setlength{\parskip}{6pt plus 2pt minus 1pt}}
}{
  \KOMAoptions{parskip=half}}
\makeatother
\usepackage{xcolor}
\makeatletter
\ifx\paragraph\undefined\else
  \let\oldparagraph\paragraph
  \renewcommand{\paragraph}{
    \@ifstar
      \xxxParagraphStar
      \xxxParagraphNoStar
  }
  \newcommand{\xxxParagraphStar}[1]{\oldparagraph*{#1}\mbox{}}
  \newcommand{\xxxParagraphNoStar}[1]{\oldparagraph{#1}\mbox{}}
\fi
\ifx\subparagraph\undefined\else
  \let\oldsubparagraph\subparagraph
  \renewcommand{\subparagraph}{
    \@ifstar
      \xxxSubParagraphStar
      \xxxSubParagraphNoStar
  }
  \newcommand{\xxxSubParagraphStar}[1]{\oldsubparagraph*{#1}\mbox{}}
  \newcommand{\xxxSubParagraphNoStar}[1]{\oldsubparagraph{#1}\mbox{}}
\fi
\makeatother

\usepackage{longtable,booktabs,array}
\usepackage{calc} 
\usepackage{etoolbox}
\makeatletter
\patchcmd\longtable{\par}{\if@noskipsec\mbox{}\fi\par}{}{}
\makeatother
\IfFileExists{footnotehyper.sty}{\usepackage{footnotehyper}}{\usepackage{footnote}}
\makesavenoteenv{longtable}
\usepackage{graphicx}
\makeatletter
\def\maxwidth{\ifdim\Gin@nat@width>\linewidth\linewidth\else\Gin@nat@width\fi}
\def\maxheight{\ifdim\Gin@nat@height>\textheight\textheight\else\Gin@nat@height\fi}
\makeatother
\setkeys{Gin}{width=\maxwidth,height=\maxheight,keepaspectratio}
\makeatletter
\def\fps@figure{htbp}
\makeatother

\makeatletter
\@ifpackageloaded{caption}{}{\usepackage{caption}}
\AtBeginDocument{%
\ifdefined\contentsname
  \renewcommand*\contentsname{Table of contents}
\else
  \newcommand\contentsname{Table of contents}
\fi
\ifdefined\listfigurename
  \renewcommand*\listfigurename{List of Figures}
\else
  \newcommand\listfigurename{List of Figures}
\fi
\ifdefined\listtablename
  \renewcommand*\listtablename{List of Tables}
\else
  \newcommand\listtablename{List of Tables}
\fi
\ifdefined\figurename
  \renewcommand*\figurename{Figure}
\else
  \newcommand\figurename{Figure}
\fi
\ifdefined\tablename
  \renewcommand*\tablename{Table}
\else
  \newcommand\tablename{Table}
\fi
}
\@ifpackageloaded{float}{}{\usepackage{float}}
\floatstyle{ruled}
\@ifundefined{c@chapter}{\newfloat{codelisting}{h}{lop}}{\newfloat{codelisting}{h}{lop}[chapter]}
\floatname{codelisting}{Listing}

\makeatother
\makeatletter
\@ifpackageloaded{caption}{}{\usepackage{caption}}
\@ifpackageloaded{subcaption}{}{\usepackage{subcaption}}
\makeatother

\ifLuaTeX
  \usepackage{selnolig}  
\fi
\usepackage{natbib}
\usepackage{bibunits}
\defaultbibliographystyle{agsm}
\defaultbibliography{ref} 
\usepackage{bookmark}

\IfFileExists{xurl.sty}{\usepackage{xurl}}{} 
\hypersetup{
  pdftitle={Title},
  pdfauthor={Author 1; Author 2},
  pdfkeywords={3 to 6 keywords, that do not appear in the title},
  colorlinks=true,
  linkcolor={blue},
  filecolor={Maroon},
  citecolor={Blue},
  urlcolor={Blue},
  pdfcreator={LaTeX via pandoc}}

\usepackage{subcaption}
\usepackage{graphicx,color}
\usepackage{amsmath}
\usepackage{amsthm}
\usepackage{wasysym}
\usepackage{url}
\usepackage{amssymb}
\usepackage{multirow}
\usepackage{bm}
\usepackage{booktabs}
\usepackage{mathrsfs}
\usepackage{fullpage}
\usepackage{setspace}
\usepackage[english]{babel}
\usepackage[utf8]{inputenc}
\usepackage{booktabs}
\usepackage{diagbox}
\usepackage{multirow}
\usepackage{hyperref}
\usepackage{dsfont}
\usepackage{subcaption}

\newtheorem{theorem}{Theorem}
\newtheorem{proposition}{Proposition}

\newtheorem{lemma}{Lemma}

\newtheorem{definition}{Definition}

\newcommand{\Var}{\text{Var}}

\newcommand{\ba}{\bm{a}}

\newcommand{\bA}{\bm{A}}
\newcommand{\bB}{\bm{B}}
\newcommand{\bC}{\bm{C}}

\newcommand{\bF}{\bm{F}}

\newcommand{\bH}{\bm{H}}
\newcommand{\bI}{\bm{I}}

\newcommand{\bX}{\bm{X}}

\newcommand{\bZ}{\bm{Z}}
\newcommand{\balpha}{\bm{\alpha}}
\newcommand{\bbeta}{\bm{\beta}}
\newcommand{\bxi}{\bm{\xi}}

\newcommand{\bomega}{\bm{\bomega}}
\newcommand{\bTheta}{\bm{\Theta}}
\newcommand{\bGamma}{\bm{\Gamma}}

\newcommand{\RR}{\mathbb{R}}
\newcommand{\PP}{\mathbb{P}}
\newcommand{\NN}{\mathbb{N}}

\newcommand{\EE}{\mathbb{E}}

\DeclareMathOperator{\ind}{\mathds{1}}

\newcommand{\anon}{1}

\begin{document}

\def\spacingset#1{\renewcommand{\baselinestretch}%
{#1}\small\normalsize} \spacingset{1}


\if1\anon
{
  \title{\bf A Probabilistic Model for \\Zero-Inflated Count Tensors with \\Structured Latent Representations}
  \author{Elena Tuzhilina\thanks{
Elena Tuzhilina was supported by the Natural Sciences and Engineering Research Council of Canada under Grant RGPIN-2023-04727, the University of Toronto Data Sciences Institute Catalyst Grant, and the University of Toronto McLaughlin Centre Accelerator Grant MC-2023-05. Yaoming Zhen was supported by the CUHK-Shenzhen Starting-up Grant UDF01004232.}\hspace{.2cm}\\
    Department of Statistical Sciences, University of Toronto\\
    and \\
    Yaoming Zhen\\
  School of Data Science, The Chinese University of Hong Kong, Shenzhen}
  \maketitle
} \fi

\if0\anon
{
  \title{\bf A Probabilistic Model for\\ Zero-Inflated Count Tensors with\\ 
  Structured Latent Representations}
  \maketitle
  \medskip
} \fi

\bigskip
\begin{abstract}
We propose a unified probabilistic framework for modeling high-dimensional count
tensors with excess zeros. Such data arise naturally in a variety of applications,
including single-cell Hi-C experiments, where observations are represented as a
third-order tensor indexed by genomic locus pairs and cells. We develop a
zero-inflated Poisson tensor model that captures the underlying latent structure
through a low-rank tensor decomposition. For single-cell Hi-C data, the proposed
framework further accommodates heterogeneous cell populations through latent cluster
structure and exploits the ordered nature of genomic loci via smooth latent
representations. We develop a likelihood-based estimation procedure together with a
Bayes-optimal classifier for distinguishing structural zeros from technical zeros,
enabling principled false-zero detection and imputation.
Theoretically, we establish identifiability of the proposed model and consistency of
the proposed estimators. Simulation studies and analyses of single-cell Hi-C data
demonstrate improved performance in false-zero detection, latent structure recovery,
and clustering.
\end{abstract}

\noindent%
{\it Keywords:} Count Tensor, False Zero Detection, Imputation, Structured Tensor Decomposition, Zero Inflation
\vfill

\newpage
\spacingset{1.8} 

\begin{bibunit}

\section{Introduction}

High-dimensional count data with excess zeros arise in modern applications,
including single-cell genomics \citep{Pierson2015gen, Nagano2017sc}, ecology
\citep{Martin2005eco, Zuur2009eco}, and recommender systems with implicit feedback
\citep{Hu2008recsys}. Many of these data are naturally represented as multiway count arrays, or tensors. Excess zeros are prevalent in such data and may arise from structural absence, technical limitations, detection failure, or measurement noise \citep{Hicks2018zeros}. Distinguishing structural from technical zeros is essential for recovering latent structure and enabling reliable downstream inference. Zero-inflated and hurdle models are widely used to explicitly
distinguish structural zeros from technical zeros \citep{Mullahy1986zi, Lambert1992zi,
Cameron2013zi}. Beyond zero inflation, tensor-valued data often exhibit additional
structural characteristics, such as heterogeneity across tensor modes, that can be
exploited to improve statistical efficiency. However, existing statistical methods do
not provide a probabilistic framework for inference on zero-inflated tensor data while
simultaneously incorporating such structural information.

One important example is chromatin contact data obtained via Hi-C technology
\citep{Lieberman2009, Rao2014}. Hi-C is a genomic assay that measures the frequency
of spatial proximity between pairs of genomic loci and is widely used to reconstruct
three-dimensional (3D) chromatin organization \citep{Duan2010, Hu2013reconstruction,
tuzhilina2022principal, tuzhilina2024conformation}, with broad applications to gene
regulation, chromosome architecture, and disease studies. A Hi-C experiment produces a symmetric contact matrix whose entries record the number of observed interactions between pairs of genomic loci, with larger counts indicating closer spatial proximity. However, contact
matrices are highly sparse due to both biological absence of contacts and technical
limitations such as low sequencing depth and incomplete capture efficiency.
Disentangling \textit{structural zeros} (true absence of contact) from \textit{technical zeros}
(dropout events) is essential for downstream tasks such as 3D reconstruction. Failure
to distinguish these two sources can bias downstream inference, while treating all
zeros as missing values discards informative structural signal.

In the literature, substantial work has focused on bulk Hi-C data, which aggregate signals across many cells into a single contact matrix whose entries represent total interaction frequencies. More recently, single-cell Hi-C technologies have enabled measurements at the level of individual cells, yielding a collection of contact matrices that capture cell-to-cell variability in chromatin organization \citep{Nagano2017sc, Stevens2017sc, Ramani2017sc}. These data naturally form a third-order tensor indexed by genomic locus pairs and cells, which we refer to as a Hi-C tensor. Modeling the data as a tensor enables information to be shared across genomic loci and cells while preserving cell-to-cell heterogeneity. A Hi-C tensor naturally resembles a multi-layer network but differs in one important aspect: Genomic loci are naturally ordered along chromosomes, so neighboring loci tend to exhibit similar interaction patterns due to the continuous nature of chromatin. By contrast, the ordering of nodes in conventional multi-layer networks is typically arbitrary and carries no intrinsic information. An appropriate statistical model should therefore account for both heterogeneous cell populations and the ordered structure of genomic loci.

Several methods have recently been proposed to address zero inflation and imputation
in single-cell Hi-C data. For example, HiCImpute \citep{Xie2022hicimpute}
distinguishes structural zeros from technical dropouts via a Bayesian hierarchical
model, scHiCSRS \citep{Xu2022schicsrs} combines self-representation smoothing with
mixture modeling for zero identification, and Higashi \citep{Zhang2022higashi2}
employs hypergraph neural networks to learn cell embeddings and impute contact maps
by leveraging correlations across cells. Although these methods effectively borrow information across cells, they operate on collections of contact matrices rather than directly modeling the data as a count tensor. Consequently, they do not provide a probabilistic tensor framework that simultaneously accounts for zero inflation, heterogeneous cell populations, and the ordered structure of genomic loci.

In this paper, we propose \texttt{ZITS} (Zero-Inflated Tensor model with Structured latent representations), a probabilistic framework for zero-inflated count tensors motivated by single-cell Hi-C data. Specifically, we employ a zero-inflated Poisson tensor model in which the intensity and masking probability tensors share a coupled low-rank decomposition. For single-cell Hi-C, we model heterogeneous cell populations through latent clusters and exploit the ordered nature of genomic loci through smooth latent representations based on basis expansions, similar to functional tensor decomposition \citep{han2024guaranteed}. The resulting likelihood-based framework enables principled inference for zero-inflated tensor data, including false-zero detection and imputation, while supporting downstream analyses such as cell clustering and 3D chromatin reconstruction.

The main contributions of this paper are fourfold. First, we propose a probabilistic tensor model for zero-inflated count data that provides a unified likelihood-based framework for statistical inference with structured latent representations. Although motivated by single-cell Hi-C data, the framework readily accommodates application-specific structural assumptions. Second, we develop a likelihood-based estimation procedure together with a Bayes-optimal classifier for distinguishing structural and technical zeros, enabling principled false-zero detection and imputation. Third, for single-cell Hi-C data, we incorporate latent cell clusters and smooth chromatin embeddings to exploit the underlying biological structure and improve statistical efficiency. Finally, we establish identifiability of the proposed model and prove consistency of the proposed estimators.

The remainder of the paper is organized as follows. Section~\ref{sec: model} introduces the proposed model and estimation procedure. Section~\ref{sec: theory} establishes the theoretical guarantees and develops a Bayes-optimal approach to false-zero detection and imputation. Section~\ref{sec: simulation} presents simulation studies, and Section~\ref{sec: hic} analyzes a single-cell Hi-C dataset. Technical proofs and auxiliary results are deferred to the Supplementary.

\section{Zero-inflated Tensor  with Structured Representation}
\label{sec: model}

\subsection{Notation}
\label{sec: notations}
The following notations are used throughout the paper. We use lowercase letter $(a)$, boldface lowercase letter $(\ba)$, boldface capital letter $(\bA)$, and boldface Euler script letter $(\bm{\mathcal{A}})$ to represent scalar, vector, matrix and higher-order tensor (order-three or above), respectively. Let $[n] = \{1, 2, \ldots, n\}$ be the $n$-set, for any positive integer $n$. We use $\bm{0}_n$ and $\bm{1}_n$ to denote the $n$-dimensional vectors of all zeros and all ones. For any matrix $\bA$ and $\bB$ of conformable dimensions, their Kronecker product, Khatri–Rao product (matching column-wise Kronecker product) and Hadamard product are denoted as $\bA \otimes \bB$, $\bA \odot \bB$ and ${\bA * \bB}$, respectively. For convenience, we denote $\bA \odot^\top \bB = (\bA^\top \odot \bB^\top)^\top$ as the matching row-wise Kronecker product of $\bA$ and $\bB$. The $(i,j)$-th entry of matrix $\bA$ is denoted by $A_{i,j}$. Similarly,  for a third order tensor $\bm{\mathcal{T}}$, its $(i,j,k)$-th entry is denoted by $\mathcal{T}_{i, j, k}$.  The identity tensor of order three and dimension $L$ is denoted by $\bm{\mathcal{I}}_L=(I_{l_1,l_2,l_3})\in\{0,1\}^{L\times L\times L}$, where $I_{l_1,l_2,l_3}=1$ if and only if $l_1=l_2=l_3$, and $I_{l_1,l_2,l_3}=0$ otherwise.

Given a tensor $\bm{\mathcal T} \in \mathbb{R}^{N_1\times N_2 \times N_3}$ and a matrix $\bA \in \mathbb{R}^{M_1 \times N_1}$, the mode-$1$ product between $\bm{\mathcal{T}}$ and $\bA$, denoted by $\bm{\mathcal{T}}\times_1 \bA \in \RR^{M_1 \times N_2 \times N_3}$, is a tensor with entries
$$(\bm{\mathcal{T}} \times_1 \bA)_{m_1, n_2, n_3} = \sum_{n_1 = 1}^{N_1}\mathcal{T}_{n_1, n_2, n_3} A_{m_1, n_1}.$$ 
For $\bB \in \mathbb{R}^{M_2 \times N_2}$, and $\bC \in \mathbb{R}^{M_3 \times N_3}$,
the mode-$2$ product $\bm{\mathcal{T}}\times_2 \bB$ and mode-$3$ product $\bm{\mathcal{T}}\times_3 \bC$ are defined in a similar fashion. Moreover, these multi-linear products are associative, leading to the multi-linear product ${\bm{\mathcal{T}}\times_1 \bA\times_2 \bB \times_3 \bC \in \RR^{M_1 \times M_2 \times M_3}}$, whose entries read 
\begin{align*}
(\bm{\mathcal{T}} \times_1 \bA\times_2 \bB \times_3 \bC)_{m_1, m_2, m_3} = \sum_{n_1 = 1}^{N_1}\sum_{n_2 = 1}^{N_2}\sum_{n_3 = 1}^{N_3}\mathcal{T}_{n_1, n_2, n_3} A_{m_1, n_1}B_{m_2, n_2}C_{m_3, n_3}. 
\end{align*}

\subsection{Model}
\label{sec: statmodel}

The proposed framework is developed in three stages. We first model excess zeros using a zero-inflated Poisson distribution. We then introduce a coupled low-rank tensor decomposition for the intensity and masking probability tensors. Finally, we show how additional structural assumptions, such as latent clustering and smoothness, can be incorporated through the latent representations. Although the proposed framework applies to general zero-inflated count tensors, we develop the methodology in the context of single-cell Hi-C data, which motivates the structural assumptions introduced later.

Let $\bm{\mathcal{C}}=(C_{i,j,k})\in\mathbb{R}^{N\times N\times K}$ denote a Hi-C contact tensor on $N$ genomic loci across $K$ cells. Specifically, $C_{i,j,k}$ records the observed contact count between genomic loci $i$ and $j$ in cell $k$, for $i,j\in[N]$ and $k\in[K]$. In practice, limited sequencing depth and measurement errors cause many true chromatin interactions to be masked by zeros, resulting in a large number of false zeros and an extremely sparse observed Hi-C tensor \citep{Xie2022hicimpute}. 

We begin by modeling the excess zeros using a zero-inflated Poisson (ZIP) distribution. Specifically, for each $i<j\in[N]$ and $k\in[K]$, we assume that the contact count $C_{i,j,k}$ follows
\begin{equation}
\label{equ: zero_inflated}
\mathbb{P}\left(C_{i,j,k}=C_{j,i,k}=c\right)=
\begin{cases}
p_{i,j,k}+(1-p_{i,j,k})e^{-\lambda_{i,j,k}}, & c=0,\\
(1-p_{i,j,k})\lambda_{i,j,k}^ce^{-\lambda_{i,j,k}}/c!, & c\in\mathbb{N}^+,
\end{cases}
\end{equation}
where $\bm{\Lambda}=(\lambda_{i,j,k})\in\mathbb{R}^{N\times N\times K}$ denotes the Poisson intensity tensor, and $\bm{\mathcal P}=(p_{i,j,k})\in\mathbb{R}^{N\times N\times K}$ denotes the masking probability tensor.

The ZIP model admits a simple latent-variable interpretation. Let $\widetilde C_{i,j,k}\sim\mathrm{Poisson}(\lambda_{i,j,k})$ denote the latent contact count and
$B_{i,j,k}\sim\mathrm{Bernoulli}(1-p_{i,j,k})$
an independent masking variable. The observed count is
$
C_{i,j,k}=B_{i,j,k}\widetilde C_{i,j,k},
$
so that observed zeros arise either from the Poisson distribution (structural zeros) or from the masking process (false zeros). The masking probabilities $p_{i,j,k}$ determine the overall sparsity of the observed Hi-C tensor, with larger values corresponding to a greater proportion of false zeros. Throughout the paper, we assume $0\le p_{i,j,k}<1$ and $\lambda_{i,j,k}>0$ to avoid degenerate cases.

Second, we reparameterize the Poisson intensity and masking probability tensors using canonical link functions. Let
$\bm{\eta}=(\eta_{i,j,k})\in\mathbb{R}^{N\times N\times K}$ and
$\bTheta=(\theta_{i,j,k})\in\mathbb{R}^{N\times N\times K}$
denote the entrywise logarithm of $\bm{\Lambda}$ and the entrywise logit transformation of $\bm{\mathcal J}-\bm{\mathcal P}$, respectively, where $\bm{\mathcal J}$ is the all-ones tensor. That is, 
$\eta_{i,j,k}=\log\lambda_{i,j,k}$ and $
\theta_{i,j,k}=\log\frac{1-p_{i,j,k}}{p_{i,j,k}}$ for all $i,j\in[N]$ and $k\in[K]$.
 This reparameterization transforms $\lambda_{i,j,k}>0$ and $0\le p_{i,j,k}<1$ into unconstrained parameters, simplifying likelihood-based inference. Moreover, the logit transformation is the canonical link for Bernoulli models, including logistic regression and latent-variable models for binary networks \citep{zhang2020flexible,JMLR:v21:17-470,lyu2023latent}.

Third, to capture the low-dimensional latent structure underlying the Hi-C tensor, we model both $\bm{\eta}$ and $\bTheta$ through the following coupled low-rank CP decompositions:
\begin{align}
\bm{\eta}
=
\bm{\mathcal I}_{L}
\times_1\balpha
\times_2\balpha
\times_3\widetilde{\bbeta}
\qquad\text{and}\qquad
\bTheta
=
\bm{\mathcal I}_{L}
\times_1\balpha
\times_2\balpha
\times_3\widetilde{\bxi},
\label{equ: CP_theta}
\end{align}
where $L$ denotes the CP rank, $\balpha\in\mathbb{R}^{N\times L}$, and $\widetilde{\bbeta},\widetilde{\bxi}\in\mathbb{R}^{K\times L}$ are the corresponding factor matrices. The proposed decomposition admits a natural interpretation. The rows of $\balpha$ represent latent embeddings of the genomic loci, whereas the rows of $\widetilde{\bbeta}$ and $\widetilde{\bxi}$ represent latent embeddings of the cells. The genomic embedding matrix $\balpha$ is shared by both decompositions under the assumption that the interaction intensities and masking probabilities arise from the same underlying chromatin organization.

On the flip side, chromatin organization exhibits substantial cell-to-cell heterogeneity. Although cells from different tissues or cell types often display distinct chromatin organizations, considerable variability is also observed among cells within the same population. To capture cell heterogeneity, we impose additional structure on the cell embedding matrices $\widetilde{\bbeta}$ and $\widetilde{\bxi}$. Suppose the $K$ cells belong to $R$ latent clusters, and let $\bZ=(Z_{k,r})\in\{0,1\}^{K\times R}$ denote the cluster membership matrix, where $Z_{k,r}=1$ if cell $k$ belongs to cluster $r$ and $Z_{k,r}=0$ otherwise. Assuming each cell belongs to exactly one cluster, we model the cell embeddings as
$\widetilde{\bbeta}=\bZ\bbeta$ and $
\widetilde{\bxi}=\bZ\bxi,$
where $\bbeta,\bxi\in\mathbb{R}^{R\times L}$. Thus, cells within the same cluster share common latent embeddings while different clusters remain distinct.

The proposed formulation admits several useful extensions. For example, cluster-specific genomic embeddings can be obtained by imposing a block-diagonal structure on $\bbeta$. Specifically, suppose $L=RL_0$ for some integer $L_0$, partition
$\balpha=
\left[
\balpha^{(1)},
\ldots,
\balpha^{(R)}
\right]$,
where $\balpha^{(r)}\in\mathbb{R}^{N\times L_0}$, and assume $\bbeta_{r,l}$ may be nonzero only when $(r-1)L_0<l\le rL_0$. Then $\balpha^{(r)}$ represents the genomic embedding for cluster $r$, while
$\balpha_{i,\cdot}
=
\left(
\balpha_{i,\cdot}^{(1)\top},
\ldots,
\balpha_{i,\cdot}^{(R)\top}
\right)^\top$
is the corresponding global embedding of genomic locus $i$. This construction resembles the mixture multilayer stochastic block model of \citet{jing2021community}.
The homogeneity assumption may also be relaxed by replacing
$\widetilde{\bbeta}=\bZ\bbeta$ and
$\widetilde{\bxi}=\bZ\bxi$
with
$\bZ\odot^\top\bbeta^{(\mathrm{heter})}$ and
$\bZ\odot^\top\bxi^{(\mathrm{heter})}$,
where $\bbeta^{(\mathrm{heter})},\bxi^{(\mathrm{heter})}\in\mathbb{R}^{K\times L}$. This allows cells within the same cluster to have distinct embeddings while the clustering structure is preserved by the sparsity pattern.

Finally, unlike most tensor data, the first two modes of a Hi-C tensor correspond to genomic loci, which are naturally ordered along the chromosome. Consequently, Hi-C data are more naturally viewed as a functional tensor \citep{dolgov2021functional, luo2023low, ni2023learning, han2024guaranteed}. Because neighboring genomic loci tend to exhibit similar interaction patterns, we assume that the latent genomic embeddings vary smoothly along the chromosome and model them using the functional representation of \citet{tuzhilina2022principal,tuzhilina2024conformation}.

To enforce this smoothness constraint, we employ a basis expansion. Let $\{h_q(\cdot)\}_{q=1}^Q$ be a collection of orthonormal basis functions, and let $\bH=(H_{i,q})\in\mathbb{R}^{N\times Q}$ denote the corresponding basis matrix, where
$
H_{i,q}=h_q(i)$ for $i\in[N],\ q\in[Q].$
We model the genomic embedding matrix as
$\balpha=\bH\bGamma,$
where $\bGamma\in\mathbb{R}^{Q\times L}$ is the coefficient matrix. This formulation naturally accommodates missing genomic bins by evaluating the basis functions only at the observed loci. Substituting the constraints into (\ref{equ: CP_theta}) yields
\begin{align}
\label{equ: final_CP}
\bm{\eta}
&=
\bm{\mathcal I}
\times_1
\bH\bGamma
\times_2
\bH\bGamma
\times_3
\bZ\bbeta,
\quad\text{and}\quad
\bTheta
=
\bm{\mathcal I}
\times_1
\bH\bGamma
\times_2
\bH\bGamma
\times_3
\bZ\bxi.
\end{align}
The resulting model is parameterized by $\bGamma$, $\bbeta$, $\bxi$, and $\bZ$.

The proposed framework combines four key components: A zero-inflated Poisson model to distinguish structural zeros from technical dropouts, coupled low-rank tensor decompositions to model interaction intensities and masking probabilities, latent cell clusters to capture heterogeneous cell populations while borrowing strength across similar cells, and smooth genomic embeddings through a basis expansion. We refer to the resulting framework as the \emph{Zero-Inflated Tensor model with Structured Latent Representation} (\texttt{ZITS}).

\subsection{Estimation}
\label{sec: estimation}

Parameter estimation under the proposed \texttt{ZITS} model proceeds in two stages. We first estimate the latent tensor representation by maximizing the likelihood, and subsequently recover the latent cell clusters from the estimated cell embeddings.

Under the \texttt{ZITS} model, the negative log-likelihood of the Hi-C tensor is
\begin{align}
\mathcal{L}\left(\bm{\eta},\bTheta; \bm{\mathcal{C}}\right) = & \frac{2}{N(N+1)K} \sum_{i\le j \in [N], k \in[K]} \left[\ind(C_{i, j, k} \ne 0) \left( \log \left( 1 + \exp\left(e^{\eta_{i, j, k}} - \theta_{i, j, k}\right)\right) - C_{i, j, k} \eta_{i, j, k}\right) \right.\nonumber\\
& \left. + \log (1+e^{-\theta_{i, j, k}}) + e^{\eta_{i, j, k}} - \log \left(1 + \exp\left(e^{\eta_{i, j, k}} - \theta_{i, j, k}\right)\right)\right],
\label{equ: nll}
\end{align}
where $\bm{\eta}$ and $\bTheta$ are parameterized according to the coupled low-rank decomposition in \eqref{equ: final_CP}. A detailed derivation of the negative log-likelihood is provided in Supplementary~\ref{sec: likelihood}.

We minimize the negative log-likelihood by gradient descent over the low-rank factors $\bGamma$, $\widetilde{\bbeta}$, and $\widetilde{\bxi}$. Specifically, at iteration $t$, we first compute $\bm{\eta}^{(t-1)}$ and $\bTheta^{(t-1)}$ from the current parameter estimates and then update
\[
\bX^{(t)}
=
\bX^{(t-1)}
-
\nu
\nabla_{\bX}
\mathcal{L}
\left(
\bm{\eta}^{(t-1)},
\bTheta^{(t-1)};
\bm{\mathcal C}
\right),
\qquad
\bX\in\left\{\bGamma,\widetilde{\bbeta},\widetilde{\bxi}\right\},
\]
where the step size $\nu$ is selected by line search. Derivations of the gradients are provided in Supplementary~\ref{appendix: derivative_homogeneous}. The algorithm is terminated when either the maximum number of iterations is reached or
\[
\max \left\{
\frac{\|\bGamma^{(t)}-\bGamma^{(t-1)}\|_F}{\|\bGamma^{(t-1)}\|_F},
\frac{\|\widetilde{\bbeta}^{(t)}-\widetilde{\bbeta}^{(t-1)}\|_F}{\|\widetilde{\bbeta}^{(t-1)}\|_F},
\frac{\|\widetilde{\bxi}^{(t)}-\widetilde{\bxi}^{(t-1)}\|_F}{\|\widetilde{\bxi}^{(t-1)}\|_F}
\right\}
\]
falls below a pre-specified threshold.

Let $\widetilde{\bbeta}^{(\mathrm{grad})}$ and $\widetilde{\bxi}^{(\mathrm{grad})}$ denote the cell embeddings obtained upon convergence of the gradient descent algorithm. The cluster membership matrix $\bZ$ can be estimated by applying $K$-means to the rows of either $\widetilde{\bbeta}^{(\mathrm{grad})}$ or $\widetilde{\bxi}^{(\mathrm{grad})}$, after which $\bbeta$ and $\bxi$ are estimated accordingly. Alternatively, the two embeddings can be combined to estimate the cluster assignments jointly. We investigate several such approaches and compare their performance in the simulation study (Section~\ref{sim: multi_cluster}).

\section{Statistical guarantees}
\label{sec: theory}

Under the proposed \texttt{ZITS} model, the Hi-C tensor admits the decomposition
\[
\bm{\mathcal C}
=
(\bm{\mathcal J}-\bm{\mathcal P})*\bm{\Lambda}
+
\left[
\bm{\mathcal C}
-
(\bm{\mathcal J}-\bm{\mathcal P})*\bm{\Lambda}
\right],
\]
where the first term represents the underlying signal and the second the random noise. Establishing statistical guarantees therefore requires both identifiability of the signal and concentration of the noise. In this section, we first establish identifiability of the zero-inflated Poisson model and the coupled tensor decomposition (Section~\ref{sec: identifiability}), then derive concentration inequalities for zero-inflated Poisson random variables (Section~\ref{sec: tail}), prove consistency of the maximum likelihood estimator (Section~\ref{sec: consistency}), and finally study false-zero imputation from the perspective of a Bayes-optimal classifier (Section~\ref{sec: false}).

\subsection{Identifiability}
\label{sec: identifiability}

The identifiability of the zero-inflated Poisson model \eqref{equ: zero_inflated} is fundamental to the proposed framework. Observed zeros may correspond either to true absence of chromatin contacts (structural zeros) or to contacts masked by technical limitations (technical zeros). Identifiability guarantees that these two sources of zeros can be distinguished probabilistically through the unique parameters $(p_{i,j,k},\lambda_{i,j,k})$, thereby enabling principled false-zero detection and imputation. We first establish identifiability of the zero-inflated Poisson distribution.

\begin{proposition}
\label{prop: ZIP_identifiability}
(Identifiability of zero-inflated Poisson distribution)
Suppose there are two sets of parameters
$\{(p_{i, j, k}, \lambda_{i, j,k}): i\le j \in [n], k\in [K] \}$
and
$\{(p^\prime_{i, j, k}, \lambda^\prime_{i, j,k}): i\le j \in [n], k\in [K]\}$
that parameterize the same distribution in \eqref{equ: zero_inflated}. Then
$p_{i, j, k} = p^\prime_{i, j, k}$ and
$\lambda_{i, j, k} = \lambda^\prime_{i, j, k}$,
for all $i\le j \in [n]$ and $k \in [K]$.
\end{proposition}

We next establish the identifiability of the proposed \texttt{ZITS} model. Once the coupled tensor decomposition is identifiable, the latent representations of the genomic loci and cells are uniquely determined (up to the usual indeterminacies of tensor decompositions), and the cell cluster memberships can be recovered.
Let $\kappa_{\bA}$ denote the Kruskal rank of a matrix $\bA$, defined as the largest integer $\kappa$ such that every collection of $\kappa$ columns of $\bA$ is linearly independent. 

\begin{proposition}
\label{prop: identifiablity}
(Identifiability of coupled tensor decomposition)
Suppose the rows of the matrix $[\bbeta, \bxi]$ are distinct and
$\kappa_{\balpha} + \frac{1}{2}\kappa_{(\bbeta^\top,  \bxi^\top)^\top } \ge L + 1$.
Then the cluster membership matrix is identifiable up to relabeling, i.e., column permutation of $\bZ$. Given a particular ordering of the cluster labels, the model parameters $\bbeta$ and $\bxi$ are identifiable up to column permutation and scaling. The embedding matrix $\balpha$ is identifiable up to the same column scalar multiplications and permutations as those of $\bbeta$ and $\bxi$. Moreover, if $\bH$ has full column rank, then $\bGamma$ is uniquely identifiable.
\end{proposition}

We note that the sufficient condition
$\kappa_{\balpha} + \frac{1}{2}\kappa_{(\bbeta^\top, \bxi^\top)^\top } \ge L + 1$
is mild and easily satisfied in practice. For example, if both $\balpha$ and $(\bbeta^\top, \bxi^\top)^\top$ have full column rank, the condition reduces to $L \ge 2$. Moreover, the Kruskal rank of $(\bbeta^\top, \bxi^\top)^\top$ is no smaller than that of either $\bbeta$ or $\bxi$, implying that the proposed coupled tensor decomposition satisfies weaker identifiability requirements than a conventional single tensor decomposition.

Finally, we briefly discuss the choice of the embedding dimension $L$. Since chromatin is embedded in three-dimensional space, a natural choice is $L=3$. However, a slightly larger embedding dimension may provide a more flexible representation of complex chromatin structures while maintaining a relatively small number of basis functions. On the other hand, choosing $L$ excessively large relative to the number of basis functions $Q$ may weaken the identifiability guarantee in Proposition~\ref{prop: identifiablity}, since
$\kappa_{\balpha}\le \min\{N,Q,L\}$
under the smoothness constraint $\balpha = \bH\bGamma$. Consequently, $L$ should be chosen to balance model flexibility, computational efficiency, and theoretical identifiability.

\subsection{Tail behavior of the zero-inflated Poisson distribution}
\label{sec: tail}

The concentration properties of zero-inflated Poisson random variables play a key role in the consistency analysis developed in the next subsection. Since the Poisson distribution is sub-exponential, it is natural to ask whether the same property holds for the zero-inflated Poisson distribution. Recall that a random variable is sub-exponential if its $\psi_1$-Orlicz norm is finite. The following proposition characterizes the $\psi_1$-Orlicz norm of the ZIP distribution.

\begin{proposition}
\label{prop: psi_1}
Suppose $X\sim$ Poisson$(\lambda)$ and $C\sim$ ZIP$(p,\lambda)$. Then
\begin{align*}
\|X\|_{\psi_1}
=
\left(\log\left(1+\frac{1}{\lambda}\log2\right)\right)^{-1}
\qquad\text{and}\qquad
\|C\|_{\psi_1}
=
\left(\log\left(1+\frac{1}{\lambda}\log\frac{2-p}{1-p}\right)\right)^{-1},
\end{align*}
where
$\|\cdot\|_{\psi_1}
=
\inf\left\{
t>0:
\EE e^{|\cdot|/t}\le2
\right\}
$
denotes the $\psi_1$-Orlicz norm.
\end{proposition}

Proposition~\ref{prop: psi_1} implies that
$\|C\|_{\psi_1}<\|X\|_{\psi_1}$
for every $p\in(0,1)$, confirming that the zero-inflated Poisson distribution has lighter tails than the corresponding Poisson distribution. This also follows from the representation
$C=B\widetilde C$,
where $B\sim\mathrm{Bernoulli}(1-p)$ and $\widetilde C\sim\mathrm{Poisson}(\lambda)$ are independent, since for any $u>0$,
$$
\PP(C\ge u)
=
\PP(\widetilde C\ge u\mid B=1)\PP(B=1)
=
(1-p)\PP(X\ge u)
<
\PP(X\ge u).
$$
An immediate consequence of Proposition~\ref{prop: psi_1} is the tail bound
$\PP(C>u)
\le
2\left(1+\frac{1}{\lambda}\log\frac{2-p}{1-p}\right)^{-u},$
which shows that the tail probability of $C$ decays exponentially. 
 However, this bound does not provide a two-sided concentration inequality for $C$, nor does it directly generalize to sharp concentration bounds for averages of independent zero-inflated Poisson random variables. To address this limitation, we derive an alternative characterization of the zero-inflated Poisson distribution through its centered moment generating function.

\begin{theorem}
\label{thm: concentration}
Let $C\sim \mathrm{ZIP}(p,\lambda)$. Then its centered moment generating function satisfies
\begin{align*}
\log \EE e^{t \left(C - (1-p)\lambda \right)}
\le
\frac{\lambda [1 + p(1-p)\lambda] t^2}
{2(1 - \max\{1, \lambda\}t)}
\qquad \text{for} \quad
|t| < \frac{1}{\max\{1,\lambda\}}.
\end{align*}
Consequently, if $C_1,\ldots,C_n$ are independent zero-inflated Poisson random variables with parameters $(p_i,\lambda_i)$, then for any $\ba\in\RR^n$ and $M>0$,
\begin{align*}
& \quad \PP\left(\frac{1}{n} \sum_{i = 1}^n a_i [C_i - (1-p_i)\lambda_i]\ge M  \right)\\
& \le \exp\left( - \frac{n^2M^2}{2\left(\sum_{i = 1}^n a_i^2 \lambda_i[1+p_i(1-p_i) \lambda_i] +nM\max\{\|\ba\|_{\infty}, \|\ba * \bm{\lambda}\|_{\infty} \} \right  )}\right),
\end{align*}
where $\bm{\lambda}=(\lambda_i)\in\RR^n$ and $\ba * \bm{\lambda}\in\RR^n$ is defined by
$(\ba * \bm{\lambda})_i=a_i\lambda_i$.
\end{theorem}

The first part of Theorem~\ref{thm: concentration} establishes a bound on the centered moment generating function of the zero-inflated Poisson distribution. Combined with Bernstein's inequality, it yields the concentration inequality in the second part of the theorem for linear combinations of independent zero-inflated Poisson random variables. These result form one of the main probabilistic tools used to establish consistency in the next subsection.

\subsection{Estimation consistency}
\label{sec: consistency}

In this subsection, we establish consistency of the proposed likelihood estimator under the coupled low-rank tensor model. We first introduce the parameter space and define a class of approximate likelihood estimators, then derive a non-asymptotic Kullback--Leibler bound, which yields consistency of the estimated intensity and masking probability tensors.

We begin by defining the parameter space
\begin{align*}
\Omega = \{(\bGamma, \bZ, \bbeta, \bxi): \|\bGamma_{., l}\| = 1,\text{ for } l\in [L], \bZ \in \Delta, \|\bbeta\|_{\infty} \le \beta_{\max}, \|\bxi\|_{\infty} \le \xi_{\max}\}, 
\end{align*}
where $\Delta$ denotes the collection of all clustering assignment matrices, and $\beta_{\max}$ and $\xi_{\max}$ specify the parameter bounds. These constraints imply the existence of constants $\lambda_{\min}$ and $\lambda_{\max}$ such that
$\lambda_{\min}\le\lambda_{i,j,k}\le\lambda_{\max},$
for all $i\neq j\in[N]$ and $k\in[K]$. Similarly, the range of the masking probabilities is determined by $\xi_{\max}$. Throughout the analysis, we assume $\lambda_{\max}, \xi_{\max}, \beta_{\max}>1$. We further fix the columns of $\bGamma$ to have unit norm in order to remove the scaling ambiguity in the coupled tensor decomposition of $\bm{\eta}$ and $\bTheta$.

Next, let $(\bGamma^*, \bZ^*, \bbeta^*, \bxi^*)\in\Omega$ denote the true model parameters.
With a slight abuse of notation, we write
$\mathcal{L}(\bGamma,\bZ,\bbeta,\bxi;\bm{\mathcal C})$
for the negative log-likelihood obtained by substituting the corresponding tensor decomposition into
$\mathcal{L}(\bm{\eta},\bTheta;\bm{\mathcal C})$.
Under the correctly specified model, $(\bGamma^*,\bZ^*,\bbeta^*,\bxi^*)$ minimizes the expected negative log-likelihood. We consider any estimator
$(\widehat{\bGamma},\widehat{\bZ},\widehat{\bbeta},\widehat{\bxi})\in\Omega$
satisfying
\begin{align}
\label{equ: estimator}
\mathcal{L}
(\widehat{\bGamma},\widehat{\bZ},\widehat{\bbeta},\widehat{\bxi};\bm{\mathcal C})
\le
\mathcal{L}
(\bGamma^*,\bZ^*,\bbeta^*,\bxi^*;\bm{\mathcal C})
+\epsilon_{N,K},
\end{align}
where $\epsilon_{N,K}>0$ is a vanishing sequence whose order will be specified later. This class includes both the exact maximum likelihood estimator and approximate solutions obtained by numerical optimization. As $\epsilon_{N,K}\to0$, the estimator approaches the global maximum likelihood estimator. Our goal is to establish consistency for this broader class of estimators.

Let
$(\bm{\eta}^*,\bTheta^*)$
and
$(\widehat{\bm{\eta}},\widehat{\bTheta})$
denote the tensors corresponding to the true parameters and the estimator, respectively. Likewise, let
$\bm{\Lambda}^*$,
$\widehat{\bm{\Lambda}}$,
$\bm{\mathcal P}^*$,
and
$\widehat{\bm{\mathcal P}}$
denote the corresponding intensity and masking probability tensors. The following theorem establishes a non-asymptotic bound on the average Kullback--Leibler divergence between the true and estimated zero-inflated Poisson models.

\begin{theorem}
\label{thm: large_deviation_inequality}
Suppose
$1-p_{i,j,k}\ge s_{N,K},$
for some sparsity level
$
s_{N,K}\le \frac{1}{e^{\lambda_{\max}}-1},
$
possibly depending on $N$ and $K$, for all $i,j\in[N]$ and $k\in[K]$. Let
$(\widehat{\bGamma},\widehat{\bZ},\widehat{\bbeta},\widehat{\bxi})$
be any estimator satisfying (\ref{equ: estimator}). Then there exist universal constants
$M_1$, $M_2$, and $M_3$ such that
\begin{align*}
\PP\left(
\frac{1}{\varphi_{N,K}}
\sum_{i\le j,\,k}
KL_{\mathrm{ZIP}}
\!\left(
p_{i,j,k}^*,\lambda_{i,j,k}^*
\,
\middle\|
\,
\widehat p_{i,j,k},
\widehat\lambda_{i,j,k}
\right)
\ge
M_1\epsilon_{N,K}
\right)
\le
M_2
\exp\!\left(
-M_3\epsilon_{N,K}\varphi_{N,K}s_{N,K}
\right),
\end{align*}
where
$
\epsilon_{N,K}
=
\frac{(QL+KR+2RL)\log\varphi_{N,K}}
{\varphi_{N,K}s_{N,K}}.
$
\end{theorem}

Similar to sparse multi-layer network models \citep{xu2023covariate, zhen2024consistent}, the sparsity level $s_{N,K}$ quantifies the signal strength of the Hi-C tensor and therefore appears in the denominator of the convergence rate. The quantity
$QL+KR+2RL$
is the number of free parameters in the proposed coupled tensor decomposition, while
$\varphi_{N,K}=N(N-1)K$
is the number of independent observations in the Hi-C tensor. 

The proof of Theorem~\ref{thm: large_deviation_inequality} is non-trivial and combines the concentration results of Section~\ref{sec: tail} with a novel analysis of the Kullback--Leibler divergence based on the hurdle Poisson representation of the zero-inflated Poisson distribution. Although every zero-inflated Poisson distribution admits an equivalent hurdle Poisson representation (but not conversely), this representation substantially simplifies the analysis. Additional details are provided in Supplementary~\ref{sec: appendix_ZIP_more}.

Combining Theorem~\ref{thm: large_deviation_inequality} with a sequence of novel inequalities established in Supplementary~\ref{sec: appdendix_D_lemmas} yields consistency of the estimated masking probability and intensity tensors.

\begin{theorem}
\label{thm: consistency_F_norm}
Under the conditions of Theorem~\ref{thm: large_deviation_inequality}, we have
\begin{align*}
& \quad \frac{1}{\varphi_{N, K}}\left(\|\widehat{\bm{\mathcal{P}}} - \bm{\mathcal{P}}^*\|_F^2 + \|\widehat{\bm{\Lambda}} - \bm{\Lambda}^*\|_F^2\right)\\
&  \le \frac{8M_1}{(1 - e^{-\lambda_{\min}})^{2}} \max \left\{ 1 - s_{N,K} +s_{N, K} e^{-\lambda_{\min}}, \frac{16(\lambda_{\max} + 1)}{s_{N, K}(1 - e^{-\lambda_{\min}})}   \right\} \epsilon_{N, K},
\end{align*}
with probability at least $1 -M_2 \exp\left(- M_3\epsilon_{N, K} \varphi_{N, K}s_{N, K}\right)$.
\end{theorem}

Theorem~\ref{thm: consistency_F_norm} establishes consistency of both the estimated masking probability tensor $\widehat{\bm{\mathcal P}}$ and the estimated intensity tensor $\widehat{\bm{\Lambda}}$ under the normalized Frobenius norm. In particular, when $\lambda_{\min}$, $\lambda_{\max}$, and $s_{N, K}$ are of constant order, we have
$
\frac{1}{\varphi_{N,K}}
\left(
\|\widehat{\bm{\mathcal P}}-\bm{\mathcal P}^*\|_F^2
+
\|\widehat{\bm{\Lambda}}-\bm{\Lambda}^*\|_F^2
\right)
=
O_p(\epsilon_{N,K}).
$
Since both $\bm{\mathcal P}$ and $\bm{\Lambda}$ 
inherit the proposed structured representation, these consistency results justify recovery of the latent genomic embeddings and cell clusters under suitable separation conditions.

\subsection{False zero detection and imputation}
\label{sec: false}

The proposed zero-inflated Poisson model naturally enables principled detection of false zeros. Recall that a zero-inflated Poisson random variable admits the representation
$
C=B\widetilde{C},
$
where $\widetilde{C}\sim\mathrm{Poisson}(\lambda)$ denotes the latent contact count and
$B\sim\mathrm{Bernoulli}(1-p)$ is an independent masking variable. Consequently, whenever $C\neq0$, we necessarily have $\widetilde{C}=C$ and $B=1$. In contrast, an observed zero may arise from two distinct sources: It is a structural zero if $\widetilde{C}=0$, and a false zero if $\widetilde{C}>0$ but $B=0$. Since the latent counts $\widetilde{C}$ are unobserved, distinguishing between these two cases becomes a binary classification problem.

Specifically, let $Y$ denote the latent class label, where $Y=1$ if an observed zero is a false zero and $Y=0$ otherwise. A direct calculation shows that
\begin{align*}
\PP(Y=1\mid C=0)
=
\frac{p(1-e^{-\lambda})}
{p(1-e^{-\lambda})+e^{-\lambda}}.
\end{align*}
For any classifier $\Psi$, define its conditional classification risk by
$\mathcal{R}(\Psi)
=
\PP\!\left(\Psi(C)\neq Y\,\middle|\,C=0\right),
$
and let
$
\Psi^*
=
\arg\min_\Psi \mathcal{R}(\Psi)
$
denote the Bayes classifier. The following theorem characterizes both the optimal decision rule and the corresponding excess risk.

\begin{theorem}
\label{thm: calssification}
Suppose $C\sim\mathrm{ZIP}(p,\lambda)$. Conditional on $C=0$, the Bayes classifier for distinguishing false zeros from structural zeros is
\begin{align*}
\Psi^*(0)
=
\begin{cases}
1, & \text{if } p>(e^\lambda-1)^{-1},\\
0, & \text{otherwise}.
\end{cases}
\end{align*}
Moreover, for any classifier $\Psi$, its excess risk is
\begin{align*}
\mathcal{R}(\Psi)-\mathcal{R}(\Psi^*)
=
\frac{\left|p(1-e^{-\lambda})-e^{-\lambda}\right|}
{p(1-e^{-\lambda})+e^{-\lambda}}
\,
\ind\!\left\{\Psi(0)\neq\Psi^*(0)\right\}.
\end{align*}
\end{theorem}

Theorem~\ref{thm: calssification} directly motivates our false-zero detection procedure. Given the consistent estimators $\widehat{\bm{\mathcal P}}$ and $\widehat{\bm{\Lambda}}$, we classify an observed zero $C_{i,j,k}=0$ as a false zero whenever
${\widehat p_{i,j,k}>
\left(e^{\widehat\lambda_{i,j,k}}-1\right)^{-1}.}$
Once a zero entry is classified as false, we impute it by the estimated Poisson intensity $\widehat\lambda_{i,j,k}$. This yields a principled, model-based imputation procedure that follows directly from the Bayes-optimal decision rule.

\section{Simulation studies}
\label{sec: simulation}

In this section, we investigate the finite-sample performance of the proposed \texttt{ZITS} estimator. We first describe the data generation procedure (Section~\ref{sim: date_generation}) and compare several initialization strategies for the non-convex likelihood optimization problem (Section~\ref{sim: initialization}). We then study the single-cluster setting (Section~\ref{sim: one_cluster}), evaluating the recovery of the parameter tensors and the performance of the proposed false-zero detection procedure as the number of observed cells increases. Finally, we examine the multi-cluster setting (Section~\ref{sim: multi_cluster}) by comparing the ability of different model outputs to recover the underlying cell populations.

\subsection{Data generation}
\label{sim: date_generation}

We generate the data directly from the proposed zero-inflated tensor model. Genomic loci are represented through piecewise smooth latent embeddings, while cells are represented through cluster-specific embeddings. These embeddings determine the parameter tensors through the coupled tensor decomposition, after which the contact counts are generated according to the zero-inflated Poisson model of Section~\ref{sec: statmodel}.

We first generate the genomic embedding matrix $\balpha \in \mathbb{R}^{N \times L}$. Let $\overline{\balpha}\in\mathbb{R}^{L\times L}$ denote the mean embedding matrix with diagonal entries $\mu_\alpha$ and off-diagonal entries $\mu_\alpha/L$. The $N$ genomic loci are then partitioned into $L$ consecutive segments of equal size $\lfloor \frac NL \rfloor$, with the last segment containing the remaining loci whenever $N$ is not divisible by $L$. For a locus $i$ belonging to segment $j$; that is,
$
(j-1)\lfloor\frac{N}{L}\rfloor
<
i
\le
j\lfloor\frac{N}{L}\rfloor,
$
for $j=1,\ldots,L-1$, or
$
(L-1)\lfloor\frac{N}{L}\rfloor
<
i
\le
N
$
when $j=L$, the embedding vector is generated by perturbing the corresponding segment center:
\[
\alpha_{i,\ell}
\overset{\text{i.i.d.}}{\sim}
\overline{\alpha}_{j,\ell}
+
\mathrm{Unif}\bigl([0,\sigma_\alpha]\bigr),
\]
for $\ell\in[L]$, where $\sigma_\alpha$ controls the within-segment variability. 

The cell embedding matrices $\bbeta,\bxi\in\mathbb{R}^{K\times L}$ are generated analogously without within-cluster perturbations. The $K$ cells are partitioned into $R$ equal-sized clusters. Specifically, cell $k$ belongs to cluster $r$ if
$(r-1)\left\lfloor\frac{K}{R}\right\rfloor
<
k
\le
r\lfloor\frac{K}{R}\rfloor,
$
for $r=1,\ldots,R-1$, or
$
(R-1)\lfloor\frac{K}{R}\rfloor
<
k
\le
K
$
when $r=R$. We then generate cluster-center matrices $\overline{\bbeta},\overline{\bxi}\in\mathbb{R}^{R\times L}$ with entries
\[
\overline{\beta}_{r,\ell}
\overset{\text{i.i.d.}}{\sim}
\mathrm{Unif}([\mu_\beta,\mu_\beta+\sigma_\beta]),
\qquad
\overline{\xi}_{r,\ell}
\overset{\text{i.i.d.}}{\sim}
\mathrm{Unif}([\mu_\xi,\mu_\xi+\sigma_\xi]),
\]
for $r\in[R]$ and $\ell\in[L]$, where $\mu_\beta,\mu_\xi,\sigma_\beta$, and $\sigma_\xi$ are user-specified parameters. Each cell is then assigned the embedding corresponding to its cluster center; that is,
$
\beta_{k,\ell}
=
\overline{\beta}_{r,\ell}$ and $
\xi_{k,\ell}
=
\overline{\xi}_{r,\ell},
$
whenever cell $k$ belongs to cluster $r$. 

Given the latent embeddings $\balpha$, $\bbeta$, and $\bxi$, we construct the parameter tensors
$
\bm{\eta}
=
\bm{\mathcal I}
\times_1
\balpha
\times_2
\balpha
\times_3
\bbeta
$
and
$
\bTheta
=
\bm{\mathcal I}
\times_1
\balpha
\times_2
\balpha
\times_3
\bxi.
$
The contact counts $C_{i,j,k}$ are then generated as
$
C_{i,j,k}
=
\widetilde C_{i,j,k}B_{i,j,k},
$
where $\widetilde C_{i,j,k}$ and $B_{i,j,k}$ are generated independently from
$
\widetilde C_{i,j,k}
\sim
\mathrm{Poisson}(\lambda_{i,j,k})
$
and
$
B_{i,j,k}
\sim
\mathrm{Bernoulli}(1-p_{i,j,k}),
$
with
$
\lambda_{i,j,k}=e^{\eta_{i,j,k}}
$
and
$
p_{i,j,k}
=
\frac{1}{e^{\theta_{i,j,k}}+1}.
$

\subsection{Initialization}
\label{sim: initialization}

Because the proposed likelihood is non-convex, the quality of the final estimator depends on the initialization. In this subsection, we propose several initialization strategies for the latent factors $\balpha$, $\bbeta$, and $\bxi$, and compare their performance in the next subsection. For simplicity, we focus on the single-cluster setting; the proposed procedures naturally extend to multiple clusters once an initial clustering of the cells is available.

In the single-cluster setting, all cells share the same latent representations, so the matrices $\bbeta$ and $\bxi$ have identical rows, i.e.
$\bbeta=\mathbf1_K\bbeta^{(0)\top},
$ and $\bxi=\mathbf1_K\bxi^{(0)\top}.$
Consequently, every frontal slice of the parameter tensors $\bm{\eta}$ and $\bTheta$ is identical and can be written as
\[
\bm{\eta}^{(0)}
=
\balpha
\operatorname{diag}(\bbeta^{(0)})
\balpha^\top,
\qquad
\bTheta^{(0)}
=
\balpha
\operatorname{diag}(\bxi^{(0)})
\balpha^\top.
\]
Therefore, estimating $\bm{\eta}^{(0)}$ and $\bTheta^{(0)}$ provides an initialization for the latent factors $\balpha$, $\bbeta$, and $\bxi$. To obtain these matrices, we first apply the method of moments.

Under the single-cluster model, the observations
$C_{i,j,1},\ldots,C_{i,j,K}$
are independent and identically distributed for each pair of genomic loci $(i,j)\in[N]\times[N]$. We therefore compute the sample mean and variance,
$m_{i,j}
=
\frac1K\sum_{k=1}^KC_{i,j,k}$ and $
v_{i,j}
=
\frac1K
\sum_{k=1}^K
(C_{i,j,k}-m_{i,j})^2,
$
and obtain the method-of-moments estimators
\[
\lambda_{i,j}^{(0)}
=
\frac{v_{i,j}+m_{i,j}^2}{m_{i,j}}-1,
\qquad
p_{i,j}^{(0)}
=
\frac{v_{i,j}-m_{i,j}}
{v_{i,j}+m_{i,j}^2-m_{i,j}}.
\]
The initial parameter matrices are then defined by
$
\eta_{i,j}^{(0)}
=
\log\lambda_{i,j}^{(0)}$ and 
$
{\Theta_{i,j}^{(0)}
=
\log\frac{1-p_{i,j}^{(0)}}{p_{i,j}^{(0)}}.}
$

Given $\bm{\eta}^{(0)}$ and $\bTheta^{(0)}$, we investigate six different strategies for recovering the latent factors $\balpha$, $\bbeta$, and $\bxi$. Since there is no canonical way to jointly balance the two matrix factorizations, these strategies differ in how they combine the information contained in $\bm{\eta}^{(0)}$ and $\bTheta^{(0)}$.

\begin{enumerate}
\item (\texttt{random}) Randomly initialize $\balpha$, $\bbeta^{(0)}$, and $\bxi^{(0)}$ without using $\bm{\eta}^{(0)}$ or $\bTheta^{(0)}$.

\item (\texttt{cp}) Stack $\bm{\eta}^{(0)}$ and $\bTheta^{(0)}$ into a third-order tensor with two frontal slices and compute a rank-$L$ CP decomposition
$
\bm{\mathcal I}
\times_1\balpha_1
\times_2\balpha_2
\times_3
(\bbeta^{(0)},\bxi^{(0)})^\top.
$
Set $\balpha=\balpha_1$. Since the decomposition is unconstrained, the factors $\balpha_1$ and $\balpha_2$ need not coincide.

\item (\texttt{cpavg}) Use the same CP decomposition as in \texttt{cp}, but set
$
\balpha=\frac{\balpha_1+\balpha_2}{2}.
$

\item (\texttt{eigenb}) Compute the rank-$L$ eigendecomposition of $\bm{\eta}^{(0)}$ and use the leading eigenvectors and eigenvalues to initialize $\balpha$ and $\bbeta^{(0)}$, respectively. Then estimate $\bxi^{(0)}$ by minimizing
$
\|\bTheta^{(0)}
-
\balpha\operatorname{diag}(\bxi^{(0)})\balpha^\top
\|_F^2,
$
which reduces to a linear regression problem.

\item (\texttt{eigenx}) Compute the rank-$L$ eigendecomposition of $\bTheta^{(0)}$ and use the leading eigenvectors and eigenvalues to initialize $\balpha$ and $\bxi^{(0)}$, respectively. Then estimate $\bbeta^{(0)}$ by minimizing
$\|\bm{\eta}^{(0)}
-
\balpha\operatorname{diag}(\bbeta^{(0)})\balpha^\top
\|_F^2.$

\item (\texttt{eigenbx}) Let $\balpha_\beta$ and $\balpha_\xi$ denote the leading $L$ eigenvectors of $\bm{\eta}^{(0)}$ and $\bTheta^{(0)}$, respectively. We estimate $\balpha$ by computing the leading $L$ left singular vectors of the concatenated matrix
$
(\balpha_\beta,\balpha_\xi),
$
thereby approximating the joint subspace
$\operatorname{span}(\balpha_\beta,\balpha_\xi)$.
The corresponding $\bbeta^{(0)}$ and $\bxi^{(0)}$ are then obtained by solving the two least-squares problems
$\|\bm{\eta}^{(0)}
-
\balpha\operatorname{diag}(\bbeta^{(0)})\balpha^\top
\|_F^2
$ and
$\|\bTheta^{(0)}
-
\balpha\operatorname{diag}(\bxi^{(0)})\balpha^\top
\|_F^2.$
\end{enumerate}

The numerical performance of these initialization strategies is compared in the next subsection.

\subsection{One cluster scenario}
\label{sim: one_cluster}

To evaluate the proposed initialization strategies, we consider a setting with $N=20$ genomic loci and $K=250$ cells, where the true embedding dimension is $L=5$. During estimation, we vary the input embedding dimension over $\widehat{L}\in\{1,3,5,7,9\}$ to investigate robustness to model misspecification. The remaining parameters are fixed at $\mu_\alpha=0.5$, $\sigma_\alpha^2=\mu_\alpha/4$, $\mu_\beta=5$, and $\sigma_\beta^2=\mu_\beta/4$. The parameter $\mu_\xi$ controls the sparsity of the observed tensor, with smaller values corresponding to higher sparsity. We consider $\mu_\xi\in\{1,5,20\}$ together with $\sigma_\xi^2=\mu_\xi/4$, representing high-, medium-, and low-sparsity settings, respectively.

Figure~\ref{fig:simu:init:re} reports the averaged relative errors
$\frac{\|\bm{\Lambda}-\widehat{\bm{\Lambda}}\|_F}{\|\bm{\Lambda}\|_F}$
and
$\frac{\|\bm{\mathcal P}-\widehat{\bm{\mathcal P}}\|_F}{\|\bm{\mathcal P}\|_F}$
under the high-sparsity setting ($\mu_\xi=1$). Overall, \texttt{eigenb} and \texttt{eigenbx} consistently achieve the smallest relative errors for both $\bm{\Lambda}$ and $\bm{\mathcal P}$. In contrast, the CP-based initializations perform only marginally better than random initialization, suggesting that unconstrained joint CP factorization of $\bm{\eta}^{(0)}$ and $\bTheta^{(0)}$ does not provide sufficiently accurate latent factors.
Comparing \texttt{eigenb} and \texttt{eigenx} further shows that initializing from $\bm{\eta}^{(0)}$ is consistently more effective than initializing from $\bTheta^{(0)}$, suggesting that the intensity tensor contains a stronger low-rank signal than the masking probability tensor. Although \texttt{eigenbx} combines information from both matrices, it performs similarly to \texttt{eigenb}, indicating that $\bTheta^{(0)}$ contributes only limited additional information.
Finally, both \texttt{eigenb} and \texttt{eigenbx} achieve their lowest errors when $\widehat L=L=5$. More importantly, their performance deteriorates only slightly when $\widehat L>L$, whereas all methods perform substantially worse when $\widehat L<L$. Thus, when the true embedding dimension is unknown, slight overestimation is preferable to underestimation.

\begin{figure}[ht!]
     \centering
     \begin{subfigure}[b]{0.43\textwidth}
         \centering
         \includegraphics[width=\textwidth]{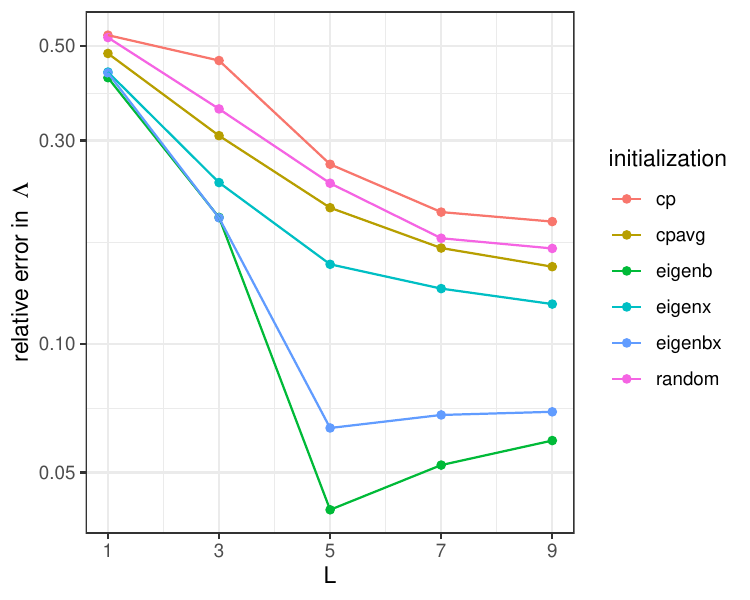}
     \end{subfigure}
       \hfill
     \begin{subfigure}[b]{0.43\textwidth}
         \centering
         \includegraphics[width=\textwidth]{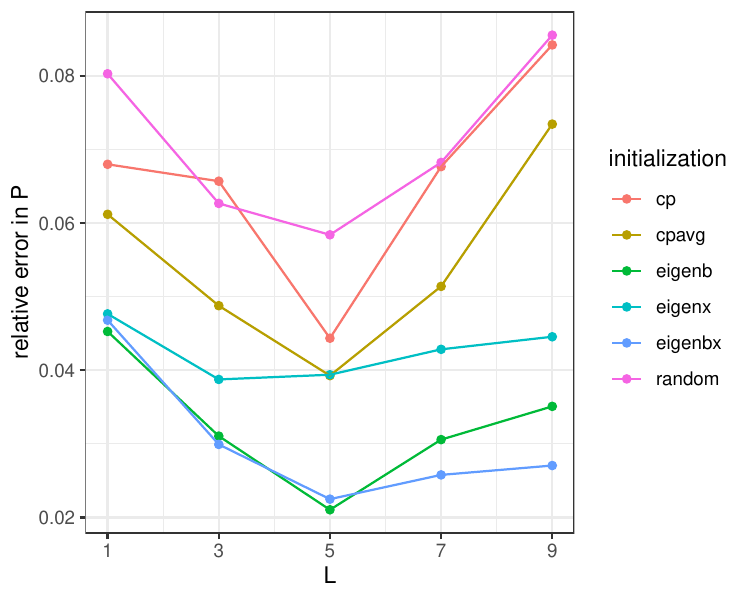}
     \end{subfigure}
        \caption{Average relative Frobenius errors of the estimated intensity tensor $\widehat{\bm{\Lambda}}$ (left) and masking probability tensor $\widehat{\bm{\mathcal P}}$ (right) for six initialization strategies, averaged over 50 simulated datasets. The true embedding dimension is $L=5$, while the fitting procedure uses varying values of $\widehat L$. Results are shown under the high-sparsity setting ($\mu_\xi=1$).}
        \label{fig:simu:init:re}
\end{figure}

We next investigate the effect of the number of cells by varying
$
K\in\{25,50,100,250,500\},
$
while keeping all other simulation settings unchanged and using the \texttt{eigenb} initialization.

The top row of Figure~\ref{fig:simu:K} reports the average relative Frobenius errors of the estimated intensity and masking probability tensors under different sparsity levels. As expected, the estimation accuracy improves as the number of cells increases across all sparsity regimes, consistent with the consistency results established in Section~\ref{sec: consistency}. The effect of sparsity, however, differs for the two parameters. As the observed tensor becomes sparser, estimation of $\bm{\Lambda}$ becomes slightly more difficult, whereas estimation of $\bm{\mathcal P}$ improves substantially. Consequently, sparse observations contain less information about the latent contact intensities but considerably more information about the masking probabilities. This phenomenon may be viewed as a blessing of sparsity for zero-inflated tensor models.

Next, we evaluate the false-zero detection procedure proposed in Section~\ref{sec: false}. An observed zero $C_{i,j,k}=0$ is classified as a false zero whenever
$\widehat p_{i,j,k}>
\left(e^{\widehat\lambda_{i,j,k}}-1\right)^{-1}.
$
The bottom row of Figure~\ref{fig:simu:K} summarizes the resulting accuracy, precision, and recall. Classification accuracy increases monotonically with the number of cells and is substantially higher in sparse settings, improving from roughly $60\%$ under low sparsity to nearly $90\%$ under high sparsity. This mirrors the improved estimation of the masking probability tensor observed in the top row.

The precision--recall curves provide further insight into this behavior. Under low sparsity, recall is close to zero, indicating that most false zeros remain undetected, while precision remains relatively high because the few detected false zeros are usually classified correctly. As the number of cells increases, recall improves substantially whereas precision changes little. Under medium- and high-sparsity settings, both precision and recall are close to one, indicating accurate false-zero detection across a wide range of sample sizes.

\begin{figure}[h]
    \centering
    \begin{subfigure}[b]{0.49\textwidth}
        \centering
        \includegraphics[width=\textwidth]{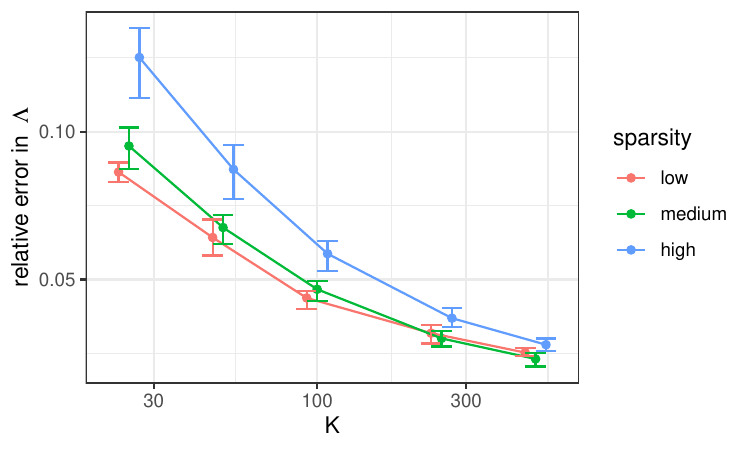}
        \label{fig:simu:K:reL}
    \end{subfigure}
    \hfill
    \begin{subfigure}[b]{0.49\textwidth}
        \centering
        \includegraphics[width=\textwidth]{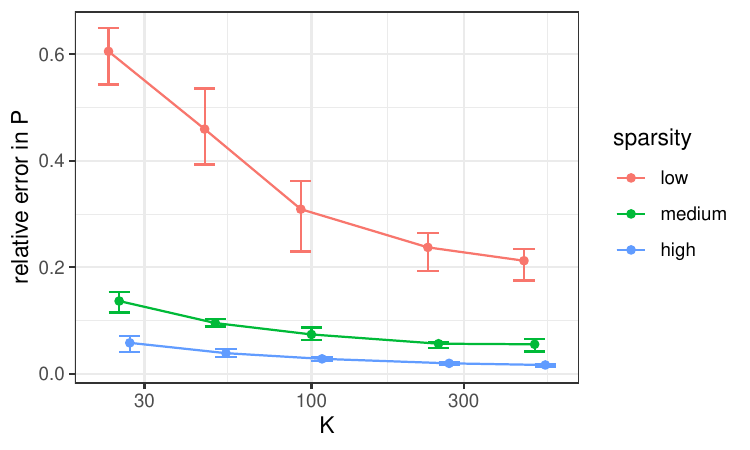}
        \label{fig:simu:K:reP}
    \end{subfigure}

    \begin{subfigure}[b]{0.49\textwidth}
        \centering
        \includegraphics[width=\textwidth]{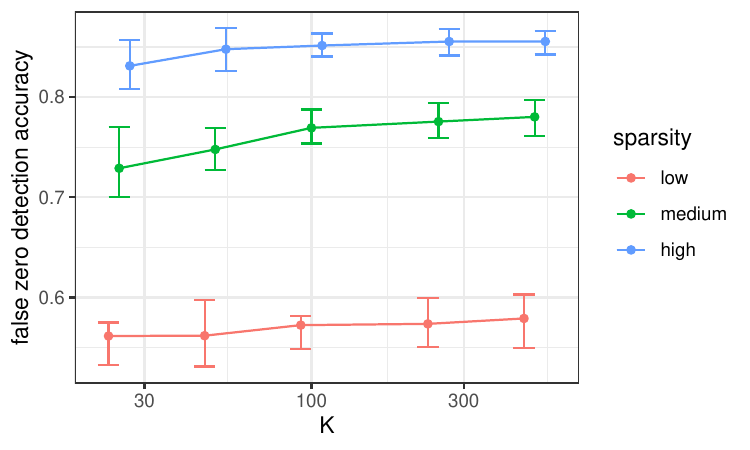}
        \label{fig:simu:K:acc}
    \end{subfigure}
    \hfill
    \begin{subfigure}[b]{0.45\textwidth}
        \centering
        \includegraphics[width=\textwidth]{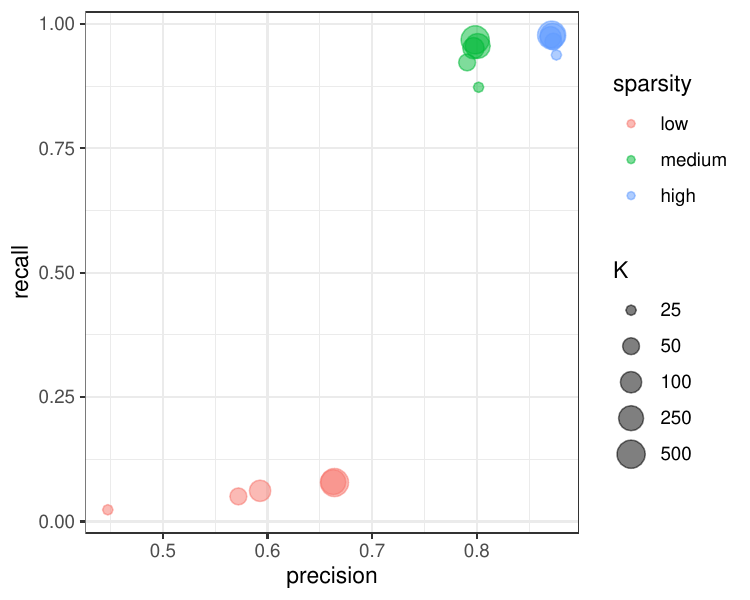}
        \label{fig:simu:K:pr}
    \end{subfigure}
    \caption{Performance of the proposed estimator as the number of cells $K$ increases under different sparsity levels (controlled by $\mu_\xi$). The top row reports the average relative Frobenius errors of the estimated intensity tensor $\widehat{\bm{\Lambda}}$ and masking probability tensor $\widehat{\bm{\mathcal P}}$. The bottom row summarizes the performance of the Bayes-based false-zero detection procedure in terms of accuracy, precision, and recall.}
    \label{fig:simu:K}
\end{figure}

\subsection{Multi-cluster scenario}
\label{sim: multi_cluster}

In this subsection, we investigate the performance of \texttt{ZITS} in the presence of multiple cell populations. We set $N=60$, $K=240$, and $L=5$, and consider the number of clusters
$R\in\{2,3,4,5,6\}.$
Each value of $R$ induces $R$ blocks of repeated rows in $\bbeta$ and $\bxi$, and hence $R$ groups of repeated frontal slices in the intensity and masking probability tensors. Throughout this experiment, we use the \texttt{eigenb} initialization.

We first study the effect of the number of clusters on parameter estimation. Figure~\ref{fig:simu:R:re} in Supplementary~\ref{app:sec:plots} reports the average relative Frobenius errors of the estimated intensity and masking probability tensors. As expected, the estimation error increases with the number of clusters for both tensors. Since the total number of cells is fixed, increasing $R$ reduces the number of observations available within each cluster, making estimation progressively more challenging. The effect is more pronounced for the intensity tensor $\bm{\Lambda}$ than for the masking probability tensor $\bm{\mathcal P}$.

We next evaluate how well the fitted model recovers the underlying cell clusters. In addition to the estimated parameter matrices $\widehat{\bbeta}$ and $\widehat{\bxi}$, we consider clustering based on the estimated tensors $\widehat{\bm{\Lambda}}$, $\widehat{\bm{\mathcal P}}$, expectation tensor $(\bm{\mathcal J}-\widehat{\bm{\mathcal P}})*\widehat{\bm{\Lambda}}$, and the imputed tensor obtained using the false-zero detection procedure of Section~\ref{sec: false}. 

For the estimated cell embedding matrices $\widehat{\bbeta}$ and $\widehat{\bxi}$, we apply $k$-means clustering directly to their rows. For the tensor-valued quantities, we vectorize the upper-triangular entries of each frontal slice, retain the first 20 principal components, and apply $k$-means in the resulting low-dimensional representation. Clustering performance is evaluated using the adjusted Rand index (ARI), and the results are shown in Figure~\ref{fig:simu:R}.

Overall, cluster recovery becomes more difficult as the number of clusters increases, regardless of the quantity used for clustering. Nevertheless, the proposed imputation procedure substantially improves clustering accuracy compared with the observed data, particularly under high-sparsity settings, demonstrating that recovering false zeros enhances the biological signal relevant for downstream analysis.

Among the model parameters, the estimated cell embedding matrix $\widehat{\bbeta}$ consistently provides the strongest cluster separation, indicating that the Poisson intensity captures most of the cluster-specific variation. In contrast, clustering based on $\widehat{\bm{\mathcal P}}$ and $\widehat{\bxi}$ performs substantially worse. This is likely because the entries of $\bm{\mathcal P}$ are constrained to lie in $[0,1]$, making cluster-specific patterns less separable than those of $\bm{\Lambda}$.

\begin{figure}[ht!]
     \begin{subfigure}[b]{\textwidth}
         \centering
         \includegraphics[width=0.9\textwidth]{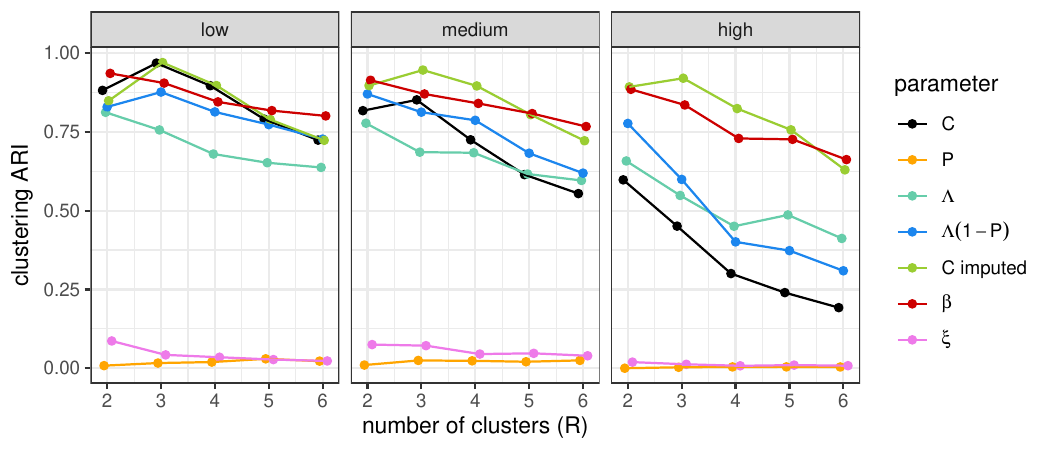}
     \end{subfigure}
     \caption{Adjusted Rand index (ARI) for recovering the true cell clusters using different estimated quantities. Columns correspond to different sparsity levels (controlled by $\mu_\xi$), while the horizontal axis shows the true number of clusters $R$.}
     \label{fig:simu:R}
\end{figure}

\section{Real Data Analysis}
\label{sec: hic}

We analyze the human single-cell Hi-C dataset of \citet{Ramani2017sc} at 10 Mb resolution, focusing on chromosomes 1--4, each of which contains more than 20 genomic loci at this resolution. The 10 Mb resolution is used to mitigate the extreme sparsity observed at finer resolutions, where more than 95\% of the entries are zero. The original dataset consists of four cell types: GM12878, HAP1, HeLa, and K562. Since GM12878 contains substantially fewer cells than the remaining cell types, we exclude it from the analysis. We then randomly sample 48 cells from each of HAP1 and HeLa and combine them with the 48 K562 cells, yielding a balanced dataset of $K=144$ cells from $R=3$ cell types. To reduce sampling variability, all experiments are repeated over 10 independent random subsamples. Following standard preprocessing for single-cell Hi-C data, we set the largest two diagonals of each contact matrix to zero to reduce diagonal dominance. Throughout the analysis, we use cubic B-splines with basis dimension $Q=5$.

We first investigate the effect of the embedding dimension. For chromosome~1, we vary
$L\in\{3,5,10,15,20\},$
and report the average ARI over 10 random subsamples. The left panel of Figure~\ref{fig:hic:ari} shows that clustering based on the imputed tensor consistently outperforms clustering based on the raw contact tensor, demonstrating that the proposed imputation improves cell-type separability. The ARI of the imputed tensor stabilizes once $L\ge10$, so we fix $L=10$ for the remainder of the analysis. Finally, the cell embeddings $\widehat{\bbeta}$ associated with the Poisson intensity consistently yield substantially higher ARI than the raw contact data, whereas the embeddings $\widehat{\bxi}$ associated with the masking probabilities provide only modest improvements.

To assess robustness across chromosomes, we repeat the experiment with $L=10$ for chromosomes~1-4. The results are shown in the right panel of Figure~\ref{fig:hic:ari}. The same conclusions hold across all chromosomes: imputation consistently improves clustering performance, and the embeddings obtained from $\widehat{\bbeta}$ consistently outperform those based on~$\widehat{\bxi}$.

\begin{figure}[h]
    \centering
     \begin{subfigure}[b]{0.49\textwidth}
         \centering
         \includegraphics[width=\textwidth]{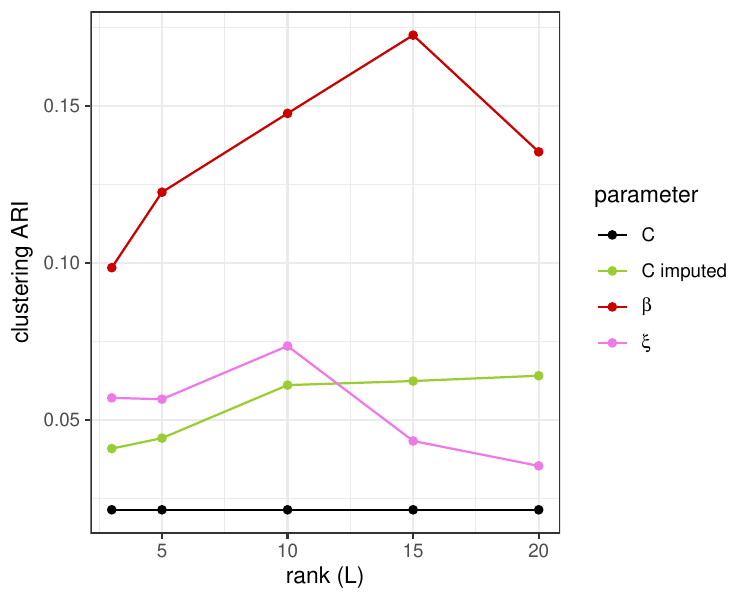}
     \end{subfigure}
    \hfill
    \begin{subfigure}[b]{0.49\textwidth}
        \centering
        \includegraphics[width=\textwidth]{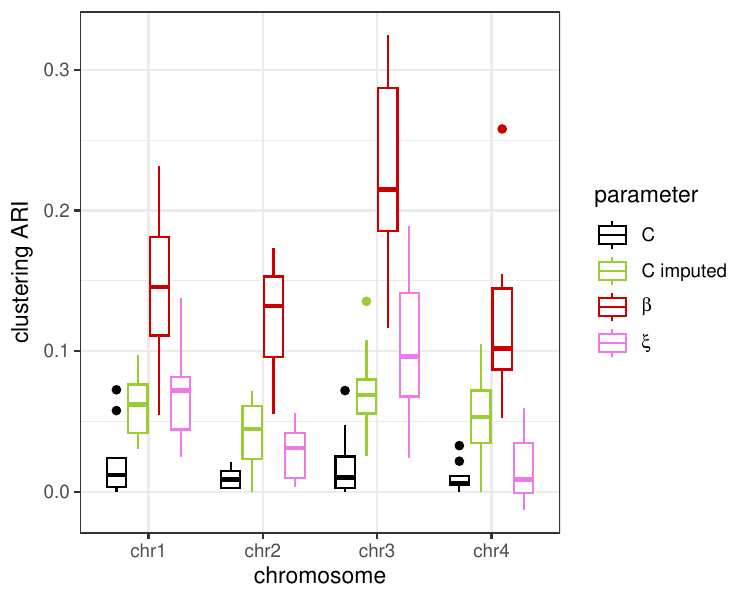}
    \end{subfigure}
    \caption{Cell-type clustering performance of \texttt{ZITS} on the Ramani single-cell Hi-C dataset. Left: ARI as a function of the embedding dimension $L$ for chromosome~1, averaged over 10 random subsamples. Right: ARI distributions across chromosomes with $L=10$.}
    \label{fig:hic:ari}
\end{figure}

We next compare \texttt{ZITS} with existing imputation and embedding methods. We first consider two methods that explicitly account for structural zeros, \texttt{HiCImpute} \citep{Xie2022hicimpute} and \texttt{scHiCSRS} \citep{Xu2022schicsrs}, using the authors' \texttt{R} implementations with default settings. As shown in the left panel of Figure~\ref{fig:hic:competitors}, \texttt{ZITS} consistently achieves the highest ARI, indicating superior recovery of the underlying cell-type structure. Among the competing methods, \texttt{HiCImpute} improves clustering relative to the raw data, whereas \texttt{scHiCSRS} yields substantially smaller improvements.

We further compare \texttt{ZITS} with \texttt{scHiCluster} \citep{Zhou2019schicluster}, which performs imputation through random walks on binarized contact matrices without explicitly modeling structural zeros. For a fair comparison, we binarize the imputed tensor produced by \texttt{ZITS} before clustering. The middle panel of Figure~\ref{fig:hic:competitors} shows that the binarized output of \texttt{ZITS} substantially outperforms both the raw binarized data and \texttt{scHiCluster}, demonstrating the advantage of modeling structural zeros before binarization.

Finally, we compare the cell embeddings produced by \texttt{ZITS} with those obtained from a standard CP tensor decomposition applied directly to the contact tensor. In the absence of competing structured low-rank embedding methods, CP provides a natural baseline. Both methods use embedding dimension $L=10$. As shown in the right panel of Figure~\ref{fig:hic:competitors}, the embeddings produced by \texttt{ZITS} consistently achieve higher ARI than those obtained from CP decomposition, highlighting the benefits of jointly modeling zero inflation, low-rank structure, and chromatin smoothness. A representative visualization of the learned embeddings is shown in Figure~\ref{fig:hic:emb} of Supplementary~\ref{app:sec:plots}.

\begin{figure}[]
    \centering
    \begin{subfigure}[b]{0.36\textwidth}
        \centering
        \includegraphics[width=0.9\textwidth]{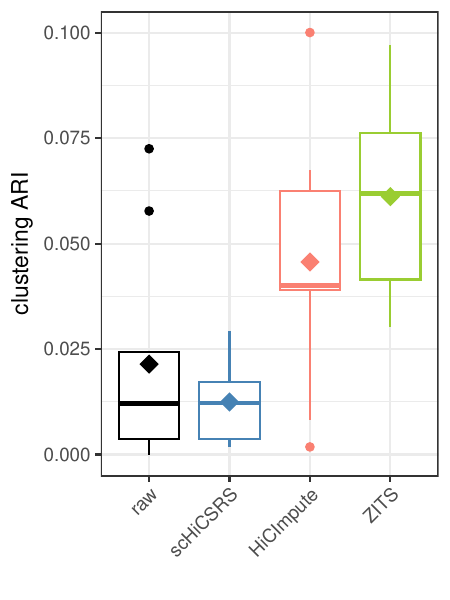}
    \end{subfigure}
    \hfill
    \begin{subfigure}[b]{0.30\textwidth}
        \centering
        \includegraphics[width=0.9\textwidth]{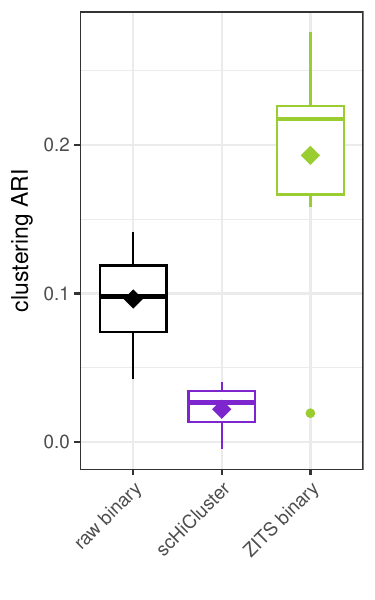}
    \end{subfigure}
    \hfill
    \begin{subfigure}[b]{0.24\textwidth}
        \centering
        \includegraphics[width=0.9\textwidth]{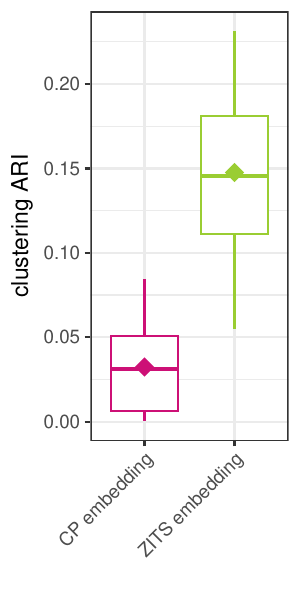}
    \end{subfigure}
    \caption{Comparison of \texttt{ZITS} with competing methods. Diamonds denote the average ARI over 10 random subsamples. Left: methods designed for count-valued single-cell Hi-C data with structural-zero modeling. Middle: methods operating on binarized contact matrices. Right: comparison of cell embeddings.}
    \label{fig:hic:competitors}
\end{figure}

\section{Conclusion}
\label{sec: conclusion}

In this paper, we proposed \texttt{ZITS}, a unified probabilistic framework for analyzing zero-inflated count tensors with latent heterogeneity and smooth structure. \texttt{ZITS} combines a zero-inflated Poisson likelihood, coupled low-rank tensor decomposition, cluster-specific latent embeddings, and spline-based smoothing within a unified estimation framework. This formulation enables simultaneous modeling of excess zeros, latent low-dimensional structure, and heterogeneous populations in sparse tensor-valued count data.

Beyond estimation, we develop a Bayes-optimal classifier for distinguishing structural and technical zeros, providing a principled approach to false-zero detection and imputation. We establish identifiability of the proposed model, derive concentration inequalities for zero-inflated Poisson random variables, and prove consistency of the proposed likelihood estimator. Through extensive simulation studies and analysis of single-cell Hi-C data, we demonstrate that \texttt{ZITS} improves parameter estimation, false-zero detection and cell-type separation compared with existing approaches.

Although motivated by single-cell Hi-C analysis, the proposed framework is broadly applicable to other zero-inflated multiway count data exhibiting latent heterogeneity and smooth structure, including applications in genomics, ecology, and recommender systems.

\section*{Disclosure statement}

The authors report there are no competing interests to declare.

\section*{Data Availability}

The Hi-C data analyzed in this study are publicly available from \citet{Ramani2017sc}. 

\section*{Declaration of Generative AI Use}

The authors acknowledge the use of generative AI for improving the visual appearance of ggplots including coloring and legend, checking grammar and spelling, and enhancing the flow throughout the manuscript.

\putbib[ref]

\end{bibunit}

\newpage

\appendix
\begin{bibunit}


\section*{Supplementary to ``A Probabilistic Model for Zero-Inflated Count Tensors with 
  Structured Latent Representations"}
\section{Additional plots and tables}
\label{app:sec:plots}

\begin{figure}[h]
    \centering
    \begin{subfigure}[b]{0.49\textwidth}
        \centering
        \includegraphics[width=\textwidth]{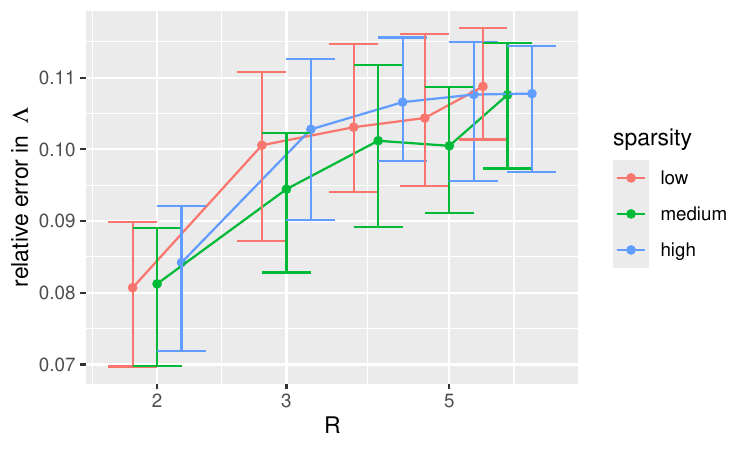}
        \label{fig:simu:R:reL}
    \end{subfigure}
    \hfill
    \begin{subfigure}[b]{0.49\textwidth}
        \centering
        \includegraphics[width=\textwidth]{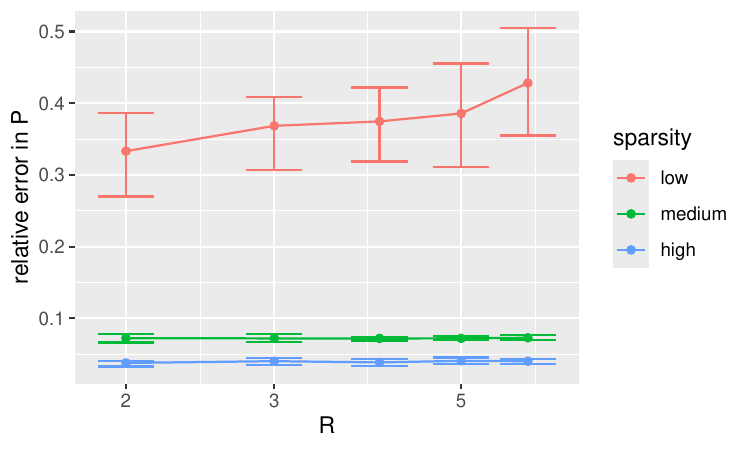}
        \label{fig:simu:R:reP}
    \end{subfigure}
    \caption{Performance of the fitting procedure as the number of clusters $R$ increases, shown for different sparsity levels (controlled by $\mu_\xi$).}
    \label{fig:simu:R:re}
\end{figure}

\begin{figure}[h]
   \centering
   \includegraphics[width=0.7\textwidth]{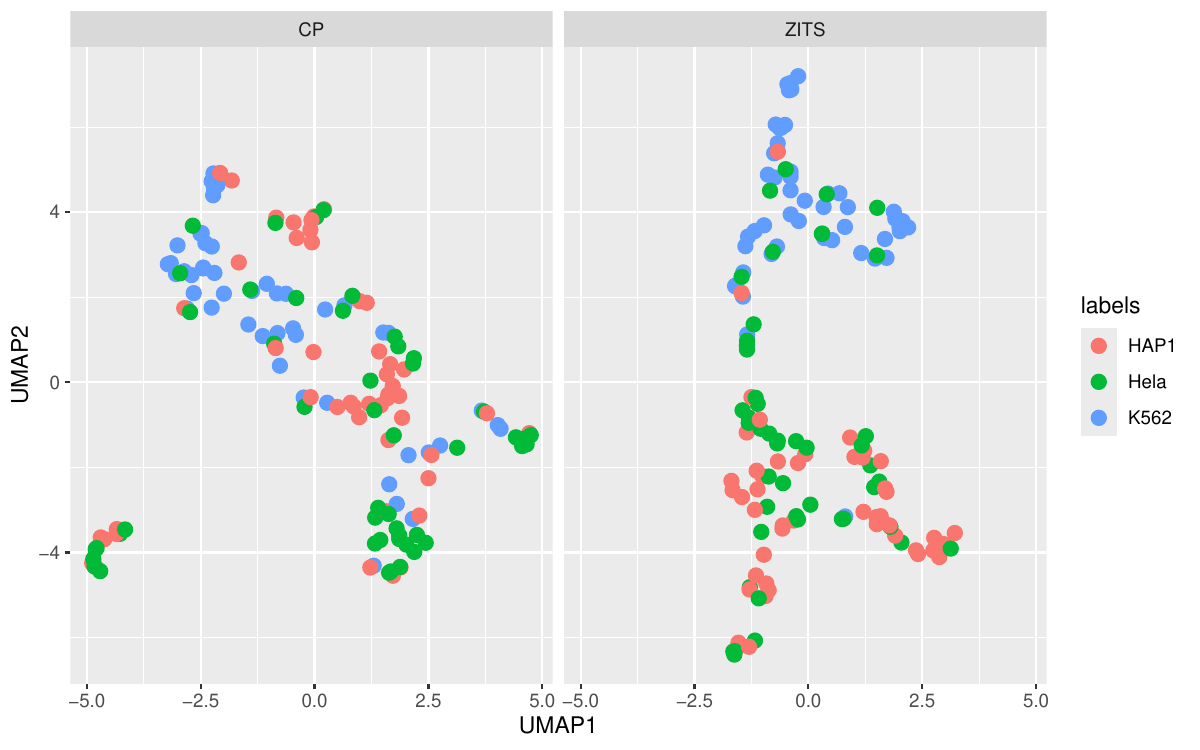}
   \caption{Two-dimensional visualization of the cell embeddings obtained by \texttt{ZITS} and CP decomposition for chromosome~3.}
   \label{fig:hic:emb}
\end{figure}

\section{Derivations}

\subsection{The negative log-likelihood}
\label{sec: likelihood}
For a zero-inflated Poisson random variable $C$ whose probability mass function is 
\begin{equation*}
\mathbb{P}\left(C=c \right)=
\begin{cases}
p + (1-p) e^{-\lambda}, \text { if } c=0; \\
(1 - p)\frac{\lambda^c e^{-\lambda}}{c!}, \text { otherwise}, 
\end{cases}
\end{equation*}
where $1-p = \frac{1}{1+ e^{-\theta}}$ and $\lambda = e^{\eta}$. The negative log-likelihood of $C$ reads 
\begin{align*}
& \quad - \log \left( \left[(1 - p)\frac{\lambda^C e^{-\lambda}}{C!}\right]^{\ind(C \ne 0)} \left[p + (1-p) e^{-\lambda}\right]^{\ind(C = 0)} \right)\\
& = - \ind(C \ne 0)\log \left( \left[(1 - p)\frac{\lambda^C e^{-\lambda}}{C!}\right]\right) - \left( 1- \ind \left(C \ne 0\right)\right)\log\left[p + (1-p) e^{-\lambda}\right]\\
& = \ind(C \ne 0) \log \left( \frac{ C!\left[p + (1-p) e^{-\lambda}\right]}{(1 - p)\lambda^C e^{-\lambda}}\right) - \log\left[p + (1-p) e^{-\lambda}\right]\\
& = \ind(C \ne 0) \left( \log \left( e^{-\lambda} + \frac{p}{1-p}\right) + \lambda - C \log \lambda  +\log C!\right) - \log (1-p) - \log (e^{-\lambda} + \frac{p}{1-p}).
\end{align*}
Under the logit transformation of $1-p$, we have $\frac{p}{1-p} = e^{-\theta}$ and $-\log(1-p) = \log (1+e^{-\theta})$. This leads to
\begin{align*}
& \quad - \log \left( \left[(1 - p)\frac{\lambda^C e^{-\lambda}}{C!}\right]^{\ind(C \ne 0)} \left[p + (1-p) e^{-\lambda}\right]^{\ind(C = 0)} \right)\\
& = \ind(C \ne 0) \left( \log \left( e^{-\lambda} + e^{-\theta}\right) + \lambda - C \log \lambda  +\log C!\right) + \log (1+e^{-\theta}) - \log (e^{-\lambda} + e^{-\theta})\\
& = \ind(C \ne 0) \left( \log \left( 1 + e^{\lambda - \theta}\right) - C \log \lambda  +\log C!\right) + \log (1+e^{-\theta}) + \lambda - \log (1 + e^{\lambda - \theta})\\
& = \ind(C \ne 0) \left( \log \left( 1 + e^{e^{\eta} - \theta}\right) - C \eta  +\log C!\right) + \log (1+e^{-\theta}) + e^\eta - \log (1 + e^{e^\eta - \theta}). 
\end{align*}
Note that the term $\ind(C \ne 0) \log C!$ only depends on the data but not the model parameters. Therefore, minimizing the negative log-likelihood is equivalent to minimizing 
\begin{align*}
 \ind(C \ne 0) \left( \log \left( 1 + e^{ e^{\eta} - \theta}\right) - C \eta\right) + \log (1+e^{-\theta}) + e^\eta - \log (1 + e^{e^\eta - \theta}). 
\end{align*}
It then follows that the negative log-likelihood of the Hi-C tensor, up to an affine transformation is 
\begin{align*}
\mathcal{L}\left(\bm{\eta}, \bTheta; \bm{\mathcal{C}}\right) = & \frac{2}{N(N+1)K} \sum_{i\le j \in [N], k \in[K]} \left[\ind(C_{i, j, k} \ne 0) \left( \log \left( 1 + e^{e^{\eta_{i, j, k} - \theta_{i, j, k}}}\right) - C_{i, j, k} \eta_{i, j, k}\right) \right.\\
& \left. + \log (1+e^{-\theta_{i, j, k}}) + e^{\eta_{i, j, k}} - \log (1 + e^{e^{\eta_{i, j, k} - \theta_{i, j, k}}})\right]. 
\end{align*}

\subsection{Derivatives of the objective function}
\label{appendix: derivative_homogeneous}
Recall that 
\begin{align*}
& \bm{\eta} = \bm{\mathcal{I}} \times_1 \balpha \times_2 \balpha \times_3 \widetilde{\bbeta} =  \bm{\mathcal{I}} \times_1 \bH \bGamma \times_2 \bH \bGamma \times_3 \widetilde{\bbeta}, \text{ and }\\
& \bm{\bTheta} = \bm{\mathcal{I}} \times_1 \balpha \times_2 \balpha \times_3 \widetilde{\bxi} =  \bm{\mathcal{I}} \times_1 \bH \bGamma \times_2 \bH \bGamma \times_3 \widetilde{\bxi}, 
\end{align*}
Therefore, by the chain rule and the symmetry of $\bm{\eta}$ in the first two modes, we have 
\begin{align*}
\frac{\partial \mathcal{L}}{\partial \widetilde{\beta}_{k, l}} = \sum_{i\le j, k^\prime}\frac{\partial \mathcal{L}}{\partial \eta_{i, j, k^\prime }} \frac{\partial \eta_{i, j, k^\prime }}{\partial \widetilde{\beta}_{k, l}} = \sum_{i\le j }\frac{\partial \mathcal{L}}{\partial \eta_{i, j, k }} \frac{\partial \eta_{i, j, k }}{\partial \widetilde{\beta}_{k, l}} = \frac{1}{2}\sum_{i, j }\frac{\partial \mathcal{L}}{\partial \eta_{i, j, k }} \frac{\partial \eta_{i, j, k }}{\partial \widetilde{\beta}_{k, l}} + \frac{1}{2}\sum_{i}\frac{\partial \mathcal{L}}{\partial \eta_{i, i, k }} \frac{\partial \eta_{i, i, k }}{\partial \widetilde{\beta}_{k, l}}. 
\end{align*}
Note that we herein define $\frac{\partial L}{\partial \eta_{i, j, k}} =\frac{\partial L}{\partial \eta_{j, i, k}}$ if $i> j$.

Since $ \frac{\partial \eta_{i, j, k }}{\partial \widetilde{\beta}_{k, l}} = \alpha_{i, l} \alpha_{j, l}$, for any $i, j \in [N]$ and $l\in [L]$, 
\begin{align*}
\frac{\partial \mathcal{L}}{\partial \widetilde{\beta}_{k, l}} &= \frac{1}{2}\sum_{i, j} (\nabla_{\bm{\eta}}\bm{\mathcal{L}})_{i, j, k} \alpha_{i, l} \alpha_{j, l} + \frac{1}{2} \sum_{i} \left(\text{Diag}(\nabla_{\bm{\eta}}\bm{\mathcal{L}})\right)_{i, k}( \balpha * \balpha)_{i, l}\\
& = \frac{1}{2} \left(\nabla_{\bm{\eta}}\bm{\mathcal{L}}\times_1 \balpha^\top \times_2 \balpha^\top\right)_{l, l, k} +\frac{1}{2}\left[\left(\text{Diag}(\nabla_{\bm{\eta}}\bm{\mathcal{L}})\right)^\top (\balpha * \balpha)\right]_{k, l}, 
\end{align*}
showing that 
\begin{align*}
\frac{\partial \mathcal{L}}{\partial \widetilde{\bbeta}} = \frac{1}{2} \text{Diag}\left(\nabla_{\bm{\eta}}\bm{\mathcal{L}}\times_1 \balpha^\top \times_2 \balpha^\top\right)^\top + \frac{1}{2} \left(\text{Diag}(\nabla_{\bm{\eta}}\bm{\mathcal{L}})\right)^\top (\balpha * \balpha). 
\end{align*}
Here, $\text{Diag} : \RR^{n\times n \times m} \rightarrow \RR^{n\times m}$ is a map taking the diagonal of the frontal slices to form the columns of the resulting matrix for any $n$ and $m$, and $*$ stands for Hadamard product. 

Similarly, we have
\begin{align*}
\frac{\partial \mathcal{L}}{\partial \widetilde{\bxi}} = \frac{1}{2} \text{Diag}\left(\nabla_{\bm{\Theta}}\bm{\mathcal{L}}\times_1 \balpha^\top \times_2 \balpha^\top\right)^\top + \frac{1}{2} \left(\text{Diag}(\nabla_{\bm{\Theta}}\bm{\mathcal{L}})\right)^\top (\balpha * \balpha). 
\end{align*}
We next turn to derive the derivative with respect to $\Gamma$. Note that
\begin{align*}
\frac{\partial \mathcal{L}}{\partial \gamma_{q, l}} &= \sum_{i\le j, k} \left(\frac{\partial \mathcal{L}}{\partial \eta_{i, j, k}} \frac{\partial \eta_{i, j, k }}{\partial \gamma_{q, l}} + \frac{\partial \mathcal{L}}{\partial \theta_{i, j, k}} \frac{\partial \theta_{i, j, k }}{\partial \gamma_{q, l}}\right)\\
 &= \frac{1}{2}\sum_{i, j, k} \left(\frac{\partial \mathcal{L}}{\partial \eta_{i, j, k}} \frac{\partial \eta_{i, j, k }}{\partial \gamma_{q, l}} + \frac{\partial \mathcal{L}}{\partial \theta_{i,j, k}} \frac{\partial \theta_{i, j, k }}{\partial \gamma_{q, l}}\right) + \frac{1}{2} \sum_{i, k} \left(\frac{\partial \mathcal{L}}{\partial \eta_{i, i, k}} \frac{\partial \eta_{i, i, k }}{\partial \gamma_{q, l}} + \frac{\partial \mathcal{L}}{\partial \theta_{i, i, k}} \frac{\partial \theta_{i, i, k }}{\partial \gamma_{q, l}}\right). 
\end{align*}
Note that 
\begin{align*}
\frac{\partial \eta_{i, j, k }}{\partial \gamma_{q, l}} & = \frac{\partial \sum_{l^\prime=1}^L \alpha_{i, l^\prime} \alpha_{j, l^\prime}\widetilde{\beta}_{k, l^\prime}}{\partial \gamma_{q, l}} = \frac{\partial \sum_{l^\prime=1}^L (\sum_{q^\prime=1}^Q H_{i, q^\prime} \gamma_{q^\prime, l^\prime}) (\sum_{q^{\prime\prime}=1}^Q H_{j, q^{\prime\prime}} \gamma_{q^{\prime\prime}, l^\prime})\widetilde{\beta}_{k, l^\prime}}{\partial \gamma_{q, l}}\\
& = \frac{\partial \sum_{q^\prime=1}^Q \sum_{q^{\prime\prime}=1}^Q  H_{i, q^\prime} \gamma_{q^\prime, l} H_{j, q^{\prime\prime}} \gamma_{q^{\prime\prime}, l}\widetilde{\beta}_{k, l}}{\partial \gamma_{q, l}} = \left[H_{i, q}(\bH\bGamma)_{j, l} + \bH_{j, q}(\bH \bGamma)_{i, l}\right]\widetilde{\beta}_{k, l}. \\
\end{align*}

It then follows that 
\begin{align*}
\sum_{i, j, k}\frac{\partial \mathcal{L}}{\partial \eta_{i, j, k}} \frac{\partial \eta_{i, j, k }}{\partial \gamma_{q, l}}
& = \sum_{i, j, k}(\nabla_{\bm{\eta}}\bm{\mathcal{L}})_{i, j, k}\left[H_{i, q}(\bH\bGamma)_{j, l} + \bH_{j, q}(\bH \bGamma)_{i, l}\right]\widetilde{\beta}_{k, l}\\
&= \left(\nabla_{\bm{\eta}}\bm{\mathcal{L}} \times_1 \bH^\top \times_2 (\bH \bGamma)^\top \times_3 \widetilde{\bbeta}^\top\right)_{q, l, l} + \left(\nabla_{\bm{\eta}}\bm{\mathcal{L}} \times_1 (\bH \bGamma)^\top\times_2 \bH^\top \times_3 \widetilde{\bbeta}^\top\right)_{l, q, l}\\
&= 2\left(\nabla_{\bm{\eta}}\bm{\mathcal{L}} \times_1 \bH^\top \times_2 \balpha^\top \times_3 \widetilde{\bbeta}^\top\right)_{q, l, l}\\
& = 2\widetilde{\text{Diag}}(\nabla_{\bm{\eta}}\bm{\mathcal{L}} \times_1 \bH^\top \times_2 \balpha^\top \times_3 \widetilde{\bbeta}^\top)_{q, l},  \text{ and }\\
 \sum_{i, k} \frac{\partial \mathcal{L}}{\partial \eta_{i, i, k}} \frac{\partial \eta_{i, i, k }}{\partial \gamma_{q, l}} & = 2\sum_{i, k} \left(\text{Diag}(\nabla_{\bm{\eta}}\bm{\mathcal{L}})\right)_{i, k}H_{i, q}(\bH\bGamma)_{i, l} \widetilde{\beta}_{k, l}= 2\left[\bH^\top \left(\text{Diag}(\nabla_{\bm{\eta}}\bm{\mathcal{L}}) \widetilde{\bbeta} * \bH\bGamma\right)\right]_{q, l}. 
\end{align*}
where $\widetilde{\text{Diag}}: \RR^{n \times m \times m} \rightarrow \RR^{n\times m}$ takes the diagonal of the horizontal slices of the input tensor to form the rows of the resulting matrix for any integers $n$ and $m$. Clearly, similar expression holds if we change $\eta_{i,j, k}$ and $\widetilde{\beta}_{k, l}$ to $\theta_{i,j, k}$ and $\widetilde{\xi}_{k, l}$, respectively. Hence, 
\begin{align*}
\frac{\partial \mathcal{L}}{\partial \bGamma} & =  \widetilde{\text{Diag}}(\nabla_{\bm{\eta}}\bm{\mathcal{L}} \times_1 \bH^\top \times_2 \balpha^\top \times_3 \widetilde{\bbeta}^\top) + \bH^\top \left(\text{Diag}(\nabla_{\bm{\eta}}\bm{\mathcal{L}}) \widetilde{\bbeta} * \balpha\right)\\
&  +\widetilde{\text{Diag}}(\nabla_{\bm{\Theta}}\bm{\mathcal{L}} \times_1 \bH^\top \times_2 \balpha^\top \times_3 \widetilde{\bxi}^\top) + \bH^\top \left(\text{Diag}(\nabla_{\bm{\Theta}}\bm{\mathcal{L}}) \widetilde{\bxi} * \balpha\right).
\end{align*}

Finally, 
\begin{align}
\label{equ: gradient_eta}
& \quad (\nabla_{\bm{\eta}}\bm{\mathcal{L}})_{i, j, k} \nonumber\\
&= \frac{2}{N(N+1)K} \left[\left(\ind(C_{i, j, k}\ne 0)-1\right)
\frac{\exp\{\eta_{i, j, k} + e^{\eta_{i, j, k}} - \theta_{i, j, k}\}}{1+ \exp\{e^{\eta_{i, j, k}} - \theta_{i, j, k}\}} - \ind(C_{i, j, k}\ne 0)C_{i,j, k} + e^{\eta_{i, j, k}}\right] \nonumber\\
&= \frac{2}{N(N+1)K} \left[-\ind(C_{i, j, k}= 0)
\frac{\exp\{\eta_{i, j, k}\}}{1+ \exp\{ \theta_{i, j, k}-e^{\eta_{i, j, k}}\}} - C_{i,j, k} + e^{\eta_{i, j, k}}\right], 
\end{align}
and 
\begin{align}
\label{equ: gradient_theta}
 (\nabla_{\bm{\Theta}}\bm{\mathcal{L}})_{i, j, k} &=\frac{2}{N(N+1)K} \left[\left(1- \ind(C_{i, j, k}\ne 0)\right)
\frac{\exp\{e^{\eta_{i, j, k}} - \theta_{i, j, k}\}}{1+ \exp\{e^{\eta_{i, j, k}} - \theta_{i, j, k}\}} - \frac{1}{1 + e^{\theta_{i, j, k}}}\right] \nonumber\\
&=\frac{2}{N(N+1)K} \left[\ind(C_{i, j, k}= 0)
\frac{1}{1+ \exp\{\theta_{i, j, k}-e^{\eta_{i, j, k}}\}} - \frac{1}{1 + e^{\theta_{i, j, k}}}\right], 
\end{align}
for any $i, j \in [N]$ and $k \in [K]$.

\section{Model identifiability}

\textbf{Proof of Proposition \ref{prop: ZIP_identifiability}: }
Since $\{(p_{i, j, k}, \lambda_{i, j,k}): i\le j \in [n] \}$ and  $\{(p^\prime_{i, j, k}, \lambda^\prime_{i, j,k}): i\le j \in [n], k\in [K]\}$ parametrize the same zero-inflated Poisson distribution, we have
\begin{align*}
p_{i, j, k} + (1 - p_{i, j, k})e^{-\lambda_{i, j, k}} &= p_{i, j, k}^\prime + ( 1- p_{i, j, k}^\prime) e^{- \lambda_{i, j, k}^\prime} \text{ and }\\
(1 - p_{i, j, k}) \lambda_{i, j, k}^c e^{ - \lambda_{i, j, k}}/c! &= (1 - p_{i, j, k}^\prime) (\lambda_{i, j, k}^\prime)^c e^{ - \lambda_{i, j, k}^\prime}/c!
\end{align*}
for $c = 1,  2, \ldots$. When $c = 1$ and $2$, we have $(1 - p_{i, j, k}) \lambda_{i, j, k} e^{- \lambda_{i, j, k}} = (1 - p_{i, j, k}^\prime) \lambda_{i, j, k}^\prime e^{- \lambda_{i, j, k}^\prime}$ and $(1 - p_{i, j, k}) \lambda_{i, j, k}^2 e^{- \lambda_{i, j, k}}/2 = (1 - p_{i, j, k}^\prime) (\lambda_{i, j, k}^\prime)^2 e^{- \lambda_{i, j, k}^\prime}/2$. Taking the ratio of the above two equations, we obtain $\lambda_{i, j, k} = \lambda^\prime_{i, j, k}$ since $p_{i, j, k} \ne 1$ and $\lambda_{i, j, k} \ne 0$.  Plugging $\lambda_{i, j, k} = \lambda^\prime_{i, j, k}$ back to the equation $(1 - p_{i, j, k}) \lambda_{i, j, k} e^{- \lambda_{i, j, k}} = (1 - p_{i, j, k}^\prime) \lambda_{i, j, k}^\prime e^{- \lambda_{i, j, k}^\prime}$ yields $p_{i,j, k} = p^\prime_{i, j, k}$, for any $i\le j \in [n]$ and $k \in [K]$. \qed

\noindent \textbf{Proof of Proposition \ref{prop: identifiablity}:} We first define a new tensor $\bm{\mathcal{T}} \in \RR^{N \times N \times 2K}$, which concatenates $\bm{\eta}$ and $\bTheta$ along the third mode such that 
\begin{align*}
T_{i, j, k} = \begin{cases}
\eta_{i,j, k} & \text{ if } k\le K, \text{ and }\\
\theta_{i, j, k -K} &\text{otherwise}. 
\end{cases}
\end{align*}
The doubly tensor decomposition $\bm{\eta} = \bm{\mathcal{I}}\times_1 \balpha\times_2 \balpha\times_3 \widetilde{\bbeta}$ and $\bTheta = \bm{\mathcal{I}}\times_1{\balpha}\times_2 \balpha \times_3 \widetilde{\bxi}$ can be rewritten as 
\begin{align}
\label{equ: T_decom}
\bm{\mathcal{T}} = \bm{\mathcal{I}} \times_1 \bm{\balpha} \times_2 \bm{\balpha}\times_3 (\widetilde{\bbeta}^\top,\widetilde{\bxi}^\top)^\top =  \bm{\mathcal{I}} \times_1 \bH\bGamma \times_2 \bH\bGamma \times_3 (\bI_2 \otimes \bZ)  (\bbeta^\top, \bxi^\top)^\top.
\end{align}
Therefore, the doubly tensor decomposition is identifiable if and only if the decomposition of $\bm{\mathcal{T}}$ in (\ref{equ: T_decom}) is identifiable. By the Kruskal's Theorem \citep{sidiropoulos2000uniqueness}, (\ref{equ: T_decom}) will be identifiable up to column permutations and scalar multiplications of $\balpha$ and $(\widetilde{\bbeta}^\top,  \widetilde{\bxi}^\top)^\top$ if $2\kappa_{\balpha} + \kappa_{( \widetilde{\bbeta}^\top,  \widetilde{\bxi}^\top)^\top } \ge 2L + 2$, or simply 
$\kappa_{\balpha} + \frac{1}{2}\kappa_{(\bbeta^\top, \bxi^\top)^\top } \ge L + 1$, where we use the fact that $\kappa_{( \widetilde{\bbeta}^\top,  \widetilde{\bxi}^\top)^\top } = \kappa_{(\bbeta^\top, \bxi^\top)^\top }$.

Moreover, since $(\widetilde{\bbeta}^\top,\widetilde{\bxi}^\top)^\top = (\bI_2 \otimes \bZ)  (\bbeta^\top, \bxi^\top)^\top$, the $l$-th column of $\left(\widetilde{\bbeta}^\top, \widetilde{\bxi}^\top\right)^\top$ comes from the $l$-th column of $(\bbeta^\top, \bxi^\top)^\top$ by adding certain repeated coordinates according to $\bZ$, for $l\in [L]$. Therefore, imposing any permutation on the columns of $(\widetilde{\bbeta}^\top,\widetilde{\bxi}^\top)^\top$ is equivalent to imposing the same column permutation on  $(\bbeta^\top, \bxi^\top)^\top$.
When we permute the columns of the matrix $(\bbeta^\top,  \bxi^\top)^\top$, we need to perform the same column permutation on $\balpha$ or equivalently on $\bGamma$ simultaneously. Clearly, this does not affect the cluster assignments. 

For the scaling multiplication issue, we can multiple the $l$-th column of $\balpha$ by $m$ and the $l$-th column of $(\widetilde{\bbeta}^\top,  \widetilde{\bxi}^\top)^\top$ by $m^{-2}$, which is equivalent to multiplying the $l$-th column of $(\bbeta^\top, \bxi^\top)^\top$ by $m^{-2}$, for any $m\ne 0$, and $l \in [L]$. This shows that scaler multiplication also does not change the cluster structure of the cells. In essence, we can extract $\bZ$ from $(\widetilde{\bbeta}^\top,  \widetilde{\bxi}^\top)^\top$, up to column permutation, or relabeling,  as long as the rows of $(\bbeta, \bxi)$ are distinct. Once we extract $\bZ$ from $(\widetilde{\bbeta}^\top,  \widetilde{\bxi}^\top)^\top$, we also obtain $\bbeta$ and $\bxi$ up to row permutation, column permutation, and column scaling.

Finally, as the basis functions $h_1(t), \ldots, h_Q(t)$ are orthogonal, the columns of $\bH$ are guaranteed to be orthogonal when $N$ is large \citep{han2024guaranteed}. Once $\bH$ has full column rank, there will be a unique $\bGamma$ such that $\balpha = \bH \bGamma$, showing the identifiability of $\bGamma$ up to column permutations and scalings.  \qed


\section{More on zero-inflated Poisson distribution}
\label{sec: appendix_ZIP_more}

\begin{definition}
\label{def: Zero-inflated}
A counting random variable $C$ follows zero-inflated Poisson distribution with parameter $(p, \lambda)$ if 
\begin{align*}
\PP(C = c) = \begin{cases}
p + (1-p) e^{-\lambda}, \text{ if } c=0, \text{ and }\\
(1 - p) \frac{\lambda^c e^{-\lambda}}{c!}, \text{ if } c \in \NN^+. 
\end{cases}
\end{align*}
\end{definition}
Since we can decompose $C = B \widetilde{C}$, where $B$ is a Bernoulli random variable with successful probability $1-p$ and $\widetilde{C}$ independent with $B$ is a Poisson random variable with intensity $\lambda$. The expectation of $C$ is 
\begin{align*}
\EE C  = \EE B \widetilde{C} = \EE B \EE\widetilde{C} = (1 - p)\lambda. 
\end{align*}
By the law of total variance, the variance of $C$ is 
\begin{align*}
\Var(C) &= \Var(B\widetilde{C})\\
&= \EE \left[\Var(B\widetilde{C}|\widetilde{C})\right] + \Var\left[\EE \left(B\widetilde{C}|\widetilde{C}\right)\right]\\
&= p(1-p)\EE (\widetilde{C}^2) + \Var \left[ (1-p) \widetilde{C}\right]\\
& = p(1-p)\lambda(1+\lambda) + (1-p)^2\lambda\\
& = \lambda (1-p) (p\lambda + 1). 
\end{align*}
Clearly, the variance of $C$ is greater than $\lambda$ if $\lambda > \frac{1}{1-p}$ and smaller than $\lambda$ if $\lambda < \frac{1}{1 -p}$.

\noindent \textbf{Proof of Proposition \ref{prop: psi_1}:}
Computing the absolute moment generating function for ${X\sim \text{Poisson}(\lambda)}$ 
\begin{align*}
\EE e^{|X|/t} = \sum_{k =0}^{+\infty}  e^{k/t} \frac{\lambda^k e^{-\lambda}}{k!} = \exp\left(\lambda (e^{1/t} - 1)\right), 
\end{align*}
which is decreasing with $t \in (0, \infty)$. Setting $\EE e^{|X|/t} =2$ yields that 
\begin{align*}
t = \left(\log \left(\frac{1}{\lambda}\log 2 + 1\right)\right)^{-1},  
\end{align*}
which is the exact $\psi_1$-norm of $X$. 

Similarly, for $C \sim ZIP(p, \lambda)$,
\begin{align*}
\EE e^{|C|/t} &= p + (1-p) e^{-\lambda} +(1-p) \sum_{k =1}^{+\infty}  e^{k/t} \frac{\lambda^k e^{-\lambda}}{k!}\\
&= p +(1-p) \sum_{k =0}^{+\infty}  e^{k/t} \frac{\lambda^k e^{-\lambda}}{k!}\\
& = p + (1-p)\exp\left(\lambda (e^{1/t} - 1)\right),
\end{align*}
which is also decreasing with $t\in (0, +\infty)$. Setting $\EE e^{|C|/t}= 2$ yields that
\begin{align*}
t = \left(\log\left(\frac{1}{\lambda}\log \frac{2-p}{1-p} + 1\right)\right)^{-1}, 
\end{align*}
which is the exact $\psi_1$-norm of $C$.  Moreover, since $\frac{2-p}{1-p}> 2$ for $p\in (0, 1)$, we have $\|C\|_{\psi_1} < \|X\|_{\psi_1}$ for $p \in (0, 1)$. \qed

In the following, we introduce the Hurdle Poisson distribution and its relationship with zero-inflated Poisson distribution. This result will be used to derive the consistency theory of the proposed estimator of ZITS. 
\begin{definition}
\label{def: Hurdel}
A counting random variable $C$ follows Hurdle Poisson distribution with parameter $(p, \lambda)$, $p \ne 1$, if 
\begin{align*}
\PP(C = c) = \begin{cases}
p, \text{ if } c=0, \text{ and }\\
\frac{1 - p}{1 - e^{-\lambda}} \frac{\lambda^c e^{-\lambda}}{c!}, \text{ if } c \in \NN^+. 
\end{cases}
\end{align*}
\end{definition}

According to the above two definitions, a Hurdle Poisson distribution with parameter $(p, \lambda)$ is a zero-inflated Poisson distribution with parameter $\left((p - e^{-\lambda})/(1 - e^{-\lambda}), \lambda\right)$, provided that $e^{-\lambda} \le p < 1 $. This is because 
\begin{align*}
p = \frac{p - e^{-\lambda}}{1 - e^{-\lambda}} + \left(1 - \frac{p - e^{-\lambda}}{1 - e^{-\lambda}}\right) e^{-\lambda}, \text{ and } \frac{1-p}{1 - e^{-\lambda}} = 1 - \frac{p - e^{-\lambda}}{1 - e^{-\lambda}}. 
\end{align*}
Conversely, a zero-inflated Poisson distribution with parameter $(p, \lambda)$ is a Hurdle Poisson distribution with parameter $\left(p + (1-p)e^{-\lambda}, \lambda\right)$. This is because 
\begin{align*}
1 - p = \frac{1 - \left(p + (1-p)e^{-\lambda}\right)}{1 - e^{-\lambda}}. 
\end{align*}

Because $0< e^{-\lambda} \le p + (1-p)e^{-\lambda}< p+(1-p) = 1$, any zero-inflated Poisson distribution can be reparametrized as Hurdle Poisson distribution. However, not every Hurdle Poisson distribution can be reparametrized as zero-inflated Poisson distribution, implying that the family of zero-inflated Poisson distributions is a special subset of the family of Hurdle Poisson distributions.

\noindent \textbf{Proof of Theorem \ref{thm: concentration}} The centered moment generating function of $C$ satisfies 
\begin{align*}
\EE e^{t \left(C - (1-p)\lambda \right)} &= \sum_{c = 1}^{+\infty} (1 - p) \frac{\lambda^c e^{-\lambda}}{c!} e^{t\left(c - (1-p)\lambda\right)} + \left(p + (1 - p) e^{-\lambda}\right) e^{- t(1 - p)\lambda}\\ 
& = (1 - p) e^{p \lambda t + \lambda(e^t - t - 1)} + pe^{- (1 - p)\lambda t}\\
&\le e^{\lambda(e^t - t- 1)}\left[ (1- p) e^{p\lambda t} + p e^{-(1-p) \lambda t}\right]\\
& = e^{\lambda(e^t - t- 1)} \EE e^{\lambda t[B - (1-p)]}, 
\end{align*}
for any $t> 0$, where the inequality comes from the fact that $e^t - t - 1> 0$, and $B$ is a Bernoulli random variable whose successful probability is $1-p$. Note that 
\begin{align*}
\EE e^{\lambda t [B - (1-p)]} & = 1+ \sum_{k= 2}^{+\infty} \frac{(\lambda t)^k}{k!}  \EE [B - (1-p)]^k\\
& = 1 +  \sum_{k= 2}^{+\infty}  \frac{(\lambda t)^k}{k!}\{(1-p) p^k + p [-(1-p)]^k\}\\
& \le 1 +\frac{p(1-p)\lambda^2 t^2}{2}+ 2p (1-p) \sum_{k = 3}^{+\infty}  \frac{(\lambda t)^k}{6}[p \vee (1-p)]^{k-1}\\
& \le 1 +\frac{p(1-p)\lambda^2 t^2}{2}+\frac{p (1-p)\lambda^3 t^3 [p\vee (1-p)]^2}{2(1 - [p \vee (1-p)] \lambda t)}\\
& \le   1+ \frac{p(1-p)\lambda^2 t^2}{2(1 - [p \vee (1-p)]\lambda t)}\\
& \le e^{\frac{p(1-p)\lambda^2 t^2}{2(1 - [p \vee (1-p)]\lambda t)}},
\end{align*}
where $p\vee (1-p): = \max\{p, 1-p\}$ and the last three inequalities hold when $t< \frac{1}{[p\vee (1-p)] \lambda}$. Moreover, 
\begin{align*}
e^t - t - 1 =\sum_{k =2}^{+\infty} \frac{t^k}{k!} \le \frac{t^2 }{2(1 - t)}, 
\end{align*}
when $t<1$, leading to 
\begin{align*}
\log \EE e^{t \left(C - (1-p)\lambda \right)} &\le \frac{\lambda t^2}{2(1 - t)} + \frac{p(1-p)\lambda^2 t^2}{2(1 - [p \vee (1-p)]\lambda t)}\\
& \le \frac{\lambda [1 + p(1-p)\lambda] t^2}{2(1 - \{1\vee [p \vee (1-p)] \lambda\}t)}\\
& \le \frac{\lambda [1 + p(1-p)\lambda] t^2}{2\left(1 - \max\{1, \lambda\}t\right)}, 
\end{align*}
where the first two inequalities hold for any $0<t< \frac{1}{\max\{1, [p\vee(1-p)]\lambda\}}$ and the last inequality holds for any $0<t< \frac{1}{\max\{1, \lambda\}}$. It is not hard to verify that the above display also holds when $ - \frac{1}{\max\{1, \lambda\}}<t\le  0$. This implies that $C$ is a sub-exponential random variable with parameter $\left(\lambda [1 + p(1-p)\lambda],  1\vee [p \vee (1-p)] \lambda \right)$, or simply a sub-exponential random variable with parameter $\left(\lambda[1 + p(1-p)\lambda], \max\{1, \lambda\}\right)$. 

 Since $C_i\sim \text{subE} \left( \lambda_i[1+p_i(1-p_i) \lambda_i], \max\{1, \lambda_i\}\right)$, we have
\begin{align*}
\sum_{i = 1}^n a_i C_i \sim \text{subE}\left(\sum_{i = 1}^n a_i^2 \lambda_i[1+p_i(1-p_i) \lambda_i], \max_i \max\{|a_i|, |a_i|\lambda_i \}      \right),
\end{align*}
where subE$(\cdot, \cdot)$ here is a shorthand notation for sub-exponential random variable with the corresponding parameters. 
By the Bernstein's inequality, 
\begin{align*}
& \quad \PP\left(\frac{1}{n} \sum_{i = 1}^n a_i [C_i - (1-p_i)\lambda_i]\ge M  \right) \\
& = \PP\left( \sum_{i = 1}^n a_i [C_i - (1 - p_i) \lambda_i]\ge nM  \right)\\
& \le \exp\left( - \frac{n^2M^2}{2\left(\sum_{i = 1}^n a_i^2 \lambda_i[1+p_i(1-p_i) \lambda_i] +nM\max_i \max\{|a_i|, |a_i|\lambda_i \} \right  )}\right). \qed
\end{align*}

\section{Technical lemmas}
\label{sec: appdendix_D_lemmas}
\begin{lemma}
\label{lemma: KL_Hellinger_Poisson}
For any $a, b\ge 0$, we have 
\begin{align*}
a - b - b\log \frac{a}{b} \ge (\sqrt{a} - \sqrt{b})^2. 
\end{align*}
\end{lemma}
\noindent \textbf{Proof of Lemma \ref{lemma: KL_Hellinger_Poisson}.} Since $\log(x+1) \le x $ for any $x>-1$, we have
\begin{align*}
a - b - b\log\frac{a}{b} & = a - b - 2b \log \left(\frac{\sqrt{a} - \sqrt{b}}{\sqrt{b}} + 1\right)\\
&\ge a - b - 2b \cdot \frac{\sqrt{a} - \sqrt{b}}{\sqrt{b}}\\
& = (\sqrt{a} - \sqrt{b})^2. \qed
\end{align*}

\begin{lemma}
\label{lemma: KL_Hellinger}
For any $x, y > 0$, we have 
\begin{align*}
\frac{x}{1-e^{-x}} \log \frac{x}{y} + \log \frac{e^y - 1}{e^x - 1}\ge \left(\sqrt{\frac{x}{1 - e^{-x}}} - \sqrt{\frac{y}{1 - e^{-y}}}\right)^2.
\end{align*}
\end{lemma}
\noindent\textbf{Proof of Lemma \ref{lemma: KL_Hellinger}.} According to Lemma \ref{lemma: KL_Hellinger_Poisson}, we have 
\begin{align*}
\left(\sqrt{\frac{x}{1 - e^{-x}}} - \sqrt{\frac{y}{1 - e^{-y}}}\right)^2 \le \frac{y}{1 - e^{-y}}  - \frac{x}{1 - e^{-x}} - \frac{x}{1 - e^{-x}} \left(\log\frac{y}{x} + \log\frac{1 - e^{-x}}{1 - e^{-y}}\right). 
\end{align*}
It then suffices to show that 
\begin{align*}
\frac{y}{1 - e^{-y}}  - \frac{x}{1 - e^{-x}} - \frac{x}{1 - e^{-x}}\log\frac{1 - e^{-x}}{1 - e^{-y}}\le \log \frac{e^y - 1}{e^x - 1}, 
\end{align*}
or equivalently the function $f_x(y) \ge 0$ for any $x>0$ and $y >0$, where
\begin{align*}
f_x(y) & = \log(e^y - 1) - \frac{y}{1 - e^{-y}}  - \frac{x}{1 - e^{-x}} \log (1 - e^{-y})\\
& - \log(e^x -1) +  \frac{x}{1 - e^{-x}} +  \frac{x}{1 - e^{-x}}\log (1 - e^{-x}).
\end{align*}
In fact, 
\begin{align*}
f^\prime_x(y) &= \frac{e^y}{e^y - 1} - \frac{1 - e^{-y} - ye^{-y}}{(1 - e^{-y})^2} -  \frac{x}{1 - e^{-x}}\cdot \frac{e^{-y}}{1 - e^{-y}}\\
& = \frac{1-e^{-y} - (1 - e^{-y}) + ye^{-y} -   \frac{x}{1 - e^{-x}}\cdot e^{-y}(1 -e^{-y})    }{(1 -e^{-y})^2}\\
& = \frac{ e^{-y}\left(\frac{y}{1-e^{-y}} -   \frac{x}{1 - e^{-x}}\right) }{1 -e^{-y}}. 
\end{align*}
Clearly, $e^{-y}>0$ and $1 - e^{-y}>0$ since $y>0$, and the sign of $f_x^\prime(y)$ depends on $y/(1-e^{-y}) - x/(1-e^{-x})$ only. Let
\begin{align*}
g(y) = \frac{y}{1-e^{-y}} = \frac{ye^y}{e^y - 1}, 
\end{align*}
for $y>0$. The first order derivative of $g(y)$ is 
\begin{align*}
g^\prime(y) = \frac{(e^y + ye^y)(e^y -1) - ye^{2y}}{(e^y-1)^2} = \frac{e^y(e^y-y - 1)}{(e^y -1)^2} = \frac{e^y\sum_{k=2}^{+\infty} y^k}{(e^y -1)^2} >0,
\end{align*}
showing that $g(y)$ is a strictly increasing function for $y\in (0,+\infty)$. Consequently, $f_x^\prime(y)< 0 $ for $y\in (0, x)$, $f_x^\prime(y)= 0 $ for $y=x$, and $f_x^\prime(y)> 0 $ for $y\in (x, +\infty)$. This implies that the continuous function $f_x(y)$ strictly monotone decreases for $y\in (0, x]$ and strictly monotone increases for $y\in[x, +\infty)$. Therefore, 
$f_x(y)\ge f_x(x) = 0$, for any $x>0$ and $y>0$. \qed

\begin{lemma}
\label{lemma: inverse_derivative}
For any $x, y >0$, we have
\begin{align*}
\vert x - y\vert \le 2\vert g(x) - g(y) \vert, 
\end{align*}
where $g(y) = ye^y/(e^y -1)$ is the function defined in the proof of Lemma \ref{lemma: KL_Hellinger}. 
\end{lemma}
\noindent\textbf{Proof of Lemma \ref{lemma: inverse_derivative}. } In the proof of Lemma \ref{lemma: KL_Hellinger}, we have 
\begin{align*}
g^\prime(y) =  \frac{e^y(e^y-y - 1)}{(e^y -1)^2} = \frac{e^y - y - 1}{e^y -2 + e^{-y}}. 
\end{align*}
It then follows that the second order derivative of $g(y)$ is 
\begin{align*}
g^{\prime \prime}(y) &= \frac{(e^y - 1)(e^y -2 + e^{-y}) - (e^y - y - 1)(e^y -e^{-y})}{(e^y -2 + e^{-y})^2}\\
&= \frac{2(e^y - 1)(e^{-y} - 1) + y (e^y -e^{-y})}{(e^y -2 + e^{-y})^2}. 
\end{align*}
Let $h(y) = 2(e^y - 1)(e^{-y} - 1) + y (e^y -e^{-y}) = (y -2) e^{y} - (y+2)e^{-y} + 4$ be the numerator of $g^{\prime \prime}(y)$.  The first and second order derivatives of $h(y)$ are 
\begin{align*}
h^\prime(y) & = (y-1)e^y + (y +1)e^{-y}, \text{ and } \quad h^{\prime \prime}(y) = y(e^{y} - e^{-y}), 
\end{align*}
respectively. Note that $h^{\prime \prime}(y)$ is continuous and $h^{\prime \prime}(y) > 0$ for $y>0$, which implies that $h^{\prime}(y)$ is a strictly monotone increasing function in $[0, +\infty)$. Therefore, $h^\prime(y) > h^\prime(0)= 0$ for any $y>0$, showing that the continuous function $h(y)$ is a strictly monotone increasing function for $y\in [0, +\infty)$. This leads to $h(y)> h(0) = 0$ for any $y>0$.  Hence, we have $g^{\prime \prime}(y)>0$ for any $y>0$, which suggests that $g^{\prime}(y)$ is a strictly monotone increasing function for $y\in (0, +\infty)$. Moreover, by L'H$\widehat{o}$pital's rule, we have
\begin{align*}
& g^\prime(y) > \lim_{y\rightarrow 0^+} \frac{e^y - y - 1}{e^y -2 + e^{-y}} = \lim_{y\rightarrow 0+} \frac{e^y - 1}{e^y  - e^{-y}} = \lim_{y\rightarrow 0+} \frac{e^y}{e^y  + e^{-y}} = \frac{1}{2}, \text{ and }\\
& g^{\prime}(y) <  \lim_{y\rightarrow +\infty} \frac{e^y - y - 1}{e^y -2 + e^{-y}} = \lim_{y\rightarrow +\infty} \frac{e^y}{e^y  + e^{-y}} = \lim_{y\rightarrow +\infty} \frac{1}{1  + e^{-2y}} = 1. 
\end{align*}
Since $g^\prime(y)>1/2$, $z = g(y)$ is a strictly monotone increasing function in $(0, +\infty)$, and its inverse $y= g^{-1}(z)$ exists, which is also a strictly monotone increasing function for $z \in (1, +\infty)$ since $\lim_{y\rightarrow 0^+} g(y) = 1$ and $\lim_{y \rightarrow + \infty}g(y) = +\infty$. Moreover, $g^\prime(y) \in (1/2, 1)$ implies that $(g^{-1})^\prime(z) \in (1, 2)$, showing that $g^{-1}(z)$ is $2$-Lipschitz continues. It immediately follows that 
\begin{align*}
\vert x- y\vert &= \vert g^{-1}\left(g(x)\right) - g^{-1}\left(g(y)\right)\vert \le 2 \vert g(x) - g(y) \vert. \qed
\end{align*}

\begin{lemma}
\label{lemma: KL_square}
For any $\kappa>1$, suppose $0< x, y \le \kappa$,  we have  
\begin{align*}
(x-y)^2 \le 16 (\kappa+1) \left( \frac{x}{1-e^{-x}} \log \frac{x}{y} + \log \frac{e^y - 1}{e^x - 1}\right).
\end{align*}
\end{lemma}
\noindent\textbf{Proof of Lemma \ref{lemma: KL_square}. } For the function $z = g(y) = y/(1 - e^{-y})$, define $g(0) = 1$ such that $g$ will be a right continues function at the neighborhood of $0$. For any $y>0$, by the mean value theorem, we have
\begin{align*}
g(y) - g(0) = g^\prime(u) (y - 0), 
\end{align*}
for some $u$ between $0$ and $y$. Since $g^\prime(u) \in (\frac{1}{2}, 1)$, we have
\begin{align*}
\frac{1}{2}y + 1< g(y) = g^\prime(u) y + 1< y +1, 
\end{align*}
for any $y>0$. 

It then follows from the result in Lemma \ref{lemma: inverse_derivative} that, 
\begin{align*}
4\vert\sqrt{g(x)} - \sqrt{g(y)}\vert = \frac{4\vert g(x) - g(y)\vert}{\sqrt{g(x)} + \sqrt{g(y)}} \ge \frac{\vert x - y \vert}{\sqrt{\kappa+1}}. 
\end{align*}
Applying the result of Lemma \ref{lemma: KL_Hellinger} yields that 
\begin{align*}
(x-y)^2  \le 16(\kappa+1)\left(\sqrt{g(x)} - \sqrt{g(y)}\right)^2\le 16 (\kappa+1) \left( \frac{x}{1-e^{-x}} \log \frac{x}{y} + \log \frac{e^y - 1}{e^x - 1}\right). \qed 
\end{align*}

\begin{lemma}
    \label{lemma: KL_divergence} 
    For two Hurdle Poisson distributions with parameters $(p, \lambda)$ and $(\widetilde{p}, \widetilde{\lambda})$, their KL-divergence reads 
    \begin{align*}
    KL_{\text{Hurdle}}(p, \lambda\| \widetilde{p}, \widetilde{\lambda}) = KL_{\text{Bernoulli}}(p\| \widetilde{p}) + (1-p) \left[\frac{\lambda}{1 - e^{-\lambda}} \log \frac{\lambda}{\widetilde{\lambda}} + \log \frac{e^{\widetilde{\lambda}}-1}{e^\lambda -1}\right], 
    \end{align*}
where 
\begin{align*}
KL_{\text{Bernoulli}}(p\| \widetilde{p}) =p \log \frac{p}{\widetilde{p}} + (1-p)\log \frac{1-p}{1-\widetilde{p}}
\end{align*}
is the KL-divergence of two Bernoulli distributions with parameter $p$ and $\widetilde{p}$, respectively. 
\end{lemma}

\noindent \textbf{Proof of Lemma \ref{lemma: KL_divergence}. } Denote $p(c)$ and $q(c)$ as the probability mass functions of the Hurdle Poisson distributions with parameters $(p, \lambda)$ and $(\widetilde{p}, \widetilde{\lambda})$, respectively. By Definition \ref{def: Hurdel}, we have
\begin{align*}
\frac{p(0)}{q(0)} =\frac{p}{\widetilde{p}}  \text{ and } \frac{p(c)}{q(c)} & = \frac{1-p}{1-\widetilde{p}} \cdot \left(\frac{\lambda}{\widetilde{\lambda}}\right)^c \cdot \frac{1-e^{-\widetilde{\lambda}}}{1 - e^{-\lambda}}\cdot \frac{e^{-\lambda}}{e^{-\widetilde{\lambda}}}\\
&= \frac{1-p}{1-\widetilde{p}} \cdot \left(\frac{\lambda}{\widetilde{\lambda}}\right)^c \cdot \frac{e^\lambda e^{\widetilde{\lambda}} -e^{\lambda}}{e^\lambda e^{\widetilde{\lambda}}  - e^{\widetilde{\lambda}}}\cdot \frac{e^{-\lambda}}{e^{-\widetilde{\lambda}}}\\
&= \frac{1-p}{1-\widetilde{p}} \cdot \left(\frac{\lambda}{\widetilde{\lambda}}\right)^c \cdot \frac{ e^{\widetilde{\lambda}}-1} {e^\lambda-1}, 
\end{align*}
for $c\in \NN^+$. It then follows that for any $c\in \NN^+$, we have
\begin{align*}
\log \frac{p(c)}{q(c)} = \log  \frac{1-p}{1-\widetilde{p}} + c \log \frac{\lambda}{\widetilde{\lambda}} + \log \frac{ e^{\widetilde{\lambda}}-1} {e^\lambda-1}.  
\end{align*}
By the definition of KL-divergence, 
\begin{align*}
 KL_{\text{Hurdle}}(p, \lambda\| \widetilde{p}, \widetilde{\lambda})& = p \log \frac{p}{\widetilde{p}} + \sum_{c=1}^{+\infty} \frac{1 - p}{1 - e^{-\lambda}} \frac{\lambda^c e^{-\lambda}}{c!}\log \frac{p(c)}{q(c)} \\
 & = p \log \frac{p}{\widetilde{p}} + (1-p)\left(\log  \frac{1-p}{1-\widetilde{p}}+  \frac{\lambda}{1-e^{-\lambda}}\log \frac{\lambda}{\widetilde{\lambda}}     +   \log \frac{ e^{\widetilde{\lambda}}-1} {e^\lambda-1} \right)\\
 &= KL_{\text{Bernoulli}}(p\| \widetilde{p}) + (1-p) \left[\frac{\lambda}{1 - e^{-\lambda}} \log \frac{\lambda}{\widetilde{\lambda}} + \log \frac{e^{\widetilde{\lambda}}-1}{e^\lambda -1}\right].\qed
\end{align*}

\begin{lemma}
\label{lemma: KL_Hurdle_square_loss}
Suppose that $0< s < p, \widetilde{p}<S< 1$, $0< \lambda, \widetilde{\lambda}\le \lambda_{\max}$ for some constants $0< s<S<1$, we have
\begin{align*}
 KL_{\text{Hurdle}}(p, \lambda\| \widetilde{p}, \widetilde{\lambda})\ge \frac{1-s+S}{4S(1-s)}(p -\widetilde{p})^2 + \frac{(1-S)}{16(\lambda_{\max}+1)}(\lambda - \widetilde{\lambda})^2,
\end{align*}
or simply
\begin{align*}
 KL_{\text{Hurdle}}(p, \lambda\| \widetilde{p}, \widetilde{\lambda})\ge \frac{1}{4S}(p -\widetilde{p})^2 + \frac{(1-S)}{16(\lambda_{\max}+1)}(\lambda - \widetilde{\lambda})^2.
\end{align*}
\end{lemma}
\noindent \textbf{Proof of Lemma \ref{lemma: KL_Hurdle_square_loss}. } By the result of Lemma \ref{lemma: KL_divergence}, it suffices to lower bound $KL_{\text{Bernoulli}}(p\| \widetilde{p})$ and 
\begin{align*}
 (1-p) \left[\frac{\lambda}{1 - e^{-\lambda}} \log \frac{\lambda}{\widetilde{\lambda}} + \log \frac{e^{\widetilde{\lambda}}-1}{e^\lambda -1}\right],
\end{align*}
respectively. The first term can be lower bounded as 
\begin{align*}
KL_{\text{Bernoulli}}(p\| \widetilde{p}) &= -2p \log \left(\frac{\sqrt{\widetilde{p}}-\sqrt{p}}{\sqrt{p}} + 1\right) -2 (1-p)\log \left(\frac{\sqrt{1-\widetilde{p}}-\sqrt{1-p}}{\sqrt{1-p}} + 1\right)\\
&\ge -2p \cdot \frac{\sqrt{\widetilde{p}}-\sqrt{p}}{\sqrt{p}} -2 (1-p) \cdot \frac{\sqrt{1-\widetilde{p}}-\sqrt{1-p}}{\sqrt{1-p}}\\
&= 2-2 \sqrt{p \widetilde{p}}- 2 \sqrt{(1-p)(1-\widetilde{p})}\\
& = (\sqrt{p} - \sqrt{\widetilde{p}})^2 + \left( \sqrt{1-p} - \sqrt{1-\widetilde{p}}\right)^2. 
\end{align*}
Since 
\begin{align*}
(\sqrt{p} - \sqrt{\widetilde{p}})^2 = \frac{(p-\widetilde{p})^2}{(\sqrt{p} + \sqrt{\widetilde{p}})^2}\ge \frac{(p-\widetilde{p})^2}{4S}, 
\end{align*}
and similarly
\begin{align*}
 \left( \sqrt{1-p} - \sqrt{1-\widetilde{p}}\right)^2 = \frac{(p - \widetilde{p})^2}{ \left( \sqrt{1-p} +\sqrt{1-\widetilde{p}}\right)^2}\ge \frac{(p - \widetilde{p})^2}{ 4(1-s)}, 
\end{align*}
we have 
\begin{align*}
KL_{\text{Bernoulli}}(p\| \widetilde{p}) \ge \frac{(1-s+S) (p -\widetilde{p})^2}{4S(1-s)}. 
\end{align*}
By Lemma \ref{lemma: KL_square}, the second term can be lower bounded as 
\begin{align*}
(1-p) \left[\frac{\lambda}{1 - e^{-\lambda}} \log \frac{\lambda}{\widetilde{\lambda}} + \log \frac{e^{\widetilde{\lambda}}-1}{e^\lambda -1}\right] \ge \frac{(1-S)}{16(\lambda_{\max}+1)}(\lambda - \widetilde{\lambda})^2. 
\end{align*}
The desired result follows immediately from the above two inequalities. \qed

\begin{lemma}
\label{lemma: ZIP_KL_bound}
For two zero-inflated Poisson distributions with parameter $(p, \lambda)$ and $(\widetilde{p}, \widetilde{\lambda})$ such that $0\le p, \widetilde{p}\le 1 - s_{N, K}$ and $0< \lambda_{\min} \le \lambda, \widetilde{\lambda} \le \lambda_{\max}$, we have 
\begin{align*}
& (p - \widetilde{p})^2 + (\lambda - \widetilde{\lambda})^2 \le 4(1 - e^{-\lambda_{\min}})^{-2}\\
& \cdot \max \left\{  1 - s_{N,K} +s_{N, K} e^{-\lambda_{\min}}, \frac{16(\lambda_{\max} + 1)}{s_{N, K}(1 - e^{-\lambda_{\min}})}     \right\} KL_{\text{ZIP}}(p, \lambda|| \widetilde{p}, \widetilde{\lambda}),  
\end{align*}
where $KL_{\text{ZIP}}(p, \lambda|| \widetilde{p}, \widetilde{\lambda})$ denotes the KL-divergence between the two ZIP distributions. 
\end{lemma}

\noindent\textbf{Proof of Lemma \ref{lemma: ZIP_KL_bound}: } Define $\pi = p + (1 - p)e^{-\lambda}$ and $\widetilde{\pi} = \widetilde{p} + (1 -\widetilde{p})e^{- \widetilde{\lambda}}$.
Since $p\le (1- s_{N, K})$ and $\lambda_{\min} \le \lambda \le \lambda_{\max}$, 
\begin{align*}
& \pi \le p + (1 -p)e^{-\lambda_{\min}} = (1 - e^{- \lambda_{\min}})p + e^{- \lambda_{\min}} \le (1 - e^{- \lambda_{\min}})(1 - s_{N, K}) + e^{- \lambda_{\min}},  \text{ and }\\
& \pi = p( 1- e^{-\lambda}) + e^{-\lambda} \ge e^{-\lambda} \ge e^{-\lambda_{\max}}. 
\end{align*}
Therefore, 
\begin{align*}
 e^{-\lambda_{\max}}\le \pi, \widetilde{\pi} \le 1 - s_{N, K} + s_{N, K} e^{-\lambda_{\min}}. 
\end{align*}
Consider the inverse mapping : $(\pi, \lambda)\rightarrow (p, \lambda)$ in such a way that 
\begin{align*}
p = \frac{\pi - e^{-\lambda}}{1 - e^{-\lambda}}. 
\end{align*}
By the mean value theorem, there exists a $(\pi^\prime, \lambda^\prime)$ between $(\pi, \lambda)$ and $(\widetilde{\pi}, \widetilde{\lambda})$ such that 
\begin{align*}
\begin{pmatrix}
p - \widetilde{p}\\
\lambda - \widetilde{\lambda}
\end{pmatrix}
= \begin{pmatrix}
\frac{1}{1 - e^{- \lambda^\prime}} & \frac{e^{-\lambda^\prime}(1 - \pi^\prime)}{(1 - e^{-\lambda^\prime})^2}\\
0 & 1
\end{pmatrix}
\begin{pmatrix}
\pi - \widetilde{\pi}\\
\lambda - \widetilde{\lambda}
\end{pmatrix}, 
\end{align*}
leading to 
\begin{align*}
(p - \widetilde{p})^2 + (\lambda - \widetilde{\lambda})^2 &\le \max_{\lambda^\prime} \left(\frac{1}{1 - e^{-\lambda^\prime}}\right)^2 \left[ (\pi - \widetilde{\pi})^2 + (\lambda-\widetilde{\lambda})^2\right]\\
&=  \left(\frac{1}{1 - e^{-\lambda_{\min}}}\right)^2 \left[ (\pi - \widetilde{\pi})^2 + (\lambda-\widetilde{\lambda})^2\right]. 
\end{align*}
where we use the fact that the spectral norm of the Jacobian matrix is $\frac{1}{1 - e^{-\lambda^\prime}}$.
It then follows from the result in Lemma \ref{lemma: KL_Hurdle_square_loss} that
\begin{align*}
 KL_{\text{ZIP}}(p, \lambda|| \widetilde{p}, \widetilde{\lambda}) & = KL_{\text{Hurdle}}(\pi, \lambda || \widetilde{\pi}, \widetilde{\lambda})\\
& \ge \frac{1}{4 (1 - s_{N,K} +s_{N, K} e^{-\lambda_{\min}})}(\pi -\widetilde{\pi})^2 + \frac{s_{N, K}(1 - e^{-\lambda_{\min}})}{16(\lambda_{\max} + 1)} (\lambda - \widetilde{\lambda})^2\\
&\ge \frac{1}{4} (1 - e^{-\lambda_{\min}})^2 \min \left\{  \frac{1}{1 - s_{N,K} +s_{N, K} e^{-\lambda_{\min}}}, \frac{s_{N, K}(1 - e^{-\lambda_{\min}})}{4(\lambda_{\max} + 1)}      \right\}\\
& \quad \times \left[ (p - \widetilde{p})^2 + (\lambda - \widetilde{\lambda})^2 \right].  
\end{align*}
The desired result follows immediately. \qed

\section{Consistency} First, to connect the zero-inflated Poisson distribution with its Hurdle Poisson representation whose KL-divergence is easier to tackle, we define 
\begin{align*}
\pi^*_{i, j, k} = p_{i, j, k}^* + (1 - p^*_{i, j,k})e^{- \lambda^*_{i,j, k}}, \text{ and }\widehat{\pi}_{i, j, k} = \widehat{p}_{i, j, k} + (1 - \widehat{p}_{i, j,k})e^{- \widehat{\lambda}_{i,j, k}}, 
\end{align*}
for any $i, j \in [N]$ and $k \in [K]$. Therefore, $p_{i, j, k}^* = \frac{\pi_{i, j, k}^* - e^{-\lambda^*_{i, j, k}}}{1 - e^{-\lambda^*_{i, j, k}}}$ and $1- p_{i, j, k}^* = \frac{1 - \pi_{i, j, k}^* }{1 - e^{-\lambda^*_{i, j, k}}}$, similarly for $\widehat{p}_{i,j, k}$ and $1-\widehat{p}_{i, j, k}$. Further denote $\varphi_{N, K} = N(N-1)K$ be the number of independent zero-inflated Poisson random variables in the Hi-C tensor. 

By the definitions of $\widehat{\pi}_{i, j, k}$ and $\widehat{\lambda}_{i, j, k}$, we have 
\begin{align*}
& \quad \frac{1}{\varphi_{N, K}} \sum_{i\le j, k}  \left[\ind(C_{i, j, k} \ne 0) \log\left( \frac{C_{i, j, k}! \widehat{\pi}_{i, j, k} (1 - e^{-\widehat{\lambda}_{i, j, k}})}{(1-\widehat{\pi}_{ij, k}) \widehat{\lambda}_{i, j, k}^{C_{i, j, k}} e^{-\widehat{\lambda}_{i, j, k}}}\right)  - \log \widehat{\pi}_{i, j, k} \right] \\
& \le \frac{1}{\varphi_{N, K}}\sum_{i\le  j, k}  \left[ \ind(C_{i, j, k} \ne 0) \log\left( \frac{C_{i, j, k}!\pi^*_{i, j, k} (1 - e^{-\lambda^*_{i, j, k}})}{(1-\pi^*_{i, j, k}) (\lambda^*_{i, j, k})^{C_{i, j, k}} e^{-\lambda^*_{i, j, k}}}\right) - \log \pi_{i, j, k}^* \right] + \epsilon_{N, K}.
\end{align*}
Reorganizing the terms yields that 
\begin{align}
\label{equ: likelihood_diff}
& \quad \frac{1}{\varphi_{N, K}}\sum_{i\le  j, k} \log \frac{\pi^*_{i,j k}}{\widehat{\pi}_{i, j,k}} \nonumber\\
& \le \frac{1}{\varphi_{N, K}}\sum_{i\le  j, k}  \ind(C_{i, j, k} \ne 0) \left(\log 
\frac{\pi^*_{i, j, k}}{\widehat{\pi}_{i, j, k}} + \log\frac{1 - \widehat{\pi}_{
i, j, k}}{1 - \pi^*_{i, j, k}}+\log\frac{e^{\lambda^*_{i, j,k}} - 1}{e^{\widehat{\lambda}_{i,j,k}} -1} + C_{i, j, k} \log \frac{\widehat{\lambda}_{i, j, k}}{\lambda_{i, j, k}^*}\right) + \epsilon_{N, K}.
\end{align}
Note that 
\begin{align*}
 & \quad \EE \frac{1}{\varphi_{N, K}}\sum_{i\le  j, k} \ind(C_{i, j, k} \ne 0) \left(\log 
\frac{\pi^*_{i, j, k}}{\widehat{\pi}_{i, j, k}} + \log\frac{1 - \widehat{\pi}_{
i, j, k}}{1 - \pi^*_{i, j, k}}+\log\frac{e^{\lambda^*_{i, j,k}} - 1}{e^{\widehat{\lambda}_{i,j,k}} -1} + C_{i,j, k} \log \frac{\widehat{\lambda}_{i, j, k}}{\lambda_{i, j, k}^*}\right)\\
& =\frac{1}{\varphi_{N, K}} \sum_{i\le  j, k}  (1 - \pi_{i, j, k}^*) \left[\left(\log 
\frac{\pi^*_{i, j, k}}{\widehat{\pi}_{i, j, k}} + \log\frac{1 - \widehat{\pi}_{
i, j, k}}{1 - \pi^*_{i, j, k}}+\log\frac{e^{\lambda^*_{i, j,k}} - 1}{e^{\widehat{\lambda}_{i,j,k}} -1}\right) + \frac{\lambda^*_{i,j,k}}{1 - e^{-\lambda^*_{i, j,k}}} \log \frac{\widehat{\lambda}_{i, j, k}}{\lambda_{i, j, k}^*}\right].
\end{align*}
Subtracting this expectation on both sides of the inequality (\ref{equ: likelihood_diff}), we obtain
\begin{align}
\label{equ: KL_to_concentration}
& \quad \frac{1}{\varphi_{N, K}}\sum_{i\le  j, k} KL_{\text{Hurdle}}(\pi_{i, j, k}^*, \lambda_{i, j, k}^*|| \widehat{\pi}_{i, j, k}, \widehat{\lambda}_{i, j, k})\le \frac{1}{\varphi_{N, K}} \sum_{i\le  j, k} \left[\ind(C_{i, j,k} \ne 0) - (1 - \pi^*_{i, j, k})\right] \nonumber\\
&\times \left(\log 
\frac{\pi^*_{i, j, k}}{\widehat{\pi}_{i, j, k}} + \log\frac{1 - \widehat{\pi}_{
i, j, k}}{1 - \pi^*_{i, j, k}}+\log\frac{e^{\lambda^*}_{i, j,k} - 1}{e^{\widehat{\lambda}_{i,j,k}} -1}\right) + \frac{1}{\varphi_{N, K}} \sum_{i\le j, k} \left(C_{i, j, k} - \frac{(1-\pi_{i, j,k}^*)\lambda_{i,j,k}^*}{1-e^{-\lambda^*_{i, j, k}}}\right) \log \frac{\widehat{\lambda}_{i, j, k}}{\lambda_{i, j,k}^*}  \nonumber\\
 & \quad+ \epsilon_{N, K}.
\end{align}

Herein, we use the fact that $\ind(C_{i, j, k} \ne 0)C_{i,j, k} = C_{i, j, k}$.  

Second, define 
\begin{equation*}
S = \left\{(\bGamma, \bZ, \bbeta, \bxi): \frac{1}{\varphi_{N, K}} \sum_{i\le j, k} KL_{\text{Hurdle}}\big(\pi_{i,j, k}^*, \lambda_{i, j, k}^* ||\pi_{i, j, k}, \lambda_{i,j, k} )\ge 4 \epsilon_{N, K} \right\}, \text{ and }
\end{equation*}
\begin{equation*}
S_u = \left\{(\bGamma, \bZ, \bbeta, \bxi): 2^{u+1}\epsilon_{N, K} \le \frac{1}{\varphi_{N, K}} \sum_{i\le j, k} KL_{\text{Hurdle}}\big(\pi_{i,j, k}^*, \lambda_{i, j, k}^* ||\pi_{i, j, k}, \lambda_{i,j, k} )< 2^{u+2} \epsilon_{N, K} \right\},
\end{equation*}
for $u = 1, 2, \ldots $, where $\pi_{i, j, k}$ and $\lambda_{i, j, k}$ are functions of $(\bGamma, \bZ, \bbeta, \bxi)$. Clearly, $S = \bigcup_{u = 1}^{\infty} S_u$. 
As such, we have
\begin{align*}
& \quad \PP\left(\frac{1}{\varphi_{N, K}}\sum_{i\le  j, k} KL_{\text{Hurdle}}(\pi_{i, j, k}^*, \lambda_{i, j, k}^*|| \widehat{\pi}_{i, j, k}, \widehat{\lambda}_{i, j, k}) \ge 4\epsilon_{N, K} \right)\\
& \le \PP\left\{ \sup_{(\bGamma, \bZ, \bbeta, \bxi) \in S}
\frac{1}{\varphi_{N, K}} \sum_{i\le  j, k} \left[\ind(C_{i, j,k} \ne 0) - (1 - \pi^*_{i, j, k})\right] \left(\log 
\frac{\pi^*_{i, j, k}}{\pi_{i, j, k}} + \log\frac{1 -\pi_{
i, j, k}}{1 - \pi^*_{i, j, k}}+\log\frac{e^{\lambda^*}_{i, j,k} - 1}{e^{\lambda_{i,j,k}} -1}\right) \right. \\
& \left. \quad +  \frac{1}{\varphi_{N, K}} \sum_{i\le j, k}\left(C_{i, j, k} - \frac{(1-\pi_{i, j,k}^*)\lambda_{i,j,k}^*}{1-e^{-\lambda^*_{i, j, k}}}\right) \log \frac{\lambda_{i, j, k}}{\lambda_{i, j,k}^*}  \right.\\
& \left. - \frac{1}{\varphi_{N, K}}\sum_{i\le  j, k} KL_{\text{Hurdle}}(\pi_{i, j, k}^*, \lambda_{i, j, k}^*||\pi_{i, j, k},\lambda_{i, j, k})\ge -\epsilon_{N, K}
\right\} \le \sum_{u = 1}^{+\infty} I_u, 
\end{align*}
where the inequality comes from (\ref{equ: KL_to_concentration}). Herein, $I_u$ is defined as 
\begin{align*}
& I_u= \PP\Bigg\{ \sup_{(\bGamma, \bZ, \bbeta, \bxi) \in S_u}
\frac{1}{\varphi_{N, K}} \sum_{i\le  j, k} \left[\ind(C_{i, j,k} \ne 0) - 1 +
 \pi^*_{i, j, k}\right] \left(\log 
\frac{\pi^*_{i, j, k}}{\pi_{i, j, k}} + \log\frac{1 - \pi_{
i, j, k}}{1 - \pi^*_{i, j, k}}+\log\frac{e^{\lambda^*}_{i, j,k} - 1}{e^{\lambda_{i,j,k}} -1}\right)\\
& + \frac{1}{\varphi_{N, K}} \sum_{i\le j, k} \left[\left(C_{i, j, k} - \frac{(1-\pi_{i, j,k}^*)\lambda_{i,j,k}^*}{1-e^{-\lambda^*_{i, j, k}}}\right) \log \frac{\lambda_{i, j, k}}{\lambda_{i, j,k}^*} -KL_{\text{Hurdle}}(\pi_{i, j, k}^*, \lambda_{i, j, k}^*|| \pi_{i, j, k}, \lambda_{i, j, k})\right]\ge -\epsilon_{N, K}
\Bigg\}.
\end{align*}
By the definition of $S_u$ and the inequality that $2^{u+1} - 1\ge 2^{u}$ for $u\ge 1$, $I_u$ is upper bounded by
\begin{align*}
&  \PP\Bigg\{ \sup_{(\bGamma, \bZ, \bbeta, \bxi) \in S_u}
\frac{1}{\varphi_{N, K}} \sum_{i\le  j, k} \left[\ind(C_{i, j,k} \ne 0) - 1 +
 \pi^*_{i, j, k}\right] \left(\log 
\frac{\pi^*_{i, j, k}}{\pi_{i, j, k}} + \log\frac{1 - \pi_{
i, j, k}}{1 - \pi^*_{i, j, k}}+\log\frac{e^{\lambda^*}_{i, j,k} - 1}{e^{\lambda_{i,j,k}} -1}\right)\\
& \quad + \frac{1}{\varphi_{N, K}} \sum_{i\le j, k} \left(C_{i, j, k} - \frac{(1-\pi_{i, j,k}^*)\lambda_{i,j,k}^*}{1-e^{-\lambda^*_{i, j, k}}}\right) \log \frac{\lambda_{i, j, k}}{\lambda_{i, j,k}^*} \\
& \quad \ge \inf_{(\bGamma, \bZ, \bbeta, \bxi) \in S_u}\frac{1}{\varphi_{N, K}} \sum_{i\le j, k}KL_{\text{Hurdle}}(\pi_{i, j, k}^*, \lambda_{i, j, k}^*||\pi_{i, j, k}, \lambda_{i, j, k})  -\epsilon_{N, K}
\Bigg\}\\
& \le  \PP\Bigg\{ \sup_{(\bGamma, \bZ, \bbeta, \bxi) \in S_u}
\frac{1}{\varphi_{N, K}} \sum_{i\le  j, k} \left[\ind(C_{i, j,k} \ne 0) - 1 +
 \pi^*_{i, j, k}\right] \left(\log 
\frac{\pi^*_{i, j, k}}{\pi_{i, j, k}} + \log\frac{1 -\pi_{
i, j, k}}{1 - \pi^*_{i, j, k}}+\log\frac{e^{\lambda^*}_{i, j,k} - 1}{e^{\lambda_{i,j,k}} -1}\right)\\
& \quad + \frac{1}{\varphi_{N, K}} \sum_{i\le j, k} \left(C_{i, j, k} - \frac{(1-\pi_{i, j,k}^*)\lambda_{i,j,k}^*}{1-e^{-\lambda^*_{i, j, k}}}\right) \log \frac{\lambda_{i, j, k}}{\lambda_{i, j,k}^*} \ge 2^{u} \epsilon_{N, K}\Bigg\}, 
\end{align*}

To simplify notations, we denote the following 
\begin{align*}
&a_{i, j, k} = \log 
\frac{\pi^*_{i, j, k}}{\pi_{i, j, k}} + \log\frac{1 -\pi_{
i, j, k}}{1 - \pi^*_{i, j, k}}+\log\frac{e^{\lambda^*}_{i, j,k} - 1}{e^{\lambda_{i,j,k}} -1}, \quad  b_{i, j, k} = \log \frac{\lambda_{i,j,k}}{\lambda_{i, j, k}^*},\\
& f_{i, j, k}(C_{i, j, k}) =  a_{i,,j, k} \ind(C_{i,j, k}\ne 0) + b_{i,j, k} C_{i, j, k},   \text{ and } Y_{i,j, k} = f_{i,j, k}(C_{i, j, k}) - \EE f_{i, j,k}(C_{i, j, k}). 
\end{align*}
It then follows from the standard symmetrization argument that 
\begin{align*}
\EE\sup_{(\bGamma, \bZ, \bbeta, \bxi) \in S_u}
\frac{1}{\varphi_{N, K}} \sum_{i\le  j, k} Y_{i,j, k} \le 2 \EE\sup_{(\bGamma, \bZ, \bbeta, \bxi)  \in S_u}\frac{1}{\varphi_{N, K}} | \sum_{i\le  j, k} f_{i, j, k}(C_{i, j, k}) \varepsilon_{i, j, k}|, 
\end{align*}
where $\varepsilon_{i,j,k}$'s are independent Rademacher variables, which are also independent with $C_{i,j,k}$'s. This leads to 
\begin{align*}
I_u &\le \PP \left( \sup_{(\bGamma, \bZ, \bbeta, \bxi) \in S_u}
\frac{1}{\varphi_{N, K}} \sum_{i\le  j, k} Y_{i,j, k} \ge 2^{u} \epsilon_{N, K} \right)  \\
& \le \PP\Bigg( \sup_{(\bGamma, \bZ, \bbeta, \bxi) \in S_u}
\frac{1}{\varphi_{N, K}} \sum_{i\le  j, k} Y_{i,j, k} - (1+ \nu) \EE\sup_{(\bGamma, \bZ, \bbeta, \bxi) \in S_u}
\frac{1}{\varphi_{N, K}} \sum_{i\le  j, k} Y_{i,j, k}\\
& \quad \ge 2^{u} \epsilon_{N, K} - 2 (1+ \nu)\EE\sup_{(\bGamma, \bZ, \bbeta, \bxi)  \in S_u}\frac{1}{\varphi_{N, K}} | \sum_{i\le  j, k} f_{i,j, k}(C_{i, j,k}) \varepsilon_{i, j, k}|
\Bigg).
\end{align*}
for any $0<\nu\le 1$. 

Third, since $S_u$ is a compact set, $0< \lambda_{\min} \le \lambda_{i, j, k} \le \lambda_{\max}< + \infty$, $e^{-\lambda_{\max}} \le \pi_{i, j, k} \le 1 - s_{N, K} + s_{N, K} e^{-\lambda_{\min}}$, and $a_{i, j, k}$ and $b_{i, j, k}$ are continues functions of the model parameters, there exists a parameter tuple $(\bGamma^\dag, \bZ^\dag, \bbeta^\dag, \bxi^\dag) \in S_u$ such that 
\begin{align*}
|f_{i, j, k} (C_{i, j, k})| &\le |
a_{i, j,k}^\dag| \ind(C_{i, j, k} \ne 0) + |b_{i, j, k}^\dag|C_{i, j, k}\\
& \le \left(| \log 
\frac{\pi^*_{i, j, k}}{\pi^\dag_{i, j, k}} + \log\frac{1 -\pi^\dag_{
i, j, k}}{1 - \pi^*_{i, j, k}} | + |\log\frac{e^{\lambda^*}_{i, j,k} - 1}{e^{\lambda^\dag_{i,j,k}} -1}| \right) \ind(C_{i, j, k}\ne 0) + |\log \frac{\lambda^\dag_{i, j, k}}{\lambda_{i, , k}^*}|C_{i, j, k} \\
& := F_{i, j, k}(C_{i, j, k}), 
\end{align*}
for any random tensor $\bm{\mathcal{C}}$ and any tensor value function $\bm{f} = (f_{i, j, k})$ parametrized by $(\bGamma, \bZ, \bbeta, \bxi) \in S_u$. Therefore, $\bF = (F_{i, j, k})$ could serve as an envelop of the function class parametrized by $S_u$, which we further denoted as $\bm{f}_u$. It then follows from the uniform entropy bound that  
\begin{align*}
\EE\sup_{(\bGamma, \bZ, \bbeta, \bxi)  \in S_u}\frac{1}{\varphi_{N, K}} | \sum_{i\le  j, k} f_{i,j, k}(C_{i, j,k}) \varepsilon_{i, j, k}| \le m_1 \varphi_{N, K}^{-1/2}\int_0^1 \sup_Q \sqrt{\log \mathcal{N}\left(\epsilon ||\bF||_{Q, 2}, \bm{f}_u, L_2(Q)\right)} d\epsilon ||\bF||_{P, 2}, 
\end{align*}
for some universal constant $m_1$, where 
\begin{align*}
\|\bF\|_{Q, 2} = \sqrt{\frac{1}{\varphi_{N, K}}\sum_{i\le j, k} \EE_Q F_{i, j, k}(C_{i, j, k})^2}, 
\end{align*}
for any measure $Q$, the supremum in the right hand side is taken over all distributions parametrized by $(\bGamma, \bZ, \bbeta, \bxi) \in S_u$, $P$ corresponds to the underlying true distribution, and $\mathcal{N}$ denotes the covering number in terms of the $L_2(Q)$ distance. For any $(\bGamma^{(1)}, \bZ^{(1)}, \bbeta^{(1)}, \bxi^{(1)})$, $(\bGamma^{(2)}, \bZ^{(2)}, \bbeta^{(2)}, \bxi^{(2)}) \in S_u$ parametrized $\bm{f}^{(1)}$ and $\bm{f}^{(2)}$, we can upper bound their $L_2(Q)$ distance as
\begin{align*}
&\quad L_2(Q)(\bm{f}^{(1)}, \bm{f}^{(2)})\\
&= \sqrt{ \frac{1}{\varphi_{N, K}}\sum_{i, j, k} \EE_Q \left[ (a^{(1)}_{i, j, k} - a^{(2)}_{i, j, k})\ind(C_{i, j, k} \ne 0) + (b^{(1)}_{i, j, k} - b^{(2)}_{i, j, k}) C_{i, j, k}\right]^2 }\\
& = \sqrt{ \frac{1}{\varphi_{N, K}}\sum_{i, j, k} \sum_{c=1}^{+\infty} \left[ (a^{(1)}_{i, j, k} - a^{(2)}_{i, j, k}) + (b^{(1)}_{i, j, k} - b^{(2)}_{i, j, k}) c\right]^2 (1-p_{i, j, k}^{(Q)})\frac{(\lambda^{(Q)})^c e^{-\lambda^{(Q)}}}{c! }}\\
& \le \sqrt{ \frac{2}{\varphi_{N, K}}\sum_{i, j, k} \left[ (a^{(1)}_{i, j, k} - a^{(2)}_{i, j, k})^2(1 - e^{-\lambda_{i, j, k}^{(Q)}}) + (b^{(1)}_{i, j, k} - b^{(2)}_{i, j, k})^2 \lambda_{i, j,k }^{(Q)} (1 + \lambda_{i, j, k}^{(Q)})\right] (1-p_{i, j, k}^{(Q)})}\\
& = \sqrt{ \frac{2}{\varphi_{N, K}}\sum_{i, j, k} \left[ (a^{(1)}_{i, j, k} - a^{(2)}_{i, j, k})^2(1 - e^{-\lambda_{i, j, k}^{(Q)}}) + (\eta^{(1)}_{i, j, k} - \eta^{(2)}_{i, j, k})^2 \lambda_{i, j,k }^{(Q)} (1 + \lambda_{i, j, k}^{(Q)})\right] (1-p_{i, j, k}^{(Q)})}.
\end{align*}

Note that 
\begin{align*}
& |\frac{\partial a_{i, j, k}}{\partial \theta_{i, j, k}}|  = |\frac{\partial a_{i, j, k}}{\partial \pi_{i, j, k}} \cdot \frac{\partial \pi_{i, j, k}}{\partial  (1- p_{i, j, k})}  \cdot \frac{\partial (1-p_{i, j, k})}{\partial \theta_{i, j, k}} | \le \frac{1 - e^{-\lambda_{i, j, k}}}{4\pi_{i, j, k}(1 - \pi_{i, j, k})} = \frac{1}{4 \pi_{i, j, k} (1 - p_{i, j, k})}\\
& \le \frac{1}{4s_{N, K} (1 - s_{N, K} + s_{N, K} e^{-\lambda_{\max}})}, \text{ and }\\
& |\frac{\partial a_{i, j, k}}{\partial \eta_{i, j, k}}|  = |\frac{\partial a_{i, j, k}}{\partial \pi_{i, j, k}} \cdot \frac{\partial \pi_{i, j, k}}{\partial \lambda_{i, j, k}}  \cdot \frac{\partial \lambda_{i,j,k}}{\partial \eta_{i, j, k}} - \frac{\lambda_{i, j, k} e^{\lambda_{i, j, k}}}{e^{\lambda_{i,j k} } - 1} | = | \frac{\lambda_{i, j, k} e^{-\lambda_{i, j, k}}}{\pi_{i, j, k} (1 - e^{-\lambda_{i, j, k}})} -  \frac{\lambda_{i, j, k} e^{\lambda_{i, j, k}}}{e^{\lambda_{i,j k} } - 1}|\\
& \le (\lambda_{i, j, k} + 1) | \frac{e^{-\lambda_{i, j, k}}}{\pi_{i, j, k}} - 1| \le \lambda_{i, j, k} + 1 \le \lambda_{\max} + 1, 
\end{align*}
where the inequality in the second line follows from the monotonicity and the assumption that $s_{N, K} \le \frac{1}{e^{\lambda_{\max}}-1}$, and the first inequality in the forth line follows from the proof of Lemma \ref{lemma: KL_square}. 

Therefore, by the mean value theorem
\begin{align*}
(a^{(1)}_{i, j, k} - a^{(2)}_{i, j, k})^2 \le \frac{(\theta^{(1)}_{i, j, k} - \theta_{i, j, k}^{(2)})^2}{2s_{N, K}^2 (1 - s_{N, K} + s_{N, K} e^{-\lambda_{\max}})^2} + 2 (\lambda_{\max} + 1)^2 (\eta_{i,j, k}^{(1)} - \eta^{(2)}_{i, j, k})^2. 
\end{align*}
It then follows that 
\begin{align*}
& \quad L_2(Q)(\bm{f}^{(1)}, \bm{f}^{(2)})^2\\
& \le \frac{2}{\varphi_{N, K}}\sum_{i, j, k} \left[ (a^{(1)}_{i, j, k} - a^{(2)}_{i, j, k})^2(1 - e^{-\lambda_{i, j, k}^{(Q)}}) + (\eta^{(1)}_{i, j, k} - \eta^{(2)}_{i, j, k})^2 \lambda_{i, j,k }^{(Q)} (1 + \lambda_{i, j, k}^{(Q)})\right] (1-p_{i, j, k}^{(Q)})\\
& =  \frac{2}{\varphi_{N, K}}\sum_{i, j, k} \left[ (a^{(1)}_{i, j, k} - a^{(2)}_{i, j, k})^2(1 - \pi^{(Q)}_{i, j, k}) + (\eta^{(1)}_{i, j, k} - \eta^{(2)}_{i, j, k})^2 \lambda_{i, j,k }^{(Q)} (1 + \lambda_{i, j, k}^{(Q)}) (1-p_{i, j, k}^{(Q)})\right]\\
& \le \frac{2}{\varphi_{N, K}}\sum_{i, j, k} \left[ \frac{ (1 - \pi^{(Q)}_{i, j, k})(\theta^{(1)}_{i, j, k} - \theta_{i, j, k}^{(2)})^2}{2s_{N, K}^2 (1 - s_{N, K} + s_{N, K} e^{-\lambda_{\max}})^2} + (\lambda_{\max} + 1)^2 (3 - 2\pi^{(Q)}_{i, j, k}-p_{i, j, k}^{(Q)}) (\eta_{i,j, k}^{(1)} - \eta^{(2)}_{i, j, k})^2\right]\\
& \le \frac{2}{\varphi_{N, K}}\sum_{i, j, k} \left[ \frac{(\theta^{(1)}_{i, j, k} - \theta_{i, j, k}^{(2)})^2}{2s_{N, K}^2 (1 - s_{N, K} + s_{N, K} e^{-\lambda_{\max}})^2} + 3(\lambda_{\max} + 1)^2  (\eta_{i,j, k}^{(1)} - \eta^{(2)}_{i, j, k})^2\right]\\
& = \frac{1}{\varphi_{N, K} s_{N, K}^2 (1 - s_{N, K} + s_{N, K} e^{-\lambda_{\max}})^2} \|\bTheta^{(1)} - \bTheta^{(2)}\|_F^2 + \frac{6(\lambda_{\max} + 1)^2 }{\varphi_{N, K}} \|\bm{\eta}^{(1)} - \bm{\eta}^{(2)}\|_F^2. 
\end{align*}
Taking the low-rank structures of  $\bm{\eta}^{(1)}, \bm{\eta}^{(2)}, \bTheta^{(1)}$, and $\bTheta^{(2)}$ into consideration, we have 
\begin{align*}
& \quad \|\bm{\eta}^{(1)} - \bm{\eta}^{(2)}\|_F \le \|\bm{\mathcal{I}} \times_1 \bH (\bGamma^{(1)} - \bGamma^{(2)})\times_2 \bH \bGamma^{(1)} \times_3 \widetilde{\bbeta}^{(1)} \|_F\\
& + \|\bm{\mathcal{I}} \times_1 \bH \bGamma^{(2)} \times_2 \bH (\bGamma^{(1)} - \bGamma^{(2)}) \times_3 \widetilde{\bbeta}^{(1)} \|_F + \|\bm{\mathcal{I}} \times_1 \bH \bGamma^{(2)} \times_2 \bH \bGamma^{(2)} \times_3 (\widetilde{\bbeta}^{(1)} - \widetilde{\bbeta}^{(2)}) \|_F\\
& \le  \sqrt{Q}\|\bGamma^{(1)} - \bGamma^{(2)} \|_F  \| \bGamma^{(1)}\|_F \|\widetilde{\bbeta}^{(1)} \|_F +\sqrt{Q} \|\bGamma^{(2)}\|_F  \|\bGamma^{(1)} - \bGamma^{(2)}\|_F\| \widetilde{\bbeta}^{(1)} \|_F\\
& + Q\|\bGamma^{(2)}\|_F^2 \| \widetilde{\bbeta}^{(1)} - \widetilde{\bbeta}^{(2)}\|_F,
\end{align*}
where the last inequality follows from the fact that the columns of $\bH$ are certain realizations of orthonormal function basis. Without loss of generality, we can assume the columns of both $\bGamma^{(1)}$ and $\bGamma^{(2)}$ have unit $l_2$-norm; otherwise we can rescaled the factors in the last inequality such that $\bGamma^{(1)}$ and $\bGamma^{(2)}$ fall into the parameter space $\Omega$. Since $\widetilde{\bbeta}^{(1)} = \bZ^{(1)} \bbeta^{(1)}$ and $\widetilde{\bbeta}^{(2)} = \bZ^{(2)}\bbeta^{(2)}$, we have
\begin{align*}
\| \widetilde{\bbeta}^{(1)} - \widetilde{\bbeta}^{(2)}\|_F & \le \|\bZ^{(1)} \bbeta^{(1)} - \bZ^{(2)} \bbeta^{(1)} \|_F +  \|\bZ^{(2)}  \bbeta^{(1)} - \bZ^{(2)}\bbeta^{(2)} \|_F\\
& \le \sqrt{RL} \beta_{\max} \|\bZ^{(1)} - \bZ^{(2)}\|_F + \sqrt{K}\|\bbeta^{(1)} - \bbeta^{(2)}\|_F. 
\end{align*}
These leads to 
\begin{align*}
& \quad \|\bm{\eta}^{(1)} - \bm{\eta}^{(2)}\|_F\\
&\le  2L \sqrt{KQ} \beta_{\max}\|\bGamma^{(1)} - \bGamma^{(2)} \|_F + QL^{3/2} \sqrt{R} \beta_{\max} \|\bZ^{(1)} - \bZ^{(2)}\|_F +Q L\sqrt{K} \|\bbeta^{(1)} - \bbeta^{(2)}\|_F\\
& = 2L^{3/2} \sqrt{QK} \beta_{\max} \|\frac{1}{\sqrt{L}} (\bGamma^{(1)} - \bGamma^{(2)})\|_F + L^{3/2} Q\sqrt{KR} \beta_{\max} \|\frac{1}{\sqrt{K}}(\bZ^{(1)} - \bZ^{(2)})\|_F\\
& \quad \quad \quad \quad\quad \quad\quad \quad + L^{3/2} Q\sqrt{RK} \beta_{\max}\| \frac{1}{\sqrt{RL} \beta_{\max}}(\bbeta^{(1)} - \bbeta^{(2)})\|_F\\
& \le 2L^{3/2} Q\sqrt{KR} \beta_{\max} \left(\|\frac{1}{\sqrt{L}} (\bGamma^{(1)} - \bGamma^{(2)})\|_F +  \|\frac{1}{\sqrt{K}}(\bZ^{(1)} - \bZ^{(2)})\|_F \right.\\
& \quad \quad \quad \left.+ \| \frac{1}{\sqrt{RL} \beta_{\max}}(\bbeta^{(1)} - \bbeta^{(2)})\|_F \right).
\end{align*}

Similarly,
\begin{align*}
& \quad \|\bm{\Theta}^{(1)} - \bm{\Theta}^{(2)}\|_F \le 2L^{3/2} Q\sqrt{KR} \xi_{\max} \left(\|\frac{1}{\sqrt{L}} (\bGamma^{(1)} - \bGamma^{(2)})\|_F +  \|\frac{1}{\sqrt{K}}(\bZ^{(1)} - \bZ^{(2)})\|_F \right.\\
& \quad \quad \quad \left.+ \| \frac{1}{\sqrt{RL} \xi_{\max}}(\bxi^{(1)} - \bxi^{(2)})\|_F \right).
\end{align*}
Consequently, 
\begin{align*}
& \quad L_2(Q)(\bm{f}^{(1)}, \bm{f}^{(2)})\\
& \le \frac{1}{\varphi^{1/2}_{N, K} s_{N, K} (1 - s_{N, K} + s_{N, K} e^{-\lambda_{\max}})} \|\bTheta^{(1)} - \bTheta^{(2)}\|_F + \frac{\sqrt{6}(\lambda_{\max} + 1) }{\varphi^{1/2}_{N, K}} \|\bm{\eta}^{(1)} - \bm{\eta}^{(2)}\|_F\\
& = \frac{2L^{3/2} Q\sqrt{KR}}{\varphi^{1/2}_{N, K}}\left(\frac{\xi_{\max}}{s_{N, K}(1 -s_{N, K}+ s_{N, K}e^{-\lambda_{\max}})} + \sqrt{6}(\lambda_{\max} +1 ) \beta_{\max}\right)\left(\|\frac{1}{\sqrt{L}}(\bGamma^{(1)} - \bGamma^{(2)}) \|_F  \right.\\
& \quad \left. \ +\| \frac{1}{\sqrt{K}}(\bZ^{(1)}- \bZ^{(2)})\|_F\right) + \frac{\sqrt{6}(\lambda_{\max} + 1) }{\varphi^{1/2}_{N, K}} \cdot 2QL^{3/2} \sqrt{KR} \beta_{\max}\| \frac{1}{\sqrt{RL} \beta_{\max}}(\bbeta^{(1)} - \bbeta^{(2)})\|_F.\\
& \quad + \frac{2QL^{3/2} \sqrt{KR} \xi_{\max}}{\varphi^{1/2}_{N, K} s_{N, K} (1 - s_{N, K} + s_{N, K} e^{-\lambda_{\max}})}  \| \frac{1}{\sqrt{
RL} \xi_{\max}}(\bxi^{(1)} - \bxi^{(2)})\|_F.
\end{align*}
To further simply notation, there exists a universal constant $m_2>0$, such that 
\begin{align*}
\quad L_2(Q)(\bm{f}^{(1)}, \bm{f}^{(2)}) \le \frac{m_2 \lambda_{\max} QL^{3/2} \sqrt{KR} \beta_{\max} \xi_{\max}}{s_{N, K}\varphi^{1/2}_{N, K}} \left(\|\frac{1}{\sqrt{L}}(\bGamma^{(1)} - \bGamma^{(2)}) \|_F\right.\\
\left.+\| \frac{1}{\sqrt{K}}(\bZ^{(1)}- \bZ^{(2)})\|_F +\| \frac{1}{\sqrt{RL} \beta_{\max}}(\bbeta^{(1)} - \bbeta^{(2)})\|_F +  \| \frac{1}{\sqrt{RL} \xi_{\max}}(\bxi^{(1)} - \bxi^{(2)})\|_F\right).
\end{align*}
It then follows that 
\begin{align*}
& \quad \mathcal{N}\left(\epsilon ||\bF||_{Q, 2}, \bm{f}_u, L_2(Q)\right)\le \mathcal{N}\left(\frac{\epsilon \|\bF\|_{Q, 2}s_{N, K}\varphi^{1/2}_{N, K}}{4m_2\lambda_{\max}QL^{3/2}\sqrt{KR}\beta_{\max} \xi_{\max}}, B(QL), ||\cdot||_2 \right) \\
&\cdot \mathcal{N}\left(\frac{\epsilon \|\bF\|_{Q, 2}s_{N, K}\varphi^{1/2}_{N, K}}{4m_2\lambda_{\max}QL^{3/2}\sqrt{KR}\beta_{\max} \xi_{\max}}, B(KR), ||\cdot||_2 \right)\\
& \cdot \mathcal{N}\left(\frac{\epsilon \|\bF\|_{Q, 2}s_{N, K}\varphi^{1/2}_{N, K}}{4m_2\lambda_{\max}QL^{3/2}\sqrt{KR}\beta_{\max} \xi_{\max}}, B(RL), ||\cdot||_2 \right)^2\\
& \le \left(\frac{12m_2\lambda_{\max}QL^{3/2}\sqrt{KR} \beta_{\max} \xi_{\max}}{\epsilon \|\bF\|_{Q, 2}s_{N, K}\varphi^{1/2}_{N, K}}\right)^{QL + KR + 2RL}, 
\end{align*}
where $B(x)$ stands for the unit $l_2$-norm ball of dimension $x$. Therefore, 
\begin{align*}
& \quad \EE\sup_{(\bGamma, \bZ, \bbeta, \bxi)  \in S_u}\frac{1}{\varphi_{N, K}} | \sum_{i\le  j, k} f_{i,j, k}(C_{i, j,k}) \varepsilon_{i, j, k}|\\
& \le m_1 \varphi_{N, K}^{-1/2} \sqrt{QL+KR+2RL}    \int_0^1  \sqrt{\log  \frac{12 m_2\lambda_{\max}QL^{3/2}\sqrt{KR}\beta_{\max} \xi_{\max}}{\epsilon \inf_Q \|\bF\|_{Q, 2}s_{N, K}\varphi^{1/2}_{N, K}}  } d\epsilon ||\bF||_{P, 2}\\
& \le \frac{12m_1m_2 \lambda_{\max} QL^{3/2} \sqrt{KR} \beta_{\max} \xi_{\max}\sqrt{QL+KR+2RL} }{\inf_Q\|\bF\|_{Q, 2}s_{N, K}  \varphi_{N, K}} \int_{\frac{12m_2 \lambda_{\max} QL^{3/2} \sqrt{KR} \beta_{\max} \xi_{\max}}{\inf_Q\|\bF\|_{Q, 2}s_{N, K} \varphi^{1/2}_{N, K}}}^{+\infty}  \frac{\sqrt{\log  t }}{t^2} dt ||\bF||_{P, 2}\\
& \le \frac{12m_1m_2 \lambda_{\max} QL
^{3/2}\sqrt{KR}\beta_{\max} \xi_{\max} \sqrt{QL+KR+2RL} }{\inf_Q\|\bF\|_{Q, 2}s_{N, K} \varphi_{N, K} \sqrt{\log \frac{12m_2 \lambda_{\max} QL^{3/2} \sqrt{KR}\beta_{\max} \xi_{\max}}{\inf_Q\|\bF\|_{Q, 2}s_{N, K} \varphi^{1/2}_{N, K}}} } \int_{\frac{12m_2 \lambda_{\max} QL^{3/2} \sqrt{KR}\beta_{\max} \xi_{\max}}{\inf_Q\|\bF\|_{Q, 2}s_{N, K} \varphi^{1/2}_{N, K}}}^{+\infty}  \frac{\log  t}{t^2} dt ||\bF||_{P, 2}\\
& = \frac{12m_1m_2 \lambda_{\max} QL^{3/2} \sqrt{KR}\beta_{\max} \xi_{\max} \sqrt{QL+KR+2RL} }{\inf_Q\|\bF\|_{Q, 2}s_{N, K} \varphi_{N, K}\sqrt{\log \frac{12m_2 \lambda_{\max} QL^{3/2} \sqrt{KR}\beta_{\max} \xi_{\max}}{\inf_Q\|\bF\|_{Q, 2}s_{N, K} \varphi_{N, K}\varphi^{1/2}_{N, K}}} } ||\bF||_{P, 2}\\
&\quad \quad \quad \times \frac{\inf_Q\|\bF\|_{Q, 2}s_{N, K} \varphi^{1/2}_{N, K}}{12m_2 \lambda_{\max} QL^{3/2} \sqrt{KR}\beta_{\max} \xi_{\max}} \left(1 + \log \frac{12m_2 \lambda_{\max} QL^{3/2} \sqrt{KR}\beta_{\max} \xi_{\max}}{\inf_Q\|\bF\|_{Q, 2}s_{N, K} \varphi^{1/2}_{N, K}} \right)\\
& = \frac{m_1  \sqrt{QL+KR+2RL} }{\varphi_{N, K}^{1/2} \sqrt{\log \frac{12m_2 \lambda_{\max} QL^{3/2} \sqrt{KR}\beta_{\max} \xi_{\max}}{\inf_Q\|\bF\|_{Q, 2}s_{N, K} \varphi^{1/2}_{N, K}}} } ||\bF||_{P, 2} \left(1 + \log \frac{12 m_2 \lambda_{\max} QL^{3/2} \sqrt{KR}\beta_{\max} \xi_{\max}}{\inf_Q\|\bF\|_{Q, 2}s_{N, K} \varphi^{1/2}_{N, K}} \right).
\end{align*}

\noindent \textbf{Lower bound of $\|\bF\|_{Q, 2}$}
Notice that 
\begin{align*}
KL_{\text{Bernoulli}}(\pi||\widetilde{\pi}) \le\left |\log\frac{\pi}{\widetilde{\pi}} + \log \frac{1 - \widetilde{\pi}}{1 -\pi}\right|. 
\end{align*}
This is because when $\pi\ge \widetilde{\pi}$, we have 
\begin{align*}
KL_{\text{Bernoulli}}(\pi||\widetilde{\pi}) = \pi \log \frac{\pi}{\widetilde{\pi}} + (1 - \pi) \log \frac{1 - \pi}{1 - \widetilde{\pi}} \le  \pi \log \frac{\pi}{\widetilde{\pi}}  - \pi\log \frac{1 - \pi}{1 - \widetilde{\pi}} \le  \log \frac{\pi}{\widetilde{\pi}}  + \log \frac{1 - \widetilde{\pi}}{1 - \pi}. 
\end{align*}
When $\pi< \widetilde{\pi}$, 
\begin{align*}
KL_{\text{Bernoulli}}(\pi||\widetilde{\pi}) \le -(1-\pi) \log \frac{\pi}{\widetilde{\pi}} + (1 - \pi)\log \frac{1 - \pi}{1 - \widetilde{\pi}} \le  \log\frac{\widetilde{\pi}}{\pi} + \log \frac{1 - \pi}{1 - \widetilde{\pi}}  .
\end{align*}
Moreover, 
\begin{align*}
\frac{\lambda}{1 - e^{-\lambda}}\log \frac{\lambda}{\widetilde{\lambda}} + \log \frac{e^{\widetilde{\lambda}} - 1}{e^{\lambda} - 1} \le  \lambda_{\max}\left|\log\frac{e^\lambda - 1}{e^{\tilde{\lambda}}-1}\right|. 
\end{align*}
This is because when $\lambda \le  \widetilde{\lambda}$, by the assmptions that $\lambda_{\max} > 1$, we have
\begin{align*}
\frac{\lambda}{1 - e^{-\lambda}}\log \frac{\lambda}{\widetilde{\lambda}} + \log \frac{e^{\widetilde{\lambda}} - 1}{e^{\lambda} - 1} \le  \log \frac{e^{\widetilde{\lambda}} - 1}{e^{\lambda} - 1} = \left|   \log \frac{e^{\lambda} - 1}{e^{\widetilde{\lambda}} - 1} \right|\le\lambda_{\max} \left|   \log \frac{e^{\lambda} - 1}{e^{\widetilde{\lambda}} - 1} \right|, 
\end{align*}
and when $\lambda > \widetilde{\lambda}$, we have
\begin{align*}
\frac{\lambda}{1 - e^{-\lambda}}\log \frac{\lambda}{\widetilde{\lambda}} + \log \frac{e^{\widetilde{\lambda}} - 1}{e^{\lambda} - 1} \le \left(\frac{\lambda}{1 - e^{-\lambda}} - 1\right)\log \frac{e^{\lambda} - 1}{e^{\widetilde{\lambda}} - 1} < \lambda_{\max} \left|  \log \frac{e^{\lambda} - 1}{e^{\widetilde{\lambda}} - 1}\right|.
\end{align*}
It then follows that within the set $S_u$, we have
\begin{align*}
& \quad \|\bF\|_{Q, 2}\\
&= \sqrt{\frac{1}{\varphi_{N, K}}\sum_{i\le j, k} \EE_Q F_{i, j, k}(C_{i, j, k})^2}\\
& \ge\sqrt{s_{N, K}(1-e^{-\lambda_{\min}})} \sqrt{\frac{1}{\varphi_{N, K}}\sum_{i\le j, k} \left(| \log 
\frac{\pi^*_{i, j, k}}{\pi^\dag_{i, j, k}} + \log\frac{1 - \pi^\dag_{
i, j, k}}{1 - \pi^*_{i, j, k}} | + |\log\frac{e^{\lambda^*}_{i, j,k} - 1}{e^{\lambda^\dag_{i,j,k}} -1}| \right)^2}\\
& \ge \frac{\sqrt{s_{N, K}(1-e^{-\lambda_{\min}})}}{\varphi_{N, K}}\sum_{i\le j, k} \left(| \log 
\frac{\pi^*_{i, j, k}}{\pi^\dag_{i, j, k}} + \log\frac{1 -\pi^\dag_{
i, j, k}}{1 - \pi^*_{i, j, k}} | + |\log\frac{e^{\lambda^*}_{i, j,k} - 1}{e^{\lambda^\dag_{i,j,k}} -1}| \right)\\
& \ge \frac{\sqrt{s_{N, K}(1-e^{-\lambda_{\min}})}}{\varphi_{N, K}} \sum_{i\le j, k} \left[ KL_{\text{Bernoulli}}(\pi^*
_{i, j, k}||\pi^\dag_{i, j, k}) \right.\\
& \quad \left. +\frac{1}{\lambda_{\max}}\frac{\lambda^*_{i, j, k}}{1 - e^{-\lambda^*_{i, j, k}}}\log \frac{\lambda^*_{i, j, k}}{\lambda^\dag_{i, j, k}} + \frac{1}{\lambda_{\max}}\log \frac{e^{\lambda^\dag_{i, j, k}} - 1}{e^{\lambda^*_{i, j, k}} - 1}\right]\\
& \ge \frac{\sqrt{s_{N, K}(1-e^{-\lambda_{\min}})}}{\varphi_{N,K}}\sum_{i\le j, k} \min\{1, \frac{1}{\lambda_{\max}(1 -\pi^*_{i, j, k})}\}KL_{\text{Hurdle}}(\pi^*_{i,j, k}, \lambda^*_{i, j, k}|| \pi^\dag_{i, j, k}, \lambda^\dag_{i, j, k})\\
& \ge \sqrt{s_{N, K}(1-e^{-\lambda_{\min}})}\min\{1, \frac{1}{\lambda_{\max}(1 - e^{-\lambda_{\max}})}\}\frac{1}{\varphi_{N,K}}\sum_{i\le j, k} KL_{\text{Hurdle}}(\pi^*_{i,j, k}, \lambda^*_{i, j, k}|| \pi^\dag_{i, j, k}, \lambda^\dag_{i, j, k})\\
& \ge \sqrt{s_{N, K}(1-e^{-\lambda_{\min}})}\min\{1, \frac{1}{\lambda_{\max}(1 - e^{-\lambda_{\max}})}\} 2^{u+1} \epsilon_{N, K}. 
\end{align*}

\noindent \textbf{Upper bound of $\|\bF\|_{P, 2}$}

For any set of parameters $(\pi = p + (1-p)e^{-\lambda}, \lambda)$ and $(\widetilde{\pi} = \widetilde{p} + (1 - \widetilde{p})e^{-\widetilde{\lambda}}, \widetilde{\lambda})$, by the mean value theorem, we have  
\begin{align*}
& |\log \frac{\pi}{\widetilde{\pi}}|\le \frac{2|\sqrt{\pi} - \sqrt{\widetilde{\pi}}|}{\sqrt{e^{-\lambda_{\max}}}}, \text{ and} |\log \frac{1 - \pi}{1 - \widetilde{\pi}}|\le \frac{2|\sqrt{1 - \pi} - \sqrt{1 - \widetilde{\pi}}|}{\sqrt{(1 -e^{-\lambda_{\min}})s_{N, K}}},
\end{align*}
leading to
\begin{align*}
(\log \frac{\pi}{\widetilde{\pi}})^2 +  (\log \frac{1 - \pi}{1 - \widetilde{\pi}})^2 & \le 4 e^{\lambda_{\max}} (\sqrt{\pi} - \sqrt{\widetilde{\pi}})^2  + \frac{4}{(1 - e^{-\lambda_{\min}})s_{N, K}} (\sqrt{1 - \pi} - \sqrt{1 - \widetilde{\pi}})^2\\
& \le 4(e^{\lambda_{\max}} + \frac{e^{\lambda_{\min}}}{(e^{\lambda_{\min}} - 1) s_{N, K}}) \min \left\{KL_{\text{Bernoulli}}(\pi||\widetilde{\pi}),KL_{\text{Bernoulli}}(\widetilde{\pi}||\pi)\right\}. 
\end{align*}
Additionally, by the mean value theorem, 
\begin{align*}
|\log \frac{e^\lambda - 1}{e^{\widetilde{\lambda}}-1}| \le \max_{\lambda_{\min} \le \lambda_0 \le \lambda_{\max}} \frac{e^{\lambda_0}}{e^{\lambda_0} -1} |\lambda - \widetilde{\lambda}|  = \frac{e^{\lambda_{\min}}}{e^{\lambda_{\min}} - 1}|\lambda - \widetilde{\lambda}|.
\end{align*}
It then follows from Lemma (\ref{lemma: KL_Hurdle_square_loss}) that 
\begin{align*}
(\log \frac{e^\lambda - 1}{e^{\widetilde{\lambda}}-1})^2&  \le \left(\frac{e^{\lambda_{\min}}}{e^{\lambda_{\min}} - 1}\right)^2 \frac{16(\lambda_{\max} +1)}{s_{N, K}(1  - e^{-\lambda_{\min}})} \min \left\{KL_{\text{Hurdle}}(\pi, \lambda|| \widetilde{\pi}, \widetilde{\lambda}), KL_{\text{Hurdle}}( \widetilde{\pi}, \widetilde{\lambda}|| \pi, \lambda)\right\}\\
& = \left(\frac{e^{\lambda_{\min}}}{e^{\lambda_{\min}} - 1}\right)^3 \frac{16(\lambda_{\max} +1)}{s_{N, K}} \min \left\{KL_{\text{Hurdle}}(\pi, \lambda|| \widetilde{\pi}, \widetilde{\lambda}), KL_{\text{Hurdle}}( \widetilde{\pi}, \widetilde{\lambda}|| \pi, \lambda)\right\},
\end{align*}
and 
\begin{align*}
 \left( |\log \frac{\pi}{\widetilde{\pi}}| +  |\log \frac{1 - \pi}{1 - \widetilde{\pi}}| + |\log \frac{e^\lambda - 1}{e^{\widetilde{\lambda}}-1}|\right)^2 \le& 12\left(e^{\lambda_{\max}} + \frac{e^{\lambda_{\min}}}{(e^{\lambda_{\min}} - 1) s_{N, K}} + \frac{4e^{3\lambda_{\min}}(\lambda_{\max} +1)}{(e^{\lambda_{\min}} - 1)^3s_{N, K}} \right) \\
 & \times \min \left\{KL_{\text{Hurdle}}(\pi, \lambda||\widetilde{\pi}, \widetilde{\lambda}),KL_{\text{Hurdle}}(\widetilde{\pi}, \widetilde{\lambda}||\pi, \lambda)\right\}, 
\end{align*}
Therefore, within the set $S_u$, we have
\begin{align*}
&\quad \frac{1}{\varphi_{N, K}} \sum_{i\le j, k\in [K]}\left(\log 
\frac{\pi^*_{i, j, k}}{\pi_{i, j, k}} + \log\frac{1 -\pi_{
i, j, k}}{1 - \pi^*_{i, j, k}}+\log\frac{e^{\lambda^*}_{i, j,k} - 1}{e^{\lambda_{i,j,k}} -1}\right)^2,\nonumber\\
& \le 12\left(e^{\lambda_{\max}} + \frac{e^{\lambda_{\min}}}{(e^{\lambda_{\min}} - 1) s_{N, K}} + \frac{4e^{3\lambda_{\min}}(\lambda_{\max} +1)}{(e^{\lambda_{\min}} - 1)^3s_{N, K}} \right) 2^{u+2}\epsilon_{N, K}.  
\end{align*}

Similarly, $ |\log \frac{\lambda}{\widetilde{\lambda}}|\le \frac{1}{\lambda_{\min}} |\lambda - \widetilde{\lambda}|$, leading to 
\begin{align*}
\frac{1}{\varphi_{N, K}} \sum_{i\le j, k\in [K]}\left(\log \frac{\lambda_{i, j,k}}{\lambda_{i,j, k}^*}\right)^2 \EE_P C_{i, j, k}^2 & \le \frac{\lambda_{\max}(1+ \lambda_{\max})}{\varphi_{N, K}} \sum_{i\le j, k\in [K]}\left(\log \frac{\lambda_{i, j,k}}{\lambda_{i,j, k}^*}\right)^2 \\
& \le \frac{16 \lambda_{\max}e^{\lambda_{\min}}(\lambda_{\max} +1)^2}{\lambda_{\min}^2 (e^{\lambda_{\min}} - 1)s_{N, K}} 2^{u+2}\epsilon_{N, K}.
\end{align*}
As such, within the set $S_u$, $\|\bF\|_{P, 2}$ can be upper bounded as
\begin{align*}
& \quad \|\bF\|_{P, 2} \\
& = \sqrt{\frac{1}{\varphi_{N, K}}\sum_{i\le j, k} \EE_P F_{i, j, k}(C_{i, j, k})^2}\\
&\le \sqrt{\frac{2}{\varphi_{N, K}} \sum_{i\le j, k\in [K]} \left[\left(\log 
\frac{\pi^*_{i, j, k}}{\pi_{i, j, k}} + \log\frac{1 -\pi_{
i, j, k}}{1 - \pi^*_{i, j, k}}+\log\frac{e^{\lambda^*}_{i, j,k} - 1}{e^{\lambda_{i,j,k}} -1}\right)^2 + \left(\log \frac{\lambda_{i, j,k}}{\lambda_{i,j, k}^*}\right)^2 \EE_P C_{i, j, k}^2 \right]}\\
& \le \sqrt{ 32 \left[e^{\lambda_{\max}} + \frac{e^{\lambda_{\min}}}{(e^{\lambda_{\min}} - 1) s_{N, K}} + \frac{4e^{3\lambda_{\min}}(\lambda_{\max} +1)}{(e^{\lambda_{\min}} - 1)^3s_{N, K}} + \frac{\lambda_{\max}e^{\lambda_{\min}}(\lambda_{\max} +1)^2}{\lambda_{\min}^2 (e^{\lambda_{\min}} - 1)s_{N, K}}\right]  2^{u+2}\epsilon_{N, K} }. 
\end{align*}

\textbf{Uniform law of large numbers}
Denote $m_3$, $m_4$ be the constants that 
\begin{align*}
 m_3 &= 2\sqrt{1-e^{-\lambda_{\min}}}\min\{1, \frac{1}{\lambda_{\max}(1 - e^{-\lambda_{\max}})}\} \text{ and }\\
 m_4 &= \sqrt{128\left[e^{\lambda_{\max}} + \frac{e^{\lambda_{\min}}}{(e^{\lambda_{\min}} - 1) } + \frac{4e^{3\lambda_{\min}}(\lambda_{\max} +1)}{(e^{\lambda_{\min}} - 1)^3} + \frac{\lambda_{\max}e^{\lambda_{\min}}(\lambda_{\max} +1)^2}{\lambda_{\min}^2 (e^{\lambda_{\min}} - 1)}\right] }. 
\end{align*}
We have $\|\bF\|_{Q, 2} \ge m_3 \sqrt{s_{N, K}}2^u \epsilon_{N,K}$ and $\|\bF\|_{P, 2} \le m_4 \sqrt{2^u \epsilon_{N, K}/s_{N, K}}$. Combining these with the previous result, we have
\begin{align*}
\quad & \EE\sup_{(\bGamma, \bZ, \bbeta, \bxi)  \in S_u}\frac{1}{\varphi_{N, K}} | \sum_{i\le  j, k} f_{i,j, k}(C_{i, j,k}) \varepsilon_{i, j, k}|\\
& \le  \frac{m_1  \sqrt{QL+KR+2RL} }{\varphi_{N, K}^{1/2} \sqrt{\log \frac{12m_2 \lambda_{\max} QL^{3/2} \sqrt{KR} \beta_{\max} \xi_{\max}}{\inf_Q\|\bF\|_{Q, 2}s_{N, K} \varphi^{1/2}_{N, K}}} } ||\bF||_{P, 2} \left(1 + \log \frac{12m_2 \lambda_{\max} QL^{3/2} \sqrt{KR}\beta_{\max} \xi_{\max}}{\inf_Q\|\bF\|_{Q, 2}s_{N, K} \varphi^{1/2}_{N, K}} \right).\\
& \le m_1 m_4  \sqrt{\frac{(QL+KR+2RL)2^u \epsilon_{N, K}} {\varphi_{N, K} s_{N, K}}} \frac{1 + \log \frac{12m_2 \lambda_{\max} QL^{3/2} \sqrt{KR}\beta_{\max} \xi_{\max}}{ m_3 2^{u} \epsilon_{N,K}s_{N, K}^{3/2} \varphi^{1/2}_{N, K}}}{\sqrt{\log \frac{12m_2 \lambda_{\max} QL^{3/2} \sqrt{KR}\beta_{\max} \xi_{\max}}{  m_3 2^{u} \epsilon_{N,K}s_{N, K}^{3/2} \varphi^{1/2}_{N, K}}} }.\\
\end{align*}
Take 
\begin{align*}
\epsilon_{N, K} = \frac{m_5 (QL + KR + 2RL) \log \varphi_{N, K} }{\varphi_{N, K}s_{N, K}}, 
\end{align*}
for some $m_5$ large enough such that 
\begin{align*}
m_1 m_4 \sqrt{\frac{ 2^{u}}{m_5 \log \varphi_{N, K}}}\frac{1 + \log \frac{12m_2 \lambda_{\max} \varphi_{N, K}^{1/2} QL^{3/2} \sqrt{KR} \beta_{\max} \xi_{\max}}{ m_3 2^{u} m_5 s_{N, K}^{1/2} (QL + KR + 2RL) \log \varphi_{N, K} }}{\sqrt{\log \frac{12m_2 \lambda_{\max} \varphi_{N, K}^{1/2}QL^{3/2} \sqrt{KR}\beta_{\max} \xi_{\max}}{m_3 2^{u} m_5 s_{N, K}^{1/2} (QL + KR + 2RL) \log \varphi_{N, K} }}} \le \frac{1}{2(1+\nu)},
\end{align*}
when $\varphi_{N, K}$ diverges. 
In this case, we have
\begin{align*}
\quad & \EE\sup_{(\bGamma, \bZ, \bbeta, \bxi)  \in S_u}\frac{1}{\varphi_{N, K}} | \sum_{i\le  j, k} f_{i,j, k}(C_{i, j,k}) \varepsilon_{i, j, k}|\le \frac{1}{2(1+\nu)} \epsilon_{N, K}.
\end{align*}
It then follows from Theorem 4 of \cite{Radoslaw2008tail} that, there exists some constant $m_6$ such that
\begin{align*}
I_u &\le  \PP\Bigg( \sup_{(\bGamma, \bZ, \bbeta, \bxi) \in S_u}
\frac{1}{\varphi_{N, K}} \sum_{i\le  j, k} Y_{i,j, k} - (1+ \nu) \EE\sup_{(\bGamma, \bZ, \bbeta, \bxi) \in S_u}
\frac{1}{\varphi_{N, K}} \sum_{i\le  j, k} Y_{i,j, k}\ge 2^u \epsilon_{N, K} -  \epsilon_{N, K}\Bigg)\\
 &\le  \PP\Bigg( \sup_{(\bGamma, \bZ, \bbeta, \bxi) \in S_u}
\frac{1}{\varphi_{N, K}} \sum_{i\le  j, k} Y_{i,j, k} - (1+ \nu) \EE\sup_{(\bGamma, \bZ, \bbeta, \bxi) \in S_u}
\frac{1}{\varphi_{N, K}} \sum_{i\le  j, k} Y_{i,j, k}\ge 2^{u-1} \epsilon_{N, K} \Bigg)\\
& \le \exp(- \frac{2^{2u - 2} \epsilon_{N, K}^2 \varphi_{N, K}^2}{4 m_4^2 \varphi_{N, K} 2^u \epsilon_{N, K}/s_{N, K}}) + 3 \exp\left( - \frac{2^{u-1}\epsilon_{N,K} \varphi_{N, K}}{m_6 \|\max_{i\le j, k} \sup_{S_u} Y_{i, j, k}\|_{\psi_1}}   \right).\\
& \le \exp(- \frac{2^{u - 4} \epsilon_{N, K} \varphi_{N, K} s_{N, K}}{m_4^2  }) + 3 \exp\left( - \frac{2^{u-1}\epsilon_{N,K} \varphi_{N, K}}{m_6 \|\max_{i\le j, k} \sup_{S_u} Y_{i, j, k}\|_{\psi_1}}   \right).
\end{align*}
Moreover, by the result of Proposition \ref{prop: psi_1},
\begin{align*}
\|\max_{i\le j, k} \sup_{S_u} Y_{i, j, k}\|_{\psi_1} &\le \max_{i\le j, k} \sup_{S} \|f_{i,j, k}(C_{i, j, k}) - \EE f_{i, j,k}(C_{i, j, k})\|_{\psi_1} \\
 &\le \max_{i\le j, k} \sup_{S} \|f_{i,j, k}(C_{i, j, k})\|_{\psi_1} 
+\max_{i\le j, k} \sup_{S}\|\EE f_{i, j,k}(C_{i, j, k})\|_{\psi_1} \\
 &\le 2 \max_{i\le j, k} \sup_{S} |a_{i,j, k}|/ \log 2 + \max_{i\le j, k} \sup_{S} (1 - p_{i, j, k})\lambda_{i, j, k}|b_{i,j, k}|/ \log 2\\
& \quad +\max_{i\le j, k} \sup_{S} |b_{i, j, k}|\left(\log\left( 1+ \frac{1}{\lambda_{i, j, k}}\log \frac{2-p_{i, j, k}}{1- p_{i,j,k}}\right)\right)^{-1}.
\end{align*}
Since $a_{i, j, k}$ and $b_{i, j, k}$ are continuous functions over the compact set $S$, there exists a constant $m_7$ such that 
\begin{align*}
\|\max_{i\le j, k} \sup_{S_u} Y_{i, j, k}\|_{\psi_1} \le m_7. 
\end{align*}
Notice that $\epsilon_{N, K}\varphi_{N, K}s_{N, }$ diverges, and thus there exists a constant $m_8$ such that 
\begin{align*}
& \quad \PP\left(\frac{1}{\varphi_{N, K}}\sum_{i\le  j, k} KL_{\text{Hurdle}}(\pi_{i, j, k}^*, \lambda_{i, j, k}^*|| \widehat{\pi}_{i, j, k}, \widehat{\lambda}_{i, j, k}) \ge 4\epsilon_{N, K} \right)\le m_8 \sum_{u=1}^{+\infty}  \exp(- \frac{2^{u} \epsilon_{N, K} \varphi_{N, K} s_{N, K}}{16 m_4^2  })\\
& \le m_8 \sum_{u=1}^{+\infty} \exp\left( - \frac{u \epsilon_{N, K} \varphi_{N, K} s_{N, K} }{16 m_4^2}\right) = \frac{m_8 \exp\left(-\epsilon_{N, K} \varphi_{N, K}s_{N, K}/16m_4^2\right)}{1 - \exp\left(-\epsilon_{N, K} \varphi_{N, K}s_{N, K}/16m_4^2\right)}. 
\end{align*}
Therefore, 
\begin{align*}
& \quad \PP\left(\frac{1}{\varphi_{N, K}}\sum_{i\le  j, k} KL_{\text{Hurdle}}(\pi_{i, j, k}^*, \lambda_{i, j, k}^*|| \widehat{\pi}_{i, j, k}, \widehat{\lambda}_{i, j, k}) \ge 4\epsilon_{N, K} \right) \\
&\le  \left[m_8 +\PP\left(\frac{1}{\varphi_{N, K}}\sum_{i\le  j, k} KL_{\text{Hurdle}}(\pi_{i, j, k}^*, \lambda_{i, j, k}^*|| \widehat{\pi}_{i, j, k}, \widehat{\lambda}_{i, j, k}) \ge 4\epsilon_{N, K} \right) \right] \exp\left(-\frac{\epsilon_{N, K} \varphi_{N, K} s_{N, K}}{16m_4^2}\right) \\
&\le (m_8 +1) \exp\left(-\frac{\epsilon_{N, K} \varphi_{N, K}s_{N, K}}{16m_4^2}\right) = M_2 \exp\left(-\epsilon_{N, K} \varphi_{N, K}s_{N, K}/16m_4^2\right), 
\end{align*}
where $M_2 = m_8 +1$. Recall that $\epsilon_{N, K} = \frac{m_5(QL + KR + 2RL) \log \varphi_{N, K}}{\varphi_{N, K}s_{N, K}}$. We adjust the constants by redefining $\epsilon_{N, K} = \frac{(QL + KR + 2RL) \log \varphi_{N, K}}{\varphi_{N, K}s_{N, K}}$, $M_1 = 4m_5$, and $M_3 = \frac{m_5}{16m_4^2}$, leading to 
\begin{align*}
 \PP\left(\frac{1}{\varphi_{N, K}}\sum_{i\le  j, k} KL_{\text{Hurdle}}(\pi_{i, j, k}^*, \lambda_{i, j, k}^*|| \widehat{\pi}_{i, j, k}, \widehat{\lambda}_{i, j, k}) \ge M_1\epsilon_{N, K} \right) \le  M_2 \exp\left(-M_3\epsilon_{N, K} \varphi_{N, K}s_{N, K}\right). 
\end{align*}

The desired result follows immediately from the fact that 
\begin{align*}
& \quad \PP\left(\frac{1}{\varphi_{N, K}}\sum_{i\le  j, k} KL_{\text{Hurdle}}(\pi_{i, j, k}^*, \lambda_{i, j, k}^*|| \widehat{\pi}_{i, j, k}, \widehat{\lambda}_{i, j, k}) \ge 4\epsilon_{N, K} \right) \\
& = \PP\left(\frac{1}{\varphi_{N, K}}\sum_{i\le  j, k} KL_{\text{ZIP}}(p_{i, j, k}^*, \lambda_{i, j, k}^*|| \widehat{p}_{i, j, k}, \widehat{\lambda}_{i, j, k}) \ge 4\epsilon_{N, K} \right). \qed
\end{align*}

\noindent \textbf{Proof of Theorem \ref{thm: consistency_F_norm}.} By the result of Lemma \ref{lemma: ZIP_KL_bound} and Theorem \ref{thm: large_deviation_inequality}, we have 
\begin{align*}
& \quad \frac{1}{\varphi_{N, K}}\left(\|\widehat{\bm{\mathcal{P}}} - \bm{\mathcal{P}}^*\|_F^2 + \|\widehat{\bm{\Lambda}} - \bm{\Lambda}^*\|_F^2\right)\\
&  \le \frac{2}{\varphi_{N, K}}\sum_{i\le j, k }\left[(\widehat{p}_{i, j, k} - p^*_{i, j, k})^2 + (\widehat{\lambda}_{i,j, k} - \lambda^*_{i, j,k})^2\right]\\
& \le \frac{8}{(1 - e^{-\lambda_{\min}})^{2}} \max \left\{ 1 - s_{N,K} +s_{N, K} e^{-\lambda_{\min}}, \frac{16(\lambda_{\max} + 1)}{s_{N, K}(1 - e^{-\lambda_{\min}})}   \right\}\cdot\\
& \quad \frac{1}{\varphi_{N, K}}\sum_{i\le j, k}KL_{\text{ZIP}}(p^*_{i, j, k}, \lambda^*_{i, j, k}|| \widehat{p}_{i, j, k}, \widehat{\lambda}_{i,j, k}) \\
&  \le \frac{8M_1}{(1 - e^{-\lambda_{\min}})^{2}} \max \left\{ 1 - s_{N,K} +s_{N, K} e^{-\lambda_{\min}}, \frac{16(\lambda_{\max} + 1)}{s_{N, K}(1 - e^{-\lambda_{\min}})}   \right\} \epsilon_{N, K}, 
\end{align*}
where the last inequality holds with probability at least $1 -M_2 \exp\left(- M_3\epsilon_{N, K} \varphi_{N, K}s_{N, K}\right)$.  \qed

\noindent \textbf{Proof of Theorem \ref{thm: calssification}.} Note that we can rewrite $\mathcal{R}(\Psi)$ as 
\begin{align*}
\mathcal{R}(\Psi)  = \EE\left[ \ind\left(\Psi(0) \ne Y\right)|C=0 \right]= \frac{ p(1 - e^{-\lambda}) \ind(\Psi(0) = 0) + e^{-\lambda} \ind(\Psi(0) = 1)} { p(1 - e^{-\lambda}) + e^{-\lambda}}. 
\end{align*}
When $p(1 - e^{-\lambda}) > e^{-\lambda}$, or simply $p> (e^{\lambda} - 1)^{-1}$, the minimizer of $\mathcal{R}(\Psi)$ is $\Psi^*(0) = 1$. Conversely, when $p\le (e^{\lambda} - 1)^{-1}$, the minimizer of $\mathcal{R}(\Psi)$ is $\Psi^*(0) = 0$. 

For any classifier $\Psi$, the excess risk is 
\begin{align*}
\mathcal{R}(\Psi) - \mathcal{R}(\Psi^*) &= \frac{ p(1 - e^{-\lambda}) \left[\ind(\Psi(0) = 0) - \ind(\Psi^*(0) = 0)\right] + e^{-\lambda}  \left[\ind(\Psi(0) = 1) - \ind(\Psi^*(0) = 1)\right]} { p(1 - e^{-\lambda}) + e^{-\lambda}}\\
&= \frac{ \left|p(1 - e^{-\lambda}) - e^{-\lambda} \right| \ind \left(\Psi(0) \ne \Psi^*(0)\right)} { p(1 - e^{-\lambda}) + e^{-\lambda}}, 
\end{align*}
where the last equality follows from the definition of $\Psi^*$, and can be verified for all possible values of $\Psi(0), \Psi^*(0) \in \{0, 1\}$. \qed

\putbib[ref]

\end{bibunit}

\end{document}